\documentclass{article}

\usepackage{arxiv}
\usepackage{xcolor}
\usepackage[utf8]{inputenc} 
\usepackage[T1]{fontenc}    
\usepackage{hyperref}       
\usepackage{url}            
\usepackage{booktabs}       
\usepackage{amsfonts}       
\usepackage{nicefrac}       
\usepackage{microtype}      
\usepackage{lipsum}		
\usepackage{enumerate}
\usepackage{graphicx}
\usepackage{natbib}
\usepackage{doi}
\usepackage{amsthm}
 \usepackage{amsmath}
 \usepackage{bbm}
\usepackage{amssymb}
\usepackage{eurosym}
\usepackage{dsfont}
\usepackage{float}
\usepackage[linesnumbered,boxed,ruled,vlined,resetcount]{algorithm2e}
\usepackage{authblk}
\theoremstyle{plain}
\newtheorem{theorem}{Theorem}
\newtheorem{proposition}{Proposition}

\newtheorem{lemma}{Lemma}

\newtheorem{remark}{Remark}

\newcommand{\NN}{\mathbb{N}}
\newcommand{\ZZ}{\mathbb{Z}}
\newcommand{\EE}{\mathbb{E}}
\newcommand{\bP}{\mathbb{P}}
\newcommand{\bR}{\mathbb{R}}

\newcommand{\cP}{\mathcal{P}}

\newcommand{\1}{\mathbb{I}}

\newcommand{\Xtilde}{\widetilde{X}}
\newcommand{\dpp}{\displaystyle}
\theoremstyle{remark}
\newtheorem{assumption}{Assumption}
\newtheorem{definition}{Definition}
\newtheorem{example}{Example}

\def\R{\mathbb{R}} 
\def\Z{\mathbb{Z}} 
\def\N{\mathbb{N}} 

\newcommand{\inr}[1]{\bigl< #1 \bigr>}

\DeclareMathOperator{\Ex}{\mathbb{E}} 
\renewcommand{\Pr}{\mathbb{P}}

\newcommand{\Xia}{\Xi_a}

\newcommand{\yinf}{\underline{y}}

\newcommand{\RR}{\mathbb{R}}
\newcommand{\lse}{\mbox{LSE}}

\newcommand{\ul}{\underline{\lambda}}
\newcommand{\ui}{\underline{i}}

\def\sE{\mathcal{E}}
\def\sG{\mathcal{G}}\def\sI{\mathcal{I}}

\def\sO{\mathcal{O}}
\def\sP{\mathcal{P}}

\newcommand{\GGmax}{\sG^{{\rm max}}}

\def\bi{{\bf i}}
\def\bj{{\bf j}}

\def\eps{\varepsilon}

\newcommand \cI{{\cal I}}
\newcommand \cG{{\cal G}}
\newcommand \cE{{\cal E}}
\newcommand \bE{{\mathbb E}}

\newcommand{\bxi}{{\boldsymbol{\xi}}}

\newcommand{\bbE}{{\pmb{\mathbb{E}}}}
\newcommand{\bbP}{{\pmb{\mathbb{P}}}}

\newcommand{\norm}[1]{\left\|#1\right\|}%
\newcommand{\cO}{\mathcal{O}}
\newcommand \bV{{\mathbb V}}
\newcommand \bN{{\mathbb N}}
\newcommand\ove[2]{\overset{#2}{#1}\;}
\newcommand \cN{{\cal N}}
\newcommand \cC{{\cal C}}
\newcommand \cS{{\cal S}}

\DeclareMathOperator*{\argmin}{arg\,min}

\begin{document}

\title{Statistical inference in large multi-way networks}





\author[$\dagger$]{Lucas Resende}
\author[$\S$]{Guillaume Lecué}
\author[$\dagger$]{Lionel Wilner}
\author[$\dagger$]{Philippe Chon\'e}


\affil[$\dagger$]{\textit{\footnotesize{CREST, ENSAE, Institut Polytechnique de Paris, France}}}
\affil[$\S$]{\textit{\footnotesize{ESSEC Business School,  France}}}


\maketitle

\begin{abstract}
We propose the \emph{Polyads estimator}, a new method to estimate structural parameters in weighted multi-way networks while controlling for rich, arbitrary structures of fixed effects.
The method is 
based on a series of classification tasks and is agnostic to both the number and structure of fixed effects.
Unlike full maximum likelihood, our estimator does not suffer from the incidental parameter problem: it is consistent and satisfies a Central Limit Theorem with no asymptotic bias, even when some dimensions of the network are short. 
For sparsely connected networks, it is also computationally faster than PPML. We provide experimental evidence that our estimator yields more reliable confidence intervals, i.e., better empirical coverage, than PPML and its bias-correction strategies. These improvements hold even under model misspecification and are more pronounced in sparse settings. While PPML remains competitive in dense, low-dimensional data, our approach offers a robust alternative for multi-way models that scales efficiently with sparsity. 
We apply the method to French health insurance claims data to study how a 2017 physician fee reform affected the geography and gender composition of doctor-patient connections.

\end{abstract}




\section{Introduction}
\label{sec:introduction}

Network data, which are becoming available at increasingly granular levels, are receiving a great deal of attention in economics, see~\cite{Grah_dePa_Book_2020}. In this paper, we consider ``multi-way'' networks that involve interactions between entities of different nature such as importing and exporting countries, buyers and suppliers, teachers and schools, doctors and patients. The strength of the interactions in such networks is commonly measured at disaggregated levels, e.g., industries or products for trade data, consultations and  medical procedures for health data, patents and citations for innovation data, etc. Multi-way network data, sometimes referred to as ``polyadic'', are indexed by multidimensional indices that represent the relevant dimensions in each case, for instance exporter, importer, product, and time in the trade example. 

To model connections in multi-way networks and control for unobserved heterogeneity along various dimensions, recent applied research has gradually considered models 
with richer structures of fixed effects, involving  higher-dimensional interactions. For instance, three-way gravity models, with 
exporter-year, importer-year, and exporter-importer fixed effects, are common in the modern trade literature.\footnote{Recent studies recognize the economic importance of the sector or product levels, potentially leading to even richer structures of fixed effects, \cite{breinlich2024trade} and \cite{Delb_Dina_WD_2020}.}
Yet as the structure of fixed effects becomes more complex, maximum likelihood estimators may be plagued by incidental parameter problems, see \cite{fernandez2016individual} and \cite{weidner2021bias}.  Specifically, as the sample size grows, so too does the number of nuisance parameters representing the fixed effects, possibly creating bias in standard maximum likelihood estimation of the parameters of interest, as first described by \cite{Neym_Scot_Ecma_1948}.

In this paper, we propose a novel estimator that does not suffer from the incidental parameter problem: our estimator is consistent and satisfies a Central Limit Theorem with no asymptotic bias. We provide experimental evidence that our estimator yields better empirical coverage than the standard PPML estimator and its analytical correction \cite{weidner2021bias}. These improvements hold even under model misspecification and are more pronounced in sparse settings. While PPML remains competitive in dense, low-dimensional data, our approach offers a robust alternative for multi-way models that scales efficiently with sparsity.


In the spirit of \cite{graham2017econometric}, we regard data through the lens of graph theory.
The key insight is that certain configurations of outcomes within subgroups of observations, which we call \emph{polyads}, have identical sufficient statistics for the fixed effects, making their relative likelihood independent of the fixed effects. 
Our framework differs from \cite{graham2017econometric} in two important dimensions. 
First, while Graham models undirected graphs, we use multipartite graphs to model multi-way networks. 
Second, Graham's network formation model considers only the extensive margin, i.e., the probability that potential links are realized. By contrast, we study the strength of connections in weighted networks, thus modeling both the intensive and extensive margins.

This study is connected to the strand of the gravity literature starting with \cite{silva2006log}.\footnote{Their seminal paper shows that traditional log-linear OLS estimation suffers from bias under heteroskedasticity, particularly when many flows are zero, which greatly motivated the adoption of Poisson models for gravity. The properties of Poisson pseudo-maximum likelihood estimators have first been investigated by \cite{gourieroux1984pseudo} in the absence of  fixed effects, with one-way fixed effect panel data applications being pioneered by \cite{HHG}. Recently, \cite{chen2024logs} argue that log-like transformations can also distort the interpretation of coefficients as percentage effects, since they depend on the units of the outcome.} 
Our method compares favorably with recent econometric studies  along several dimensions.
First, it accommodates multi-way models with an arbitrary number of node groups, 
in contrast to debiasing methods such as \cite{fernandez2016individual}, \cite{jochmans2017two}, and 
\cite{weidner2021bias}, which are restricted to two- and three-way structures.
Second, unlike PPML estimation, our estimator has no incidental parameter problem by construction. Even relative to bias-corrected PPML procedures, our approach remains advantageous: the corrections may themselves be biased in finite 
samples, as \cite{weidner2021bias} and \cite{zylkin2024bootstrap} argue. 
Third, compared with approaches such as \cite{charbonneau2012multiple}, our estimator better exploits the available variability while remaining computationally 
feasible, and the convexity of our loss function delivers strong numerical performance relative to more general semiparametric method-of-moments procedures (e.g., \cite{jochmans2017two, yang2023three}).

Our study is also connected to the literature on discrete choice models for panel and network data (\cite{Rasc_1960}, \cite{Ande_1973},\cite{Cham_Hand_1984},
\cite{Hono_Kyri_Ecma_2000},
\cite{Magn_Ecma_2004}). Presenting 
the conditional likelihood methods used in these settings, the recent review of \cite{Dano_Hono_Weid_WP_2025}
highlights how identification strategies relate to difference-in-differences approaches. Specifically, they  provide the differencing vectors that are valid to  identify the parameters of interest.\footnote{\cite{Muris_Pakel_WP_2025} follow this approach to study the formation of triadic networks. They introduce a hexad logit estimator that extends the tetrad logit estimator of \cite{graham2017econometric}.} 
We proceed the same way for count data and multi-way networks. The  Polyads estimator can be thought of as a nonlinear version of a difference-in-differences estimator. 
Contrary to the above cited literature, 
the polyad method handles count data and recovers both the existence and intensity of relationships. 

For researchers working with sparse networks --- i.e., where most potential connections are not realized ---, our approach offers a distinct computational advantage: polyads can be constructed by looping over pairs of edges with strictly positive counts (i.e., realized connections), allowing the estimation procedure to scale with the number of observed relationships rather than the number of potential relationships. This is a significant step forward, since the usage of tetrad-based methods has been limited by its computational cost: the available methods for computing tetrad-based statistics, which work only for two- or three-way models, either (i) require looping over all pairs of edges, including unrealized connections (e.g., \cite{graham2017econometric, Muris_Pakel_WP_2025}); \footnote{A slight modification of our method yields a computationally efficient estimator to handle the extensive-only case; see Appendix \ref{sec:binary} for details.} or (ii) rely on matrix multiplications that may scale worse in sparse networks.\footnote{See Remark \ref{rem:jochmans} on the computational complexity of \cite{jochmans2017two}.} In particular, as exemplified by our experiments, our computational implementation enables the use of the polyads method on large administrative datasets. See Section \ref{subsec:cost_considerations} for a more detailed discussion.

We establish consistency and asymptotic normality under mild assumptions, extending the current theoretical framework. 
We identify sufficient conditions on the geometry
of the graph that guarantee consistency and asymptotic normality.
As \cite{graham2017econometric, jochmans2018semiparametric}, we allow for unbounded fixed effects. 
We are able to obtain consistency and asymptotic normality without compactness assumptions on the structural parameter by modifying classical results from \cite{newey1994large} under the light of convex analysis tools from \cite{rockafellar-1970a} and asymptotic statistics results from \cite{Andersen1982}. Finally, contrary to \cite{graham2017econometric} and \cite{jochmans2018semiparametric}, our proofs do not require the existence of a limiting risk function. 


The practical limitations of our method merit clear statement. First, we do not consider interdependencies between observations beyond those captured by fixed effects and observed covariates. Second, our approach is designed specifically for count data, rather than continuous weights. Third, the method becomes computationally inefficient in dense networks where most potential relationships are realized. Within these constraints, however, our estimator provides a powerful tool for inference in multi-way networks with high-dimensional fixed effects structures of any type.

The remainder of the paper proceeds as follows. Section \ref{sec:model_assumptions} introduces the Poisson model and the structure of fixed effects. Section \ref{sec:beta_estimation} introduces the Polyads estimator. Section \ref{sec:theory} establishes its theoretical properties, demonstrating consistency and asymptotic normality. Section \ref{sec:computational} develops the computational implementation, emphasizing how the algorithm efficiently exploits sparsity. Section \ref{sec:experiments} documents the finite-sample properties of the estimator using artificial data and healthcare claims data. 
\section{Model assumptions}
\label{sec:model_assumptions}

We consider a count variable $Y_{i_1i_2\cdots i_D}\in\NN$ that is indexed  by
$(i_1, \dots, i_D)\in \sI =  [n_1] \times \cdots \times [n_D]$.
Using the $D$-dimensional index 
$\bi = (i_1, \dots, i_D) \in \sI$, we represent the variable as $Y_\bi$. 

Throughout the paper, we
think of $Y=(Y_\bi)_\bi\in \NN^\sI$ as a random $D$-partite graph, with the sets $[n_1],[n_2],\dots[n_D]$ representing the nodes of each category, the multidimensional index $\bi$ representing a potential (hyper-)edge of a the graph, and~$Y_{\bi}$ being the number of connections along edge~$\bi$, see the concrete examples below. 
We denote by~$E$ the set of positive edges, i.e., the set of $D$-dimensional indices~$\bi$ such that $Y_{\bi}>0$. The graph is sparse when the data contains many zeros, a case where our method delivers especially good results. 

\smallskip

For each $\bi$,  $Y_{\bi}$ is assumed to depend on a set of $p$ explanatory variables $X_\bi \in \R^p$ and a set of fixed effects. A level of fixed effect is  represented by a proper subset $g$ of $[D]$. By abuse of notation, we set
$g(\bi) = (i_d : d \in g)$ and the fixed effects for level $g$ are denoted as $\theta_{g(\bi)}^g \in \R$. The structure of the fixed effects in the model is represented by a collection $\sG$ of fixed effects levels. 
Of particular interest to us is the structure $\GGmax$ consisting of the $D$ subsets of $[D]$ of cardinal $D-1$; in this particular case, each level of fixed effect absorbs the variations of $Y_\bi$ in all but one dimension of $\bi$.
The set of all fixed effects is denoted by $\theta^\sG = \{\theta^g_\rho : g\in \sG, \rho\in g(\sI) \}$.

\smallskip


\begin{assumption}\label{assump:model}
    Let $\beta_\star \in \R^p$ be the parameter of interest (ie the structural parameter).
    The distribution of $Y = (Y_\bi \in \N : \bi\in\sI)$ conditionally on $X = (X_\bi \in \N : \bi\in\sI)$ is $\bP_{\beta_\star, \theta^\sG}^{Y|X}=  \bigotimes_{\bi \in \sI} \sP(\lambda_\bi)$ , where $\sP(\lambda_\bi)$ is the Poisson distribution with intensity  $\lambda_\bi>0$ given by
    \begin{equation}\label{eq:Poisson_model}
      \lambda_\bi =  \lambda_\bi(\beta_\star, \theta^\sG)  =   \exp\left(\beta_\star^\top  X_\bi + \sum_{g \in \sG} \theta_{g(\bi)}^g\right).
    \end{equation}
    It follows that the residual $\eps_{\bi}=Y_\bi-\lambda_\bi$ satisfies $\EE (\eps_\bi | X) = \EE (\eps_\bi |X_\bi)=0$. In other words, the explanatory variables $X_\bi$ are assumed to be strongly exogenous.
\end{assumption}

The log-likelihood of $(\beta, \theta^\cG)$ in the model \eqref{eq:Poisson_model} at the observed graph $y$ is
\begin{equation}
    \label{eq:loglikelihood}
    \ln \bP_{\beta,\theta^\sG}\left( Y = y | X \right) = \sum_{\bi \in \sI} - \lambda_\bi - \ln y_\bi! + y_\bi\left\{ \beta^\top  X_\bi + \sum_{g \in \sG} \theta_{g(\bi)}^g \right\}
\end{equation}
involves a potentially high number of fixed effects. The MLE  
estimator of the parameter of interest $\beta_\star$ has been shown to have an IPP for $D\geq 3$, see \cite{weidner2021bias} as well as the experiments presented in Section~\ref{sec:experiments}.



The following examples show how our framework encompasses gravity models and other classical econometric models.

\begin{example}[One-way model in panel data] Taking $D= 2$ and $\sG = \{ \{1\} \}$ yields the structure of 
the classical model studied by \cite{HHG}
\[ 
\ln \lambda_{i_1i_2} =  \beta_\star^\top  X_{i_1i_2} + \theta_{i_1}^1.
\]
\end{example}

\begin{example}[Two-way model]\label{ex:two_way} Taking $D = 2$ and $\sG = \GGmax = \{ \{1\}, \{2\} \}$ yields the structure of the standard gravity model studied by \cite{silva2006log} , i.e., $ 
\ln \lambda_{i_1i_2} = \beta_\star^\top  X_{i_1i_2} + \theta_{i_1}^1 + \theta_{i_2}^2
$. The usual econometric model 
\[
\ln \lambda_{ij} = \beta_\star^\top  X_{ij} + u_i + v_j
\]
obtains when relabeling the two-dimensional indices $(i_1,i_2)$ as $(i,j)$ and the fixed effects $\theta_{i_1}^1$ and $\theta_{i_2}^2$ as $u_i$ and~$v_j$ respectively.
In the trade literature, $i$ is an exporter, $j$ is an importer, $X_{i_1i_2}$ is a feature of the dyad (e.g., sharing borders or same language, having a free trade agreements in force). 
\end{example}

\begin{example}[Three-way model]
\label{ex:3:way}
Taking $D = 3$ and $\sG = \GGmax = \{ \{1,2\}, \{1,3\}, \{2,3\} \}$ yields the structure of the  model studied by \cite{weidner2021bias}, i.e., $ 
\ln \lambda_{i_1i_2i_3} = \beta_\star^\top  X_{i_1i_2 i_3} + \theta_{i_1i_2}^{1,2} + \theta_{i_1i_3}^{1,3} + \theta_{i_2i_3}^{2,3}$.
Relabeling the edges $(i_1,i_2,i_3)$ as $(i,j,t)$ and the fixed effects
$(\theta_{i_1i_2}^{1,2},\theta_{i_1i_3}^{1,3}, \theta_{i_2i_3}^{2,3})$ as $(u_{ij},v_{jt},w_{it})$, we obtain the usual econometric formulation
\[
\ln \lambda_{ijt} = \beta_\star^\top  X_{ijt} + u_{ij} + v_{jt}+w_{it}.
\]
This model is  used in the trade literature in the presence of a time dimension, where $i_1$ is an exporter, $i_2$ is an importer, and $i_3$ is the time. 
\end{example}


\section{Estimation of the homophily parameter $\beta_\star$  via polyads} \label{sec:beta_estimation}

To avoid the incidental parameter problem mentioned above, we first  condition the likelihood on a sufficient statistics for the fixed effects given by the degrees of the nodes, see Subsection~\ref{sec:degree:sufficient:stat}. This approach has been  followed in two-way contexts by \cite{charbonneau2012multiple}, \cite{graham2017econometric} and \cite{jochmans2018semiparametric}. 
We construct a loss function based on a set of `directions' that generate all the variability in the data for given degrees.

\medskip

\subsection{Generalized degrees and polyad transformations}
\label{sec:degree:sufficient:stat}

Consider a level of fixed effect $g\in\sG$. For $\rho\in g(\sI)$, a sufficient statistics  for the fixed effect $\theta_{\rho}^g$ appearing in the log-likelihood~\eqref{eq:loglikelihood} is
\[
\delta^g_\rho(y) = \sum_{\bi \in \sI : g(\bi) = \rho } y_\bi,
\]
which we call the degree of $\rho$ relative to the fixed effect level $g$ in the graph~$y$.
This quantity generalizes the notion of degree used in \cite{graham2017econometric}, capturing the total number of connections between edges $\bi$ for which $g(\bi) = \rho$.\footnote{Consider for instance Example~\ref{ex:3:way}, where the index $i_1,i_2,i_3$ are denoted $i,j,t$ as in many gravity models. The degree of (4,5) relative to the fixed effect level $u_{ij}$ is $\delta^{1,2}_{4,5}(y)=\sum_{t} y_{45t}$.} 
The family of degrees for the fixed effects level $g \in \sG$ is denoted by $\delta^g(y) = (\delta^g_\rho(y):\rho \in g(\sI))$. The family of all degrees is denoted by 
$\delta(y)=\left(\delta^g(y): g\in \sG \right)$.

As announced above, we condition the likelihood on the family of degrees $\delta(Y)$:
\begin{equation}
    \label{eq:cond:loglikelihood}
     \bP_{\beta,\theta^\sG}\left( Y = y \,|\, X,\delta(Y)=\delta(y) \right) = 
     \frac{\dpp \prod_\bi e^{-\lambda_\bi} \lambda_{\bi}^{y_\bi}/y_\bi!}{\dpp \sum_{z|\delta(z)=\delta(y)} \prod_\bi e^{-\lambda_\bi} \lambda_{\bi}^{z_\bi}/z_\bi!}.
 \end{equation}
Rewriting the conditional likelihood as
\begin{eqnarray}\label{eq:conditional_likelihood_degree}
     \bP_{\beta,\theta^\sG}\left( Y = y \,|\, X,\delta(Y)=\delta(y) \right) &=&
     \frac{\exp\left\{ \sum_\bi y_\bi\big(\sum_{g\in\sG} \theta^g_{g(\bi)}\big)+ y_\bi \beta^\top  X_\bi -\ln y_\bi!  \right\} }     {\sum_{z:\delta(z)=\delta(y)} \exp\left\{ \sum_\bi z_\bi\big(\sum_{g\in\sG} \theta^g_{g(\bi)}\big)+ z_\bi \beta^\top  X_\bi-\ln z_\bi!  \right\}}
     \nonumber\\
     &=&
     \frac{\exp\left\{ \sum_{g\in\sG,\rho\in g(\sI)} \delta^g_\rho(y)\theta^g_\rho+ \sum_\bi\big( y_\bi \beta^\top  X_\bi -\ln y_\bi!\big)  \right\} }     {\sum_{z:\delta(z)=\delta(y)} \exp\left\{ \sum_{g\in\sG,\rho\in g(\sI)} \delta^g_\rho(z)\theta^g_\rho+ \sum_\bi\big( z_\bi \beta^\top  X_\bi-\ln z_\bi!\big)  \right\}}
     \nonumber\\
     &=&
     \frac{\exp\left\{ \sum_\bi y_\bi \beta^\top  X_\bi -\ln y_\bi!  \right\} }     {\sum_{z:\delta(z)=\delta(y)} \exp\left\{ \sum_\bi z_\bi \beta^\top  X_\bi-\ln z_\bi!  \right\}}
\end{eqnarray}
shows that it does not depend on the fixed effects $\theta^\sG$ regardless of their structure $\sG$.
In the next subsection, we characterize the support of the distribution of the $Y$ conditional on all degrees $\delta(Y)$.

\medskip

\begin{definition}\label{def:polyads}
    A {\bf polyad} $\xi$ of $\cI=[n_1]\times \cdots\times [n_D]$ is a $2\times D$ matrix 
    \begin{equation*}
        \xi = \begin{pmatrix}
            j_1 & j_2 & \cdots & j_D\\ 
            j_1^\prime & j_2^\prime & \cdots & j_D^\prime\\ 
        \end{pmatrix}
    \end{equation*}where for all $d\in[D]$, $j_d\neq j_d^\prime\in[n_d]$. We denote by $\Xi$ the set of all polyads of $\cI$.  
    \end{definition}

An edge of a polyad $\xi$ is an index $\bi=(i_d)_d\in\cI$ such that $i_d\in\{j_d, j_d^\prime\}$ for all $d\in[D]$.
We denote by $\cE(\xi)$ the set of all edges of $\xi$. 
Any polyad $\xi$ has $|\cE(\xi)|= 2^D$ edges. In other words, a polyad $\xi$ induces a subgraph of $Y$ made of $2^D$ edges with weights $(y_{\bi})_{\bi\in \cE(\xi)}$.

Polyads as introduced in Definition~\ref{def:polyads} are generalizations to the $D$-dimensional framework of tetrads from \cite{charbonneau2012multiple, graham2017econometric, jochmans2018semiparametric} defined for $D=2$. The total number of polyads, i.e. the size of $\Xi$, is $\prod_{d=1}^D n_d(n_d-1)$. In the upper left corner of Figure \ref{fig:def_examples} we exemplify a polyad on $D=2$,

\begin{definition}
      Let $\xi\in\Xi$ be a polyad. Let $\bi\in\cI$. The {\bf sign of $\bi$ relative to $\xi$} is defined as 
    \begin{equation*}
        s_\xi(\bi) = \prod_{d=1}^D \left( \mathbf{1}\{i_d = j_d\} - \mathbf{1}\{i_d = j_d'\} \right).
    \end{equation*}
\end{definition}

In particular, $s_\xi(\bi) =0$ when $\bi$ is not an edge of $\xi$ and $s_\xi(\bi)\in\{\pm1\}$ otherwise. The main purpose of  the sign function $s_\xi$ is that it gives the signs of the diff-in-diff property stated below. The sign function is represented in the upper right corner of Figure \ref{fig:def_examples}, each edge has a sign associated to it, notice that the signs sum zero for all axis and that the sign of $\bj=(j_d)_{d\in [D]}$ - the first row vector of $\xi$ - is always 1.

We now define a class of transformations  indexed by polyads. These transformations act on $D$-partite graphs with integer edge weights, i.e., on the set $\ZZ^\sI$. Recall that we see the dependent variable $Y \in \NN^\sI$
as a graph with \emph{nonnegative} edge weights. This discrepancy plays an important role in the analysis developed below.   

\smallskip

\begin{definition} Let $\xi\in\Xi$ be a polyad. The polyad transformation $T_\xi: \ZZ^\cI\rightarrow \ZZ^\cI$ is defined by $T_\xi(y) = y + s_{\xi}$ where $s_\xi=(s_\xi(\bi):\bi\in\cI)$ i.e. $T_\xi(y)$ is a graph with weights given for all $\bi\in\cI$ by
\begin{equation}
\label{eq:diff:and:diff:operatir}
     T_\xi(y)_\bi = y_\bi +s_\xi(\bi).
\end{equation}For all $r\in\ZZ$, $T_\xi^r(y) = y+rs_\xi$.
\end{definition}

For any polyad $\xi$, the transformation $T_\xi$ alters only the subgraph of $y$ induced by the polyad $\xi$. In other words, $y'_\bi=y_\bi$ for all $\bi\notin \cE(\xi)$. The proposition below states the polyad transformations preserve degrees and, conversely, that they allow to generate all graphs sharing the same degrees as a given graph.\footnote{The converse result is stated in Proposition~\ref{prop:from:degrees:to:polyads} only for $\sG=\GGmax$. In the appendix, we consider any fixed effect structure $\sG$.}

\begin{proposition}[Characterization of degree-preserving transformations]\label{prop:from:degrees:to:polyads}
Consider any graph $y\in \ZZ^\cI$ and any polyad $\xi\in\Xi$. If  $y'=T_\xi(y)$, then we have
\begin{equation}
    \label{eq:degree:preservation}
    \delta(y')=\delta(y).
\end{equation}
Conversely, take two graphs $y$ and $y'$ in $\ZZ^\cI$ having the same degrees, i.e., such that $\delta(y')=\delta(y)$.
Suppose furthermore that the structure of fixed effect is $\GGmax$. Then there exists a finite sequence of integers $r_1,\dots,r_m\in\ZZ$ and finite sequence of polyads such that
\begin{equation}
    \label{eq:from:y:to:yprime}
    y'= T_{\xi_1}^{r_1} \circ \cdots \circ T_{\xi_m}^{r_m} (y).
\end{equation}
\end{proposition}

\begin{proof}
Pick a polyad $\xi$, a fixed effect level $g\in\sG$, $\rho\in g(\sI)$, and $y'=T_\xi(y)$. We observe  that the equality
\begin{equation}
\sum_{\bi\in\sI:g(\bi)=\rho} s_\xi(\bi) =0
\label{eq:sum:sign}    
\end{equation}
immediately implies

\[
\delta^g_{\rho}(y')= \sum_{\bi\in\sI:g(\bi)=\rho} y'_{\bi}=
\sum_{\bi\in\sI:g(\bi)=\rho} y_{\bi} + s_\xi(\bi) =
\sum_{\bi\in\sI:g(\bi)=\rho} y_{\bi}  = \delta^g_{\rho}(y),
\]
and hence the direct part of the proposition.
To prove~\eqref{eq:sum:sign}, we consider an edge $\bi$ such that
$g(\bi)=\rho$.
If $\bi$ is not an edge of the polyad-induced subgraph, i.e., if $\bi\notin\sE(\xi)$, the sign $s_\xi(\bi)=0$. Consider now the edges that belong to $\sE(\xi)$. There exist $2^{D-|g|}\geq 2$ edges $\bi\in \sE(\xi)$ such that $g(\bi)=\rho$, with half of them having  $s_\xi(\bi)=1$ and the other half having $s_\xi(\bi)=-1$, 
which yields~\eqref{eq:sum:sign}. The converse part is proved in the Appendix \ref{subsec:degree_preserving}. 
\end{proof}

According to Proposition~\ref{prop:from:degrees:to:polyads}, the conditioning set in the likelihood~\eqref{eq:cond:loglikelihood}
can be written as
\begin{equation}\label{eq:total_set_cond_degree}
\left\{\,Y \,|\, \delta(Y)= \delta(y)\,\right\}=
\left\{\, 
Y= T_{\xi_1}^{r_1} \circ \cdots \circ T_{\xi_m}^{r_m} (y)
:
m\geq 0, (\xi_1,\dots,\xi_m)\in\Xi^m, (r_1,\dots,r_m)\in\ZZ^m
\,\right\},
\end{equation}
which can be represented only at a prohibitively high computational cost. To bypass this computational barrier, we do not condition on $Y$ such that $\delta(Y) = \delta(y)$ but only on a subset of \eqref{eq:total_set_cond_degree}, given rise to the concepts of classes and orbits. 

\subsection{A classification problem on polyads and the associated estimator of $\beta_\star$}
\label{sub:a_classification_problem_on_polyad_and_the_associated_estimator_of_}

Rather than attempting to exploit the data variations within the whole set \eqref{eq:total_set_cond_degree}, we propose to exploit variations in the directions induced by each polyad separately, i.e., to restrict attention to sequences of polyads of length $m=1$ in \eqref{eq:total_set_cond_degree}. 
Because the count variable $Y_\bi$ takes nonnegative values, the transformations of the graph $T_\xi^r(Y)$
that do not belong to $\NN^\sI$ are irrelevant because they occur with zero probability. We thus further restrict the conditioning set.
For any polyad $\xi\in\Xi$, we introduce the nonnegative integers $m_\xi(y)$ and $M_\xi(y)$ given by
\begin{equation}
\label{eq:original-Mm}
m_\xi(y) = \bigwedge_{\bi : s_\xi(\bi) = 1} y_\bi  \ \ \ \text{ and }\ \ \   M_\xi(y) = \bigwedge_{\bi : s_\xi(\bi) = -1} y_\bi. \end{equation}
The range of integers $r$ such that $T_\xi^r(y)\in\NN^\sI$ is
$\{-m_\xi(y), \dots, M_\xi(y)\}$.
Accordingly, we define the orbit $\sO_\xi(y)$ of a graph $y\in\NN^\sI$ with respect to~$\xi$ as 
\begin{equation}
    \label{eq:def:orbit}
\sO_\xi(y) := \left\{ \, T_\xi^r(y):   -m_\xi(y) \leq r \leq M_\xi(y) \,\right\}.
\end{equation}
All graphs $y$ in the orbit $\sO_\xi(Y)$ of the observed graph have the same degrees as $Y$, $\delta(y)=\delta(Y)$, and coincide with~$Y$ outside the subgraph induced by $\xi$: $Y_\bi=y_\bi$ for $\bi\notin\sE(\xi)$.

\begin{figure}[ht]
    \centering
    \includegraphics[width=\linewidth]{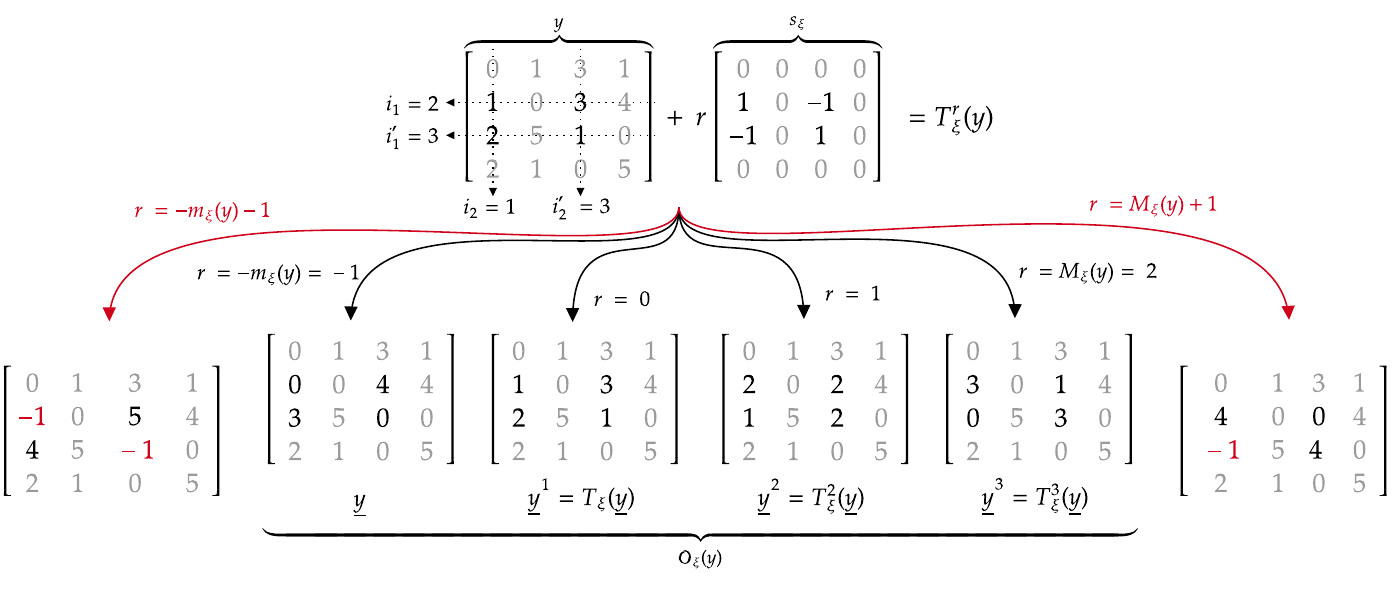}
    \caption{A numerical example illustrating the definitions introduced in this section and their interrelations is provided under the two-way setting (Example \ref{ex:two_way}). The figure depicts a polyad $\xi = \begin{pmatrix}
        2 & 1\\
        3 & 3
    \end{pmatrix}$, the sign $s_\xi(\bi)$ of an edge $\bi$ with respect to the polyad, the quantities $m_\xi(y)$ and $M_\xi(y)$, as well as the orbit $\sO_\xi(y)$. Notice that taking $r$ below $-m_\xi(y)$ or over $M_\xi(y)$ yields negative --- and therefore invalid --- edges.}
    \label{fig:def_examples}
\end{figure}

We introduce the loss function at the level of the polyad~$\xi$ as the opposite of the log-likelihood condition on $Y\in\sO_\xi(y)$:
\begin{equation}
\label{eq:elementary:polyad:loss}   
\beta \to \ell_\xi(y|X,\beta) = -\ln\bP_\beta \left(\,Y=y \,|\, X, Y\in \sO_\xi(y) \,\right).
\end{equation}
A polyad $\xi$ is not informative if the orbit $\sO_\xi(Y)$ is a singleton, i.e., if both $m_\xi(Y)$ and $M_\xi(Y)$ are zero. For non-informative polyads, we have $\ell_\xi(y|X,\beta) =0$ for all $\beta \in \R^p$. We can thus restrict attention to informative (or ``active'') polyads $\xi$ for which the  corresponding orbit $\sO_\xi(Y)$ has room for potential variation in the data, i.e., contains at least two distinct elements. Formally, active polyads satisfy   $|\sO_\xi(Y)| = m_\xi(Y)+M_\xi(Y)+1\geq 2$, such polyads play a key role in both theory and practice.  

To simplify notations, we denote the transformed graph~$T^r_\xi(y)$ by $y^r$, hence $y^r_\bi=y_\bi+r  s_\xi(\bi)$. The same computation as in~\eqref{eq:conditional_likelihood_degree} yields 
\begin{eqnarray*}
    \bP_\beta \left(\,Y=y \,|\, X, Y\in \sO_\xi(y) \,\right) &=& 
\frac{\dpp \exp\left\{ \sum_\bi y_\bi \beta^\top  X_\bi -\ln y_\bi!  \right\} }     {\dpp \sum_{r = -m_\xi(y)}^{M_\xi(y)}\exp\left\{ \sum_\bi y_\bi^r \beta^\top  X_\bi-\ln y_\bi^r!  \right\}} \label{eq:cond:proba:orbit:logit}\\
&=&
\left[\sum_{r = -m_\xi(y)}^{M_\xi(y)}\exp\left\{ \sum_\bi r s_\xi(\bi) \beta^\top  X_\bi + \ln \frac{y_\bi!}{y_\bi^r!}  \right\}
\right]^{-1}. \label{eq:cond:proba:orbit}
\end{eqnarray*}
As already explained, all graphs in the orbit $\sO_\xi(Y)$ coincide with $Y$ outside the subgraph induced by the polyad $\xi$. Formally, for $\bi\notin\sE(\xi)$, we have $s_\xi(\bi)=0$ and hence $y_\bi^r=y_\bi$ for all $r$ between $-m_\xi$ and $M_\xi$.  In the above sums over edges $\bi$, we can thus restrict attention to edges $\bi\in\sE(\xi)$. From~\eqref{eq:elementary:polyad:loss},
the loss function associated with the polyad $\xi$  is
\begin{equation}
    \label{eq:def_ell:diff:diff}
    \beta \to \ell_\xi(Y | X, \beta) = \ln\left( \sum_{r = -m_\xi(Y)}^{M_\xi(Y)} \exp\left\{ r \beta^\top  \Xtilde_\xi +\sum_{\bi \in\sE(\xi)} \ln \frac{Y_\bi!}{Y^r_\bi!}  \right\} \right),
\end{equation}
where the generalized ``difference-in-differences'' (DiD) operator $\Xtilde_\xi$ is defined as 
\begin{equation}
\label{eq:def:diff:and:diff}
    \widetilde{X}_\xi = \sum_{\bi \in \sI} s_\xi(\bi) X_\bi
    = \sum_{\bi \in \sE(\xi)} s_\xi(\bi) X_\bi.
\end{equation}  
In Example~\ref{ex:two_way}, consider the polyad 
\[
\xi=\begin{pmatrix}
            1 & 1 \\ 
            2 & 2 \\ 
        \end{pmatrix}.
\]
In this two-way example, we recover the usual DiD formula:
\[
\Xtilde_\xi= X_{22}- X_{21} -(X_{12} - X_{11})= X_{22}- X_{21} -X_{12} + X_{11}.
\]
In Example~\ref{ex:3:way}, consider the polyad 
\[
\xi=\begin{pmatrix}
            1 & 1 & 1\\ 
            2 & 2 & 2\\ 
        \end{pmatrix}.
\]
In this three-way example, we get (the opposite of) a ``triple difference'' formula:
\begin{eqnarray*}
-\Xtilde_\xi &=& [X_{222}- X_{221} -(X_{212} - X_{211})]-
[X_{122}- X_{121} -(X_{112} - X_{111})] \\
&=&
X_{222}- X_{221} -X_{212} + X_{211}-
X_{122}+ X_{121} +X_{112} - X_{111}.
\end{eqnarray*}

\begin{remark}[Difference-in-differences (DiD)] For a given polyad $\xi$, we may look at the sign  $s_\xi=(s_\xi(\bi))_\bi$ as a DiD operator acting on the observed graph $y$ and the observed feature vector $(X_\bi)_\bi$ either in an additive or a multiplicative way leading to the polyad transformation $T_\xi(y) = y+ s_\xi$ and the polyad feature $\Xtilde_\xi = \inr{s_\xi, (X_\bi)_\bi}$. The tensor $s_\xi=(s_\xi(\bi))_\bi$ provides the sign in a DiD approach.
\end{remark}

Finally, we call~$\Xia$ the set of informative (or active) polyads and form  the loss function taking $y = Y$ (the observed graph with its features $X$):
\begin{equation}\label{eq:def:loss:function}
\beta \to \widehat{L}_\Xi(Y|X,\beta) = \sum_{\xi\in\Xi} \ell_\xi(Y|X,\beta) = \sum_{\xi\in\Xia} \ell_\xi(Y|X,\beta).
\end{equation}
Two immediate observations will play an important role in our analysis. First, because the LogSumExp function is convex,\footnote{Recall that the LogSumExp function $\RR_{+*}^S\rightarrow\RR$\ is given by $\lse(U_0,\cdots,U_{S-1};S) = \ln\left(\sum_{s=0}^{S-1} e^{U_s}\right)$.} the function $\ell_\xi(y | X, \beta)$  is convex in $\beta$ for any polyad~$\xi$ and hence the loss function is convex. Second, given two active polyads $\xi$ and $\xi'$, the terms $\ell_\xi(Y|X,\beta)$ and $\ell_{\xi'}(Y|X,\beta)$ in the above sum are not independent if the two polyads share at least one edge, i.e., if $|\sE(\xi)\cap\sE(\xi')|\geq 1$.

We are now in a position to define our loss function and the associated estimator.

\begin{definition}\label{def:loss_fct_estimator}
Given the observed graph $Y$, our loss function is $\beta \to \widehat{L}_\Xi(Y | X, \beta)$ and the Polyads  estimator of the parameter of interest $\beta_\star$ is given by
\begin{equation}
    \label{eq:def_beta_hat}
    \widehat{\beta}_\Xi = \argmin_{\beta} \widehat{L}_\Xi(Y | X, \beta).
\end{equation}  
\end{definition}

The Polyads estimator relies on an objective function formed by summing over overlapping polyads, which intrinsically introduces statistical dependence. This sharply contrasts with the full Conditional MLE, which operates as a sum over independent edges. Furthermore, because we do not compute the full composition of polyad operators (as discussed in Proposition~\ref{prop:from:degrees:to:polyads}), our approach does not explore the entire combinatorial space of graphs sharing the observed degree sequence. Consequently, we naturally anticipate a loss of statistical efficiency. A fundamental question arises: exactly how much information is sacrificed in this process?\footnote{We wish to thank St{'e}phane Bonhomme for raising this question.} We formally assess this information loss in Appendix~\ref{sec:loss_info}, where we identify two structural parameters that drive the variance inflation. We also provide examples of both sparse and dense networks were this penalty is small.

The above Polyads estimator can be interpreted in relation to a logit classification problem. Specifically, given a polyad~$\xi$ and the observed orbit $\sO_\xi(Y)$, the minimal rank $m_\xi(Y)$ is distributed according to a conditional Logit model, see \cite{McFa_Chap_1974}. To see this, consider the transformed graph 
$\yinf=T^{-m_\xi(Y)}(Y)$, which can be thought of as the ``minimal'' graph in the orbit of the observed graph $y$. In the example of Figure~\ref{fig:def_examples}, the graph $\yinf$ is the second graph among  the six graphs shown on the bottom line (that is for $r=-m_\xi(y)=-1$). The observed orbit  $\sO_\xi(Y)$ can thus be represented as
\[
\sO_\xi(Y) := \left\{ \, T_\xi^m(\yinf):   0\leq m \leq |\sO_\xi(Y)|-1 \,\right\}. 
\]
If the polyad $\xi$ is active, the  size of the orbit $|\sO_\xi(Y)|= m_\xi(Y)+ M_\xi(Y)+1$ is greater than 2.
Changing the indices  $m=r+m_\xi(Y)$ in~\eqref{eq:cond:proba:orbit:logit} yields the conditional logit structure for the distribution of $m_\xi(Y)$
\begin{equation}
    \label{eq:logit}
\bP_\beta(m_\xi(Y)=m_\xi(y) \,|\, X, Y\in \sO_\xi(y))  = \frac{e^{v(m_\xi(y),\xi;\beta)}}{\sum_{m=0}^{|\sO_\xi(Y)|-1} e^{v(m,\xi;\beta)}}
\ \ \mbox{ with }\ \
v(m,\xi;\beta) =  m \beta^\top  \Xtilde_\xi -\sum_{\bi \in\sE(\xi)} \ln \yinf_\bi^m!
\end{equation}

\begin{remark}\label{rem:analogy_logit}(Analogy with conditional likelihood methods used in panel and network data)
For any active polyads $\xi\in\Xia$, the tensor $(s_\xi(\bi))_\bi$, whose entries belong to $\{-1,0,1\}$,
is a ``differencing vector'' in the sense of \cite{Dano_Hono_Weid_WP_2025}. 
It applies linearly to the features / coordinates of $X$ to deliver the generalized difference-in-differences term $\widetilde{X}_\xi$.
The corresponding transformation $T_\xi$ plays the same role for the count data $Y$. Conditionally on the observed graph $y$ belonging to the orbit $\sO_\xi(Y)$, we thus obtain a conditional logit classification problem, with the number of alternatives,
$|\sO_\xi(Y)|-1$, being polyad-specific.


\end{remark}

The conditional logit distribution of $m_\xi(Y)$ yields the following result.



\begin{lemma} \label{lem:derivatives}
For every polyad $\xi$ and $y=(y_\bi)_{\bi}$, the function $\beta\to \ell_\xi(y | X, \beta)$ has gradient 
\begin{equation}
    \label{eq:foc}
\nabla_\beta \ell_\xi(y | X, \beta)   =   \left(  \Ex_\beta\left[ 
m_\xi(Y) |X,  Y \in \sO_\xi(y) \right] - m_\xi(y) \right) \widetilde{X}_\xi  
\end{equation}
and its Hessian is non-negative since it is given by
\begin{equation}
    \label{eq:soc}
     \nabla_\beta^2 \ell_\xi(y| X, \beta)   = \mathbb{V}_\beta\left[ 
m_\xi(Y) |X,  Y \in \sO_\xi(y) \right] \widetilde{X}_\xi\widetilde{X}_\xi^\top  
\end{equation}
\end{lemma}

\begin{proof}
Denote by $p_\xi(m;\beta)$ the conditional probability that $m_\xi(Y)=m$ given by \eqref{eq:logit}. Notice that
\begin{eqnarray}
\nabla_\beta p_\xi(m;\beta) &=& m p_\xi(m;\beta) \Xtilde_\xi- \sum_{m'=0}^{|\sO_\xi(y)|-1} m'  p_\xi(m;\beta) p_\xi(m';\beta) \Xtilde_\xi \nonumber\\
&=&
p_\xi(m;\beta) \left\{ m - \Ex_\beta\left[ 
m_\xi(Y) |X,  Y \in \sO_\xi(y) \right] \right\} \Xtilde_\xi.   \label{eq:diff:logit}
\end{eqnarray}

Using $\ell_\xi(y | X, \beta)=-\ln p_\xi(m_\xi(y);\beta)$, we get
    \[
    \nabla_\beta \ell_\xi(y | X, \beta)= \Xtilde_\xi \sum_{m=0}^{|\sO_\xi(y)|-1} p_\xi(m; \beta) \,(m-m_\xi(y)) ,
    \]
which gives~\eqref{eq:foc}. Differentiating the above equality and using~\eqref{eq:diff:logit} yields
\[
 \nabla_\beta^2 \ell_\xi(y| X, \beta)  = \widetilde{X}_\xi\widetilde{X}_\xi^\top 
 \sum_{m=0}^{|\sO_\xi(y)|-1} p_\xi(m; \beta) \,(m-m_\xi(y)) 
 \left\{ m- \Ex_\beta\left[ 
m_\xi(Y) |X,  Y \in \sO_\xi(y) \right]\right\},
\]
and hence~\eqref{eq:soc}.
\end{proof}

As mentioned above, the loss function $\beta \to \widehat{L}_{\Xi}(Y | X, \beta)$ defined by~\eqref{eq:def:loss:function} is convex (in~$\beta$). 
Thanks to Lemma~\ref{lem:derivatives}, we can now compute its Hessian for $y=Y$ in
\[
\nabla^2_\beta \widehat{L}_\Xi(y | X, \beta) =  \sum_{\xi \in \Xia} \mathbb{V}_\beta\left[ 
    m_\xi(Y) | Y \in \sO_\xi(y) \right] \widetilde{X}_\xi\widetilde{X}_\xi^\top 
    \]
and discuss strict convexity.

\begin{lemma}
\label{lem:strictly_convex}
    The loss function $\beta\to \widehat{L}_\Xi(Y | X, \beta)$ defined by~\eqref{eq:def:loss:function} is strictly convex if and only if
$ \sum_{\xi \in \Xi_a} \widetilde{X}_\xi \widetilde{X}_\xi^\top  \succ 0$.
\end{lemma}

In other words, $\beta\to \widehat{L}_\Xi(Y|X,\beta)$ is strictly convex in $\beta$ if the DiD-features vary in the data.

\begin{proof}
Since $|\Xia| < \infty$ for any finite sample, $|\sO_\xi(y)| > 1$ for all $\xi\in\Xia$ and the shape of the distribution of $m_\xi(Y)$ in \eqref{eq:logit}, the minimum $\underline{\sigma}^2(\beta) := \min_{\xi\in\Xia} \mathbb{V}_\beta\left[ 
    m_\xi(Y) | Y \in \sO_\xi(y) \right]$ is strictly positive for all $\beta\in\R^p$. Therefore, for all $\beta\in\R^p$,
    \[ \nabla^2_\beta L(y | X, \beta) \succeq \underline{\sigma}^2(\beta) \sum_{\xi \in \Xi_a} \widetilde{X}_\xi\widetilde{X}_\xi^\top.\]  
\end{proof}

\section{Large sample properties of the Polyads estimator}
\label{sec:theory}

In this section, we establish consistency and asymptotic normality of the Polyads estimator introduced in Section \ref{sec:beta_estimation}. The concept of an active polyad is central to both our theoretical results and the implementation of the method (see Section \ref{sec:computational}). Recall that a polyad $\xi$ is active (for the observed graph $Y$) when its orbit satisfies $|\sO_\xi(Y)| > 1$, meaning it exhibits variation from which the parameters can be identified. This observation motivates normalizing the loss function by the expected number of active polyads and suggests deriving limiting behavior as the number of active polyads increases.

Let $\widehat{N}_a = |\Xi_a|$ denote the observed number of active polyads and let $N_a$ denote its conditional expectation given $X$ under Assumption \ref{assump:model}. Formally,
\begin{equation}\label{eq:normalization_active_polyads}
    \widehat{N}_a :=\sum_{\xi\in\Xi} \mathbb{I}\{\xi \text{ is active}\} \text{ and } N_a := \bE_{\beta_\star}[\widehat{N}_a|X] =  \sum_{\xi\in\Xi} \bP_{\beta_\star}(\xi \text{ is active}|X). 
\end{equation}As mentioned in Remark~\ref{rem:analogy_logit}, $N_a$ plays the role of the average number of data, we therefore normalize the loss function by $N_a$, defining
\begin{equation}
    \label{eq:def_risk}
    \widehat{Q}_\Xi(\beta) = \frac{1}{N_a}\widehat{L}_\Xi(Y|X, \beta) \text{ and } Q_\Xi(\beta) = \bE_{\beta_\star} \left[ \widehat{Q}_\Xi(\beta) \mid X \right].
\end{equation}In machine learning $\beta\to  Q_\Xi(\beta)$ is refered as the risk function.

Recall that $n = \prod_{d=1}^D n_d$ is the size of the graph indexed by $\sI = [n_1] \times [n_2] \times \dots \times [n_D]$ and that if $\Xi$ is the set of polyads given by the indices $\sI$, then $|\Xi|$ has order $n^2$. Through this section, one is given a sequence $\left(\Xi^{(1)}, \Xi^{(2)}, \dots\right)$ of families of polyads, defined over growing sets of indices. Let $n_d^{(k)}$ be the size of the $d$-th dimension associated with $\Xi^{(k)}$. We impose no restrictions on how each dimension $n_d^{(k)}$ grows, requiring only that $n^{(k)} = \prod_{d=1}^D n_d^{(k)} \to \infty$ --- or, equivalently, $|\Xi^{(k)}|\to\infty$. Notably, our results accommodate short panels where one dimension remains bounded while others diverge. For three-way models, this includes settings where the PPML estimator suffers from the incidental parameter problem \citep{weidner2021bias}. We also do not assume the existence of a limiting risk function $Q_\infty$, thus avoiding restrictive assumptions on the asymptotic behavior of fixed effects and covariates. In particular, we do not impose that fixed effects or covariates are drawn from any probability distribution. To circumvent a convoluted notation we use $(k)$ as index in place of $\Xi^{(k)}$, for instance, $\widehat{\beta}_{\Xi^{(k)}}$ is written as $\widehat{\beta}^{(k)}$. Also, notice that associated with each $(k)$ we also have covariates $X^{(k)}$, fixed effects $\theta^{(k)}$ and random variables $Y^{(k)}$.

Finally, unlike \cite{graham2017econometric,jochmans2018semiparametric}, we do not require the components of $\left[n_d^{(k)}\right]$ to arise from random sampling. This permits a more general interpretation of the data generating process. Consider a two-way model of doctor-patient interactions. Under the framework of \cite{graham2017econometric,jochmans2018semiparametric}, one assumes the existence of a large population graph containing all doctors and patients, from which patients and doctors are randomly sampled, inducing distributions on both fixed effects and covariates. Our approach is more constructive: for each $(k)$, the distribution of fixed effects and covariates may differ entirely. In the doctor-patient example (see Section~\ref{sec:experiments:real}), our framework accommodates data collection that expands geographically. For instance, initially observing patients and doctors in Paris only, then adding doctors from Marseille, then adding also patients from Marseille, and so forth. Imposing a random sampling structure would limit this possibility, as any finite sample could contain doctors from Marseille with positive probability.


\subsection{Consistency}

The Polyads estimator \eqref{eq:def_beta_hat} is defined as the minimizer of a loss function that is itself the sum of possibly non-independent terms. Our consistency result proceeds in two steps. First, we show that the normalized empirical loss $\widehat{Q}^{(k)}$ converges to its expectation $Q^{(k)}$ as $k\to\infty$. Second, we establish regularity of $Q^{(k)}$ around its minimum $\beta_\star$. Note that at no point do we require the existence of a limit risk function $\lim_k Q^{(k)}$.

Assumption \ref{ass:first_order_geometry} controls the approximation between the empirical and expected versions of $\widehat{Q}^{(k)}$. Observe that $\ell_\xi$ and $\ell_{\xi'}$ are independent as long as $\xi$ and $\xi'$ share no edges. Moreover, a polyad $\xi$ contributes to the loss only when it is active. Thus, Assumption \ref{ass:first_order_geometry} controls the amount of dependence across the terms of $\widehat{Q}^{(k)}$, enabling a law of large numbers for the normalized losses.

\begin{assumption}\label{ass:first_order_geometry}
We assume that, as $k \to \infty$, 
\begin{equation*}
    \Ex_{\beta_\star, \theta^{(k)}}\left[ \left. \sum_{\xi,\xi' \in \Xi^{(k)}} \mathbf{1}_{ \xi \text{ and }\xi' \text{ are active}} \right|X^{(k)}\right] \gg \Ex_{\beta_\star, \theta^{(k)}}\left[ \left. \sum_{\xi,\xi' \in \Xi^{(k)}} \mathbf{1}_{ \xi \text{ and }\xi' \text{ are active and }|\sE(\xi) \cap \sE(\xi')| \geq 1 } \right|X^{(k)} \right].
\end{equation*}
\end{assumption}

Assumption \ref{ass:first_order_geometry} is not restrictive: as $n^{(k)}$ grows, the number of positive entries $Y_\bi^{(k)}$ with disjoint indices should also increase, and therefore, the number of pairs of active polyads sharing no edge should outnumber those sharing at least one edge. Notice that if one assumes a distribution on the fixed effects and on the covariates this assumption is automatically true, as in this case the probability that $\xi$ and $\xi'$ are both active is always upper and lower bounded by a constant, thus it is just a matter of comparing the number of pairs of polyads -- which has order $(n^{(k)})^4$ -- and the number of pairs of polyads sharing an edge -- which has order $(n^{(k)})^3$. This last approach is used by \cite{graham2017econometric} and \cite{jochmans2018semiparametric}.

Next, Assumption \ref{ass:consistency} ensures enough curvature of $Q^{(k)}$ around the minimum $\beta_\star$ so that it can be well identified.

\begin{assumption}\label{ass:consistency}
 We assume that as $k \to \infty$, the Polyads estimator $\widehat{\beta}^{(k)}$ defined in \eqref{eq:def_beta_hat} is the unique minimizer of $\widehat{Q}^{(k)}(\beta)$. We assume that there exists $c_0>0$ such that for all $k$ large enough
 \begin{equation*}
     \frac{1}{N_a^{(k)}}\sum_{\xi\in\Xi^{(k)}} \bE_{\beta_\star}\left[\mathbb{V}_{\beta_\star}\left[ 
m_\xi(Y^\prime) |X,  Y^\prime \in \sO_\xi(Y) \right]|X \right] \widetilde{X}_\xi\widetilde{X}_\xi^\top \succeq c_0 I_p
 \end{equation*}where $Y, Y^\prime$ are iid distributed according to Assumption~\ref{assump:model}.
\end{assumption}

We are now in position to state the following consistency result:

\begin{theorem}[Consistency of the Polyads estimator]\label{theo:consistency_tetrads}Grant Assumptions \ref{assump:model}, \ref{ass:first_order_geometry} and \ref{ass:consistency}. We assume that there exists $c_0$ such that for all $k$, $\bi$ and $\xi$, $\norm{\widetilde X_\xi}_2\leq c_0$ and $\lambda_\bi(\beta_\star, \theta^\sG)\leq c_0$.  Then, for almost all $\left(X^{(k)}\right)_k$ and for all $\eps>0$, if $n^{(k)} \to \infty$ as $k\to\infty$, then
\[\bP_{\beta_\star, \theta^{(k)}}^{\left. Y^{(k)}\right|X^{(k)}}\left[ \norm{\widehat{\beta}^{(k)} - \beta_\star}_2\geq \eps\right]\to 0.\]   
\end{theorem}
\begin{proof}
    The proof can be found in Section~\ref{sub:consistency_of_}. The proof uses the convexity of the loss function as a key ingredient. It is based on a slight modification of Theorem 2.7 of \cite{newey1994large} that requires revisiting classical results in convex analysis (such as those from Chapter~10 in \cite{rockafellar-1970a}), as well as some results in asymptotic statistics from \cite{newey1994large} and \cite{Andersen1982}.
\end{proof}


\subsection{Asymptotic normality}

Before introducing the assumptions and main result on the asymptotic normality of the Polyads estimator we first investigate the error $\hat\beta^{(k)} - \beta_\star$. A Taylor approximation yields, for some $\overline{\beta}^{(k)}$ between $\beta_\star$ and $\widehat{\beta}^{(k)}$,
\[ \nabla \widehat{Q}^{(k)}\left(\beta_\star\right) - \nabla \widehat{Q}^{(k)}\left(\widehat{\beta}^{(k)}\right) = \nabla^2 \widehat{Q}^{(k)}\left(\overline{\beta}^{(k)}\right) \left(\beta_\star - \widehat{\beta}^{(k)}\right), \]
which implies
\begin{equation}\label{eq:behavior_err_beta}
    \widehat{\beta}^{(k)} - \beta_\star = - \left( \nabla^2 \widehat{Q}^{(k)}\left(\overline{\beta}^{(k)}\right) \right)^{-1} \nabla \widehat{Q}^{(k)}\left(\beta_\star\right).
\end{equation}

If $\widehat{\beta}^{(k)}$ is close to $\beta_\star$, $\widehat{Q}^{(k)}$ is a good approximation of $Q^{(k)}$ and $\nabla^2 Q^{(k)}(\beta_\star) \to \Gamma$ for some invertible $\Gamma$ then we may approximate
\[ \widehat{\beta}^{(k)} - \beta_\star \approx - \Gamma^{-1} \nabla \widehat{Q}^{(k)}\left(\beta_\star\right).\]
The last approximation shows that to control the fluctuations of $\widehat{\beta}^{(k)} - \beta_\star$ we need to control the fluctuations of $\nabla \widehat{Q}^{(k)}\left(\beta_\star\right)$, which is not a sum of independent random variables. Following \cite{graham2017econometric, jochmans2018semiparametric}, we use Hájek projections to approximate $\nabla \widehat{Q}^{(k)}\left(\beta_\star\right)$ by a sum of independent random variables for which classical central limit theorems hold. The following assumption is useful to control the error of this approximation:

\begin{assumption}\label{ass:second_order_geometry} Given $c \in \R^p$, define $w_\xi = \inr{\nabla \ell_{\xi}(\beta_\star), \Gamma^{-1} c}$.
We assume that, as $k \to \infty$, 
\begin{equation*}
    \Ex_{\beta_\star, \theta^{(k)}}\left[ \left. \sum_{\xi,\xi' \in \Xi^{(k)}} w_\xi w_{\xi'} \mathbf{1}_{|\sE(\xi) \cap \sE(\xi')| = 1 } \right|X^{(k)}\right] \gg \Ex_{\beta_\star, \theta^{(k)}}\left[ \left. \sum_{\xi,\xi' \in \Xi^{(k)}} |w_\xi w_{\xi'}| \mathbf{1}_{|\sE(\xi) \cap \sE(\xi')| \geq 2 } \right|X^{(k)}\right].
\end{equation*}

\end{assumption}


Assumption \ref{ass:second_order_geometry} bears resemblance to Assumption \ref{ass:first_order_geometry}. Assumption \ref{ass:first_order_geometry} asks for the active polyads to be distributed among the edges of $Y_\bi$ in a way that no edge $\bi$ has a substantial amount of active polyads $\xi$ such that $\bi \in \sE(\xi)$. Meanwhile and up to the weights $w_\xi$, which are zero when $\xi$ is not active, Assumption \ref{ass:second_order_geometry} looks at the cases where two polyads share one edge and requires that, up to some weights, in most of these cases only one edge is being shared.

Assumption \ref{ass:technical_CLT} below is a technical assumption to control the convergence of the Hessian:

\begin{assumption}\label{ass:technical_CLT}
    We assume that there exists $\Gamma$ invertible such that as $k\to +\infty$, $\nabla^2 Q^{(k)}(\beta_\star) \to \Gamma$.
\end{assumption}

In order to present the main result we introduce the quantity $\Sigma^{(k)}$, which is the covariance of the Hájek projection of $\nabla \widehat{Q}^{(k)}\left(\beta_\star\right)$ and is given by
\begin{equation}\label{eq:asymptotic_cov}
   \Sigma^{(k)} = \left(N_a^{(k)}\Gamma\right)^{-1} \left( \sum_{\bi\in\cI} \bE_{\beta_\star}\left[ \left.\bar{s}^{(k)}_{\bi} \bar{s}^{(k)\top}_{\bi} \right| X^{(k)}\right] \right) \left(N_a^{(k)}\Gamma\right)^{-1},
\end{equation}
where $N_a^{(k)}$ is the expected number of active polyads and
\[ \bar{s}_{\bi}^{(k)} = \sum_{\xi:\bi\in\cE(\xi)} \bE_{\beta^*}\left[\nabla \ell_\xi\left(Y^{(k)} | X^{(k)}, \beta_\star\right) \left| X^{(k)}, Y^{(k)}_\bi\right.\right]. \]

Under Assumption~\ref{ass:second_order_geometry}, $\Sigma^{(k)}$ will play the role of the asymptotic variance of the Polyads estimator if it does not vanish to zero as $k\to+\infty$ and if a third moment condition holds. We gather these last conditions in Assumption \ref{ass:third_order_and_non_zero_variance} below:

\begin{assumption}\label{ass:third_order_and_non_zero_variance}
Given $c \in \R^p$. We assume that
\begin{equation}\label{eq:asumption_CLT_T1_non_zero}
N_a^{-1}\Ex_{\beta_\star, \theta^{(k)}}\left[ \left. \sum_{\xi,\xi' \in \Xi^{(k)}} \mathbf{1}_{ \xi \text{ and }\xi' \text{ are active}} \right|X^{(k)}\right] \not\gg \Ex_{\beta_\star, \theta^{(k)}}\left[ \left. \sum_{\xi,\xi' \in \Xi^{(k)}} w_\xi w_{\xi'} \mathbf{1}_{|\sE(\xi) \cap \sE(\xi')| = 1 } \right|X^{(k)}\right].
\end{equation}

For $m\in\{2,3\}$, we define $\widehat{S}_m^{(k)} = \sum_{\bi\in\cI} |Z_\bi|^m$ where $Z_\bi = \inr{\bar s_\bi, \Gamma^{-1}c}$, $S_m^{(k)} = \bE_{\beta_\star}\left[\left.\widehat{S}_m^{(k)}\right|X^{(k)}\right]$ and $V_m^{(k)} = \sum_\bi \bV_{\beta_\star}\left[|Z_\bi|^m\left|X^{(k)}\right.\right]$. Assume that exists a sequence $\left(a^{(k)}\right)_k$ such that as $k \to \infty$, $a^{(k)}\to \infty$ and
\begin{equation}
\label{eq:small_ball_condition}
    S_3^{(k)}+ a^{(k)}\sqrt{V_3^{(k)}}  \ll \left( S_2^{(k)} - a^{(k)} \sqrt{V_2^{(k)}}\right)^{3/2}.
\end{equation}  
 \end{assumption} 

 Assumption \eqref{eq:small_ball_condition} is useful to bound the third moment of the Hájek projection of $\nabla \widehat{Q}^{(k)}\left(\beta_\star\right)$ in a central limit theorem and is not restrictive. It is satisfied as long as the quantities $\bar{s}_\bi$ (and their variances) are not concentrated on a limited number of edges $\bi$ as $k\to\infty$.

\begin{theorem}[Asymptotic normality of the Polyads estimator]\label{theo:CLT_polyads_estim_Graham}
Grant Assumptions \ref{assump:model} to \ref{ass:third_order_and_non_zero_variance} for some $c\in\R^p$, $c\neq0$. We assume that there exists $c_0$ such that for all $k$, $\bi$ and $\xi$, $\norm{\widetilde X_\xi}_2\leq c_0$ and $\lambda_\bi(\beta_\star, \theta^\sG)\leq c_0$. It holds that, for almost all $\left(X^{(k)}\right)_k$, conditionally on $\left(X^{(k)}\right)_k$,  if $n^{(k)} \to \infty$ as $k \to \infty$, then
\begin{equation*}
    \frac{\inr{ \widehat{\beta}^{(k)} - \beta_\star,c}}{\sqrt{c^\top \Sigma^{(k)}  c}}  \overset{d}{\to} \cN(0,1).
\end{equation*}
\end{theorem}
\begin{proof}
The proof of Theorem \ref{theo:CLT_polyads_estim_Graham}, given in Section \ref{sec:proof_of_CLT_polyads}, follows from a general lemma presented in Section \ref{sub:main_theorem_for_asymptotic_normality_under_convexity_assumption}, adapted from Brunel’s lecture notes \cite{victor}, which reference \cite{MR1026303} and \cite{MR1186263}. Unlike the classical setup, where the loss is a sum of independent terms, we handle dependent polyads and do not assume the existence of a limiting $Q_\infty$, only a limiting covariance matrix at $\beta_\star$. Hence, we extend the standard asymptotic normality proof under convexity to this dependent setting.
\end{proof}

 In Theorem \ref{theo:CLT_polyads_estim_Graham}, the asymptotic normality is written in terms of $\Sigma^{(k)}$, which cannot be directly evaluated from the sample. In Section \ref{subsec:computational_variance} we discuss two alternatives to approximate $\Sigma^{(k)}$ from the sample; in Section \ref{sec:experiments} we provide empirical validation that the confidence intervals obtained by each approach are accurate.
\section{Computational implementation}
\label{sec:computational}

Lemma \ref{lem:derivatives} provides the tools to solve the optimization problem \eqref{eq:def_beta_hat} defining the Polyads estimator $\widehat{\beta}_\Xi$. Since each term in the loss function is convex and we have expressions for both their gradients and their Hessians, we can efficiently solve the minimization problem using Newton's method. However, computing the gradient and Hessian involves summing over all polyads $\xi \in \Xi$, which naively requires looping over $\prod_{d=1}^D n_d (n_d - 1)$ terms, an operation that quickly becomes computationally expensive. The main goal of this section is to reduce this complexity by avoiding unnecessary iterations, which is done characterizing the set of active polyads $\Xia$. Besides that, we also discuss the implementation of two approximations of the variance. These approximations are essential to construct confidence intervals. The methods presented in this section are of special interest when $Y$ is sparse. In particular, the computational complexity of our methods outperforms the PPML alternative \cite{correia2019ppmlhdfe} when the size of $E = \{ \bi : Y_\bi > 0 \}$ is of order smaller than $\sqrt{n}$.

\begin{remark}
    The computational implementation provided in this section can be easily extended to the binary network formation case, providing a computational gain over the brute force implementations currently used in the literature. This extension is discussed in Appendix \ref{sec:binary} and is closely related to the methods discussed in \cite{graham2017econometric} and \cite{Muris_Pakel_WP_2025}.
\end{remark}

\subsection{Permutating polyads}

A key tool that we explore to obtain efficient computational implementations of the Polyads estimator is their invariance to permutations. A permutation of a polyad defined by $(\bi, \bi')$ is obtained by flipping some indices between $\bi$ and $\bi'$. For instance, 
\[
\xi'=\begin{pmatrix}
    1 & 2 & 2\\ 
    2 & 1 & 1\\ 
\end{pmatrix} \text{ is a permutation of }\xi=\begin{pmatrix}
    1 & 1 & 1\\ 
    2 & 2 & 2\\ 
\end{pmatrix} \text{ on indices }d\in\{2,3\}.
\]
We say that a permutation is odd when an odd number of indices are flipped and even when an even number of indices are flipped.
Notice that each polyad $\xi$ has a total of $2^D$ unique permutations, including itself. Lemma \ref{lem:tools_polyads} below collects the main tools that will be necessary in this section:

\begin{lemma}
    \label{lem:tools_polyads} It holds that:
    \begin{enumerate}[(i)]
        \item If $\bi \in \sE(\xi)$ for some $\xi$, then $\bi$ belongs to all $2^D$ permutations of $\xi$ and exists exactly one permutation that can be written as $(\bi, \bi')$ for some $\bi' \in \sI$.

        \item Let $\xi'$ be any permutation of $\xi$. Then $s_{\xi'}(\bi) = s_{\xi}(\bi), \forall \bi$ iff the permutation is even and $s_{\xi'}(\bi) = -s_{\xi}(\bi), \forall \bi$ iff the permutation is odd. In particular, for odd permutations $m_{\xi'}(y) = M_\xi(y)$ and  $M_{\xi'}(y) = m_\xi(y)$ for all $y \in \Z^\sI$ and for even permutations $m_{\xi'}(y) = m_\xi(y)$ and  $M_{\xi'}(y) = M_\xi(y)$ for all $y \in \Z^\sI$.

        \item If $\xi'$ is a permutation of $\xi$, then $\ell_\xi(y|X,\beta) = \ell_{\xi'}(y|X,\beta)$ for all $\beta \in \R^p$.
    \end{enumerate}
\end{lemma}
\begin{proof}
    To see (i), notice that flipping the $d$-th index of $\bi$ with the $d$-th index of $\bi'$ produces a new edge that still belongs to $\sE(\xi)$. Repeating this operation for all possible combinations of indices produces all $2^D$ permutations. Also, given $\bi\in\sE(\xi)$, to find the unique $\bi'$ such that $(\bi, \bi')$ is a permutation of $\xi$ one just needs to flip the $d$-th index of $\bi$ with the $d$-th index of $\bi'$ if and only if $i_d \neq i_d'$.

    To see (ii), notice that flipping one index changes the sign of $s_\xi(\bi)$, thus flipping an even number of indices preserves the sign while flipping an odd number of indices changes it. The expressions for $m_{\xi'}(y)$ and $M_{\xi'}(y)$ follow directly from the definition.

    Finally, (iii) follows observing that $\sO_\xi(y) = \sO_{\xi'}(y)$ and so
    \[ \Pr_\beta( m_\xi(Y) = m_\xi(y) | X, Y \in \sO_\xi(y) ) =
    \begin{cases}
        \Pr_\beta( m_{\xi'}(Y) = m_{\xi'}(y) | X, Y \in \sO_{\xi'}(y) )\text{, if the permutation is even}\\
        \Pr_\beta( M_{\xi'}(Y) = M_{\xi'}(y) | X, Y \in \sO_{\xi'}(y) )\text{, if the permutation is odd}
    \end{cases}. \]
    Since the event $M_{\xi'}(Y) = M_{\xi'}(y)$ is equivalent to $m_{\xi'}(Y) = m_{\xi'}(y)$ we are done.
\end{proof}

\subsection{Efficiently gathering all active polyads}

A direct conclusion of Lemma \ref{lem:tools_polyads} is that if $\xi \in \Xia$, then exists a permutation $\xi'$ of $\xi$ such that $i_d < i_d'$ for all $d = 2, \dots, D$ and $m_{\xi'}(y)>0$. If $M_{\xi'} = 0$ then this permutation is unique, but if $M_{\xi'} > 0$ then there are two such permutations, one with $i_1 < i_1'$ and another with $i_1 > i_1'$. Based on this observation we define the following set:
\[ \Xia^\star = \left\{ \xi = (\bi, \bi') : m_\xi(y) > 0 \text{ and }i_d < i_d' \,\forall d=2,\dots,D \text{ and } (M_\xi(y) = 0 \text{ or } i_1 < i_1') \right\}. \]

The set $\Xia^\star$ contains exactly one permutation of each active polyad $\xi \in \Xia$ and, by part (iii) of Lemma \ref{lem:tools_polyads},
\begin{equation}
    \label{eq:loss_xia_star}
    \widehat{L}_\Xi(y|X,\beta) = 2^D \sum_{\xi \in \Xia^\star } \ell_\xi(y|X,\beta) \text{ for all }\beta\in\R^p.
\end{equation}

Thus, to solve \eqref{eq:def:orbit} it suffices to look at all polyads in $\Xia^\star$. The definition of $\Xia^\star$ also leads to an efficient method to construct it. Notice that $m_\xi(y)$ it is positive if and only if $y_\bi > 0$ for all $\bi$ with $s_\xi(\bi) = 1$. In particular, we need at least $y_\bi$ to be positive to have $m_\xi(y)$ positive. This suggests looping over pairs $\bi,\bi' \in E$. The procedure to do it differs slightly depending on the parity of $D$. 

First, take $D=2$ and let $i_1 \neq i_1'$ be given. To have $(\bi, \bi') \in \Xia^\star$ we need to find $i_2 \neq i_2'$ such that $y_{i_1 i_2} \wedge y_{i_1' i_2'} > 0$.
Now let $D=3$ and $i_1 \neq i_1'$ be given.
We search for $i_2 \neq i_2'$ and $i_3 \neq i_3'$ satisfying $y_{i_1 i_2 i_3} \wedge y_{i_1 i_2' i_3'} \wedge y_{i_1' i_2 i_3'} \wedge y_{i_1' i_2' i_3} > 0$, in particular, $y_{i_1 i_2 i_3} \wedge y_{i_1 i_2' i_3'} > 0$.
More generally, given $i_1$ we let
\[E_{i_1} = \{ (j_2, \dots, j_D)  : y_{i_1j_2\dots j_D} > 0 \},\]
it holds that if $(\bi, \bi') \in \Xia^\star$, then
(i) $(i_2, \dots, i_D) \in E_{i_1}$ and $(i_2',\dots, i_D') \in E_{i_1'}$ when $D$ is even;
or (ii) $(i_2, \dots, i_D), (i_2', \dots, i_D') \in E_{i_1}$ when $D$ is odd.
Thus, one only needs to loop over the pairs $i_1 \neq i_1'$ and, for each of these pairs, over $((i_2, \dots, i_D), (i_2', \dots, i_D')) \in E_{i_1} \times E_{i_1'}$ (if $D$ is even) or $((i_2, \dots, i_D), (i_2', \dots, i_D')) \in  E_{i_1} \times E_{i_1}$ (if $D$ is odd).
Then, one simply verifies the remaining conditions for $(\bi, \bi') \in \Xia^\star$.
This procedure is summarized in Algorithm \ref{alg:construct_polyads}.

\begin{algorithm}\label{alg:construct_polyads}
\caption{Construct $\Xia^\star$}
\KwIn{$\{E_{i_1}\}_{i_1\in[n_1]}$}
\KwOut{$\Xia^\star$}

\For{$i_1 \in [n_1]$}{
    \For{$i_1' \in [n_1]$, $i_1' \neq i_1$}{
        \For{$(i_2, \dots, i_D) \in E_{i_1}$}{
            $E' \gets E_{i_1'}$\textbf{ if }$D \% 2 = 0$\textbf{ else }$E_{i_1}$\;
        
            \For{$(i_2', \dots, i_D')  \in E'$ \text{ such that } $i_d < i_d'$ for all $d=2,\dots,D$}{
                $\xi=((i_1,\cdots, i_D), (i_1', \cdots, i_D'))$\;

                \If{$m_\xi(y) > 0$ and ($M_\xi(y) = 0$ or $i_1 < i_1'$)}{
                    $\Xia^\star \gets \Xia^\star \cup \{ \xi \}$\;
                }
            }
        }
    }
}
\end{algorithm}

The next theorem establishes the computational complexity of constructing $\Xia^\star$ using Algorithm \ref{alg:construct_polyads}.

\begin{theorem}
    When $D$ is odd assume there exists $c \geq 1$ that $|E_{i_1}| < c\frac{|E|}{n_1}$ for all $i_d \in [n_d]$. Make no assumption if $D$ is even. The set $\Xia^\star$ can be computed in $O(|E|^2)$ using Algorithm \ref{alg:construct_polyads}.
\end{theorem}

\begin{proof}
    First, notice that checking for $m_\xi(y)>0$ and $M_\xi(y)=0$ requires checking the values of all $2^D$ edges in $\sE(\xi)$.
    Implementing the sets $E_{i_1}$ as hash tables allows us to check these values in constant time.
    Thus, we just need to count the number of times the innermost loop is executed.
    
    If $D$ is even, the innermost loop is executed $\sum_{i_1 \neq i_1'} |E_{i_1}||E_{i_1'}| = \left( \sum_{i_1} |E_{i_1}| \right)^2 - \sum_{i_1} |E_{i_1}|^2 \leq |E|^2$ times.
    If $D$ is odd, the innermost loop is executed $\sum_{i_1 \neq i_1'} |E_{i_1}|^2 \leq c \frac{|E|}{n_1} \sum_{i_1 \neq i_1'} |E_{i_1}| \leq c|E|^2$ times.
    Thus, in both cases the total complexity is $O(|E|^2)$.
\end{proof}

\begin{remark}
    In practice, one not only keep track of the polyads in $\Xia^\star$ but also of their corresponding edge values $\{ y_\bi : \bi\in\sE(\xi) \}$ and of $\widetilde{X}_\xi$. This precomputation allows us to avoid recomputing these quantities at each iteration of the optimization algorithm. 
    Besides that, implementing $E_{i_1}$ as an ordered list allows the usage of binary search, which although theoretically slower than a hash table, tends to be faster in practice.

    Another key point is that we do not require all features $X_\bi$ to be precomputed. All that suffices is a function that can map $\bi$ into $X_\bi$, this function will be called $2^D$ times for each active polyad to obtain $\widetilde{X}_\xi$. This is essential to get an efficient implementation of our method, otherwise the computational cost would be at least the cost of computing all features, which is $O(n)$.
\end{remark}

\subsection{Solving the optimization problem}

Once the set of polyads $\Xia^\star$ is computed, we minimize the loss \eqref{eq:loss_xia_star} using Newton's method.  
Lemma~\ref{lem:derivatives} provides closed-form expressions for the gradient and Hessian of each $\ell_\xi$, so a Newton step can be computed exactly.  
When the loss is strictly convex (see Lemma \ref{lem:strictly_convex}), Newton's method converges from any initial value $\beta^0$.  
Algorithm~\ref{alg:update_beta} displays one update step from $\beta^t$ to $\beta^{t+1}$ for $t\geq0$.

\begin{algorithm}\label{alg:update_beta}
\caption{Newton update for $\beta^t$}
\KwIn{$\Xia^\star, \beta^t$}
\KwOut{$\beta^{t+1}$}

$g \gets 0$\;
$H \gets 0$\;

\For{$\xi \in \Xia^\star$}{
    $(\mu, \sigma^2) \gets \texttt{EvaluateMoments}(\xi, \beta^t)$\;
    $g \gets g + (\mu - m_\xi(y)) \widetilde{X}_\xi$\;
    $H \gets H + \sigma^2 \widetilde{X}_\xi \widetilde{X}_\xi^\top$\;
}

$\beta^{t+1} \gets \beta^t - H^{-1} g$\;
\end{algorithm}

The function \texttt{EvaluateMoments} in Algorithm \ref{alg:update_beta} must return the expectation and variance of $m_\xi(Y)$ conditioned on $Y \in \sO_\xi(y)$ when $Y$ has law parametrized by $\beta = \beta^t$. Notice that from \eqref{eq:logit}, for all $\beta$,
\begin{align*}
    \bP_\beta\!\left(m_\xi(Y)=m \,\middle|\, X, Y\in \sO_\xi(y)\right)
    & = 
    \frac{\exp\bigl(v(m,\xi;\beta)\bigr)}
    {\sum_{m'=0}^{m_\xi(y)+M_\xi(y)} \exp\bigl(v(m',\xi;\beta)\bigr)}\\
    & = 
    \frac{\exp\bigl(v(m,\xi;\beta) - v(0,\xi;\beta)\bigr)}
    {\sum_{m'=0}^{m_\xi(y)+M_\xi(y)} \exp\bigl(v(m',\xi;\beta) - v(0,\xi;\beta)\bigr)},
\end{align*}
where
\[
v(m,\xi;\beta) 
= m \,\beta^\top \widetilde{X}_\xi
\;-\;
\sum_{\bi \in \sE(\xi)} \ln \bigl(\yinf_\bi^{\,m}! \bigr),
\]
and $\yinf_\bi^{\,m} = y_\bi+(m-m_\xi(y))s_\xi(\bi)$ . Directly evaluating $v(m,\xi;\beta)$ for each $m$ would require computing log factorials, we avoid this computation by noticing that
\[
v(m,\xi;\beta) - v(m-1,\xi;\beta)
= \beta^\top \widetilde{X}_\xi 
  + \sum_{\bi : s_\xi(\bi) = -1} \ln\!\bigl(y_\bi - (m-1-m_\xi(y))\bigr)
  - \sum_{\bi : s_\xi(\bi) = 1} \ln\!\bigl(y_\bi + (m-m_\xi(y))\bigr),
\]
so the values $v(m,\xi;\beta) - v(0, \xi; \beta)$ can be computed sequentially by cumulative summation without a log factorial.

\begin{algorithm}\label{alg:evaluate_moments}
\caption{Evaluate $\mathbb{E}_\beta\left[ 
m_\xi(Y) |X,  Y \in \sO_\xi(y) \right]$ and $\mathbb{V}_\beta\left[ 
m_\xi(Y) |X,  Y \in \sO_\xi(y) \right]$}
\KwIn{$\xi, \beta$}
\KwOut{$\mu, \sigma^2$}

$v_0 \gets 0$\;
$Z \gets 1$ \tcp*[r]{normalizing constant (unscaled)}
\For{$m = 1$ \KwTo $m_\xi(y)+M_\xi(y)$}{
    $v_m \gets v_{m-1} 
    + \beta^\top\widetilde{X}_\xi
    + \sum_{\bi : s_\xi(\bi) = -1} \ln\!\bigl(y_\bi - (m-1-m_\xi(y))\bigr)
    - \sum_{\bi : s_\xi(\bi) = 1} \ln\!\bigl(y_\bi + (m-m_\xi(y))\bigr)$\;
    $Z \gets Z + e^{v_m}$\;
}

\For{$m = 0$ \KwTo $m_\xi(y)+M_\xi(y)$}{
    $p_m \gets e^{v_m}/Z$\;
}

$\mu \gets \sum_{m} m\,p_m$\;
$\sigma^2 \gets \sum_m m^2 p_m - \mu^2$\;
\end{algorithm}

The procedure \texttt{EvaluateMoments}, given in Algorithm~\ref{alg:evaluate_moments}, computes the moments $\mu,\sigma^2$ required in Algorithm~\ref{alg:update_beta} using approximately $2^D |\sO_\xi(y)|$ operations.  
Since this cost scales linearly with the orbit size, evaluating all $Y \in \sO_\xi(y)$ may become prohibitive when the orbit is large.  
In practice, whenever $|\sO_\xi(y)|$ exceeds a predefined threshold $L$, we approximate the conditional distribution of $m_\xi(Y)$ by restricting the computation to the truncated set
\[
m \in 
\Bigl[
m_\xi(y) - \bigl( L/2 \wedge m_\xi(y) \bigr), \dots, m_\xi(y) + \bigl(L/2 \wedge M_\xi(y)\bigr)
\Bigr].
\]
This truncation has negligible numerical effects, because the distribution of $m_\xi(Y)$ is concentrated around $m_\xi(y)$, and extreme values contribute essentially nothing to the expectation or variance.

\begin{theorem}\label{thm:complexity}
    Algorithm \ref{alg:update_beta} runs with $O(\widehat{N}_a)$ operations.
\end{theorem}
\begin{proof}
    Follows from the discussion above and observing that Algorithm \ref{alg:update_beta} requires exactly one loop over all $\xi \in \Xia^\star$.
\end{proof}

\subsection{Evaluating the variance}
\label{subsec:computational_variance}

We now discuss two approaches for evaluating the covariance matrix of the Polyads estimator. To shorten notation let $\nabla \ell_\xi\left(Y | X, \widehat{\beta}_\Xi\right)$ be denoted by $\nabla \widehat{\ell}_\xi$ and define
\[ \widehat{\Gamma} = \nabla^2 L_\Xi\left(Y | X, \widehat{\beta}_\Xi\right). \]
In Theorem \ref{theo:CLT_polyads_estim_Graham}, the asymptotic normality is written in terms of $\Sigma^{(k)}$, which can not be directly evaluated from the sample. In practice, \eqref{eq:asymptotic_cov} suggest to approximate it by $\widehat{\Sigma}$ given by
\begin{equation*}
   \widehat{\Sigma} = \widehat{\Gamma}^{-1} \widehat{\Omega} \widehat{\Gamma}^{-1},\text{ where }\widehat{\Omega} = \sum_{\bi\in\cI} \left( \sum_{\xi:\bi\in\cE(\xi)} \nabla \widehat{\ell}_\xi \right) \left( \sum_{\xi:\bi\in\cE(\xi)} \nabla \widehat{\ell}_\xi \right)^\top.
\end{equation*}

We also implement and empirically verify the performance of another variance estimator. Equation \eqref{eq:behavior_err_beta} suggests approximating the covariance of $\left( \nabla^2 \widehat{Q}^{(k)}\left(\overline{\beta}^{(k)}\right) \right) \left( \widehat{\beta}^{(k)} - \beta_\star\right)$ by the expectation of
\[ \left( \nabla \widehat{Q}^{(k)}\left(\beta_\star\right)\right) \left( \nabla \widehat{Q}^{(k)}\left(\beta_\star\right)\right)^\top = \left( \widehat{N}_a^{(k)} \right)^{-2}\sum_{\xi, \xi' \in \Xi^{(k)}} \left( \nabla \ell_\xi\left(Y^{(k)} | X^{(k)}, \beta_\star\right) \right) \left( \nabla \ell_{\xi'}\left(Y^{(k)} | X^{(k)}, \beta_\star\right) \right)^\top.  \]
Notice that if $\xi$ and $\xi'$ share no edges, the expectation of their corresponding term is zero since it is the product of independent quantities with zero mean. This suggests approximating $\Sigma^{(k)}$ by
\begin{equation*}
   \widehat{\Sigma}' = \widehat{\Gamma}^{-1} \widehat{\Omega}' \widehat{\Gamma}^{-1},\text{ where }\widehat{\Omega}' = \sum_{\xi, \xi' \text{ sharing edges} } \left( \nabla \widehat{\ell}_\xi \right) \left( \nabla \widehat{\ell}_{\xi'} \right)^\top.
\end{equation*}

The difference between $\Omega$ and $\Omega'$ is subtle. The next lemma illuminates this difference and provides computationally tractable expressions for $\Omega$ and $\Omega'$.

\begin{lemma}
    Let
    \[ \sI_a = \{\bi : \exists \xi \in \Xia \text{ such that } \bi \in \sE(\xi) \} \]
    be the set of all edges that belong to at least one active polyad. It holds that
    \begin{align}
        \label{eq:formula_omega}
        \widehat{\Omega} &= \sum_{\bi\in\cI_a} \left( \sum_{\bi' : (\bi, \bi') \in \Xia} 2^D \nabla \widehat{\ell}_{(\bi, \bi')} \right) \left( \sum_{\bi' : (\bi, \bi') \in \Xia} 2^D\nabla \widehat{\ell}_{(\bi, \bi')} \right)^\top\\
        & = \sum_{\bi\in\cI_a} \sum_{\bi' : (\bi, \bi') \in \Xia} \sum_{\bi'' : (\bi, \bi'') \in \Xia} \left( \nabla \widehat{\ell}_{(\bi, \bi')}\right) \left( \nabla \widehat{\ell}_{(\bi, \bi'')}\right)^\top 2^{2D} \nonumber
    \end{align}
    and
    \begin{equation}
        \label{eq:formula_omega_prime}
        \widehat{\Omega}' = \sum_{\bi\in\cI_a} \sum_{\bi' : (\bi, \bi') \in \Xia} \sum_{\bi'' : (\bi, \bi'') \in \Xia} \left( \nabla \widehat{\ell}_{(\bi, \bi')}\right) \left( \nabla \widehat{\ell}_{(\bi, \bi'')}\right)^\top 2^{D + \sum_{d=1}^D \mathbf{1}\{ i_d' \neq i_d''\}}.
    \end{equation}
\end{lemma}
\begin{proof}
To see \eqref{eq:formula_omega} we start noticing that if $\bi \not \in \sI_a$, then its corresponding term is zero. Thus, the first sum can be taken only over $\bi \in \sI_a$. Now recall from Lemma \ref{lem:tools_polyads} that if $\bi \in \sE(\xi)$ for some $\xi$, then $\bi$ belongs to all $2^D$ permutations of $\xi$ and there exists exactly one permutation that can be written as $(\bi, \bi')$ for some $\bi'$. Since each $\ell_\xi$ is invariant by permutation we have
\[ \sum_{\xi:\bi\in\cE(\xi)} \nabla \widehat{\ell}_\xi = \sum_{\bi' : (\bi, \bi') \in \Xia} 2^D \nabla \widehat{\ell}_{(\bi, \bi')} \]
and so \eqref{eq:formula_omega} is proved.

The proof of \eqref{eq:formula_omega_prime} is more intricate. Notice that if $\xi, \xi'$ share one edge we can put this edge ``in evidence'' to obtain a permutation $(\bi, \bi')$ of $\xi$ and a permutation $(\bi, \bi')$ of $\xi'$. If $\xi, \xi'$ share exactly one edge this representation is unique and a total of $2^D \times 2^D$ permutations ($2^D$ for each polyad in the pair) will be represented by the pair $(\bi, \bi'), (\bi,\bi'')$. Now assume that the polyads share exactly $m$ edges, in this case the representation is not unique anymore, as there are $m$ possible choices of $\bi$. To avoid double counting we split the $2^D \times 2^D$ permutations evenly between all possible choices of $\bi$, making each one account for $\frac{2^D \times 2^D}{m}$ permutations. It remains to understand, for a given pair $(\bi, \bi'), (\bi, \bi'')$, how many edges they share. Notice that if $i_d' = i_d''$, then we can simultaneously flip the $d$-th coordinate of $\bi$ with the $d$-th coordinate of $\bi'$ and $\bi''$ to obtain a new shared edge. It can not be done if $i_d' \neq i_d''$. Thus, the number of shared edges is $m = 2^{\sum_{d=1}^D \mathbf{1}\{ i_d' = i_d''\}}$, which yields
\[ \frac{2^D \times 2^D}{ 2^{\sum_{d=1}^D \mathbf{1}\{ i_d' = i_d''\}} } = 2^{D + \sum_{d=1}^D \mathbf{1}\{ i_d' \neq i_d''\}} \]
permutations counted for each pair $(\bi, \bi'), (\bi, \bi'')$.\end{proof}

Indeed, \eqref{eq:formula_omega} and \eqref{eq:formula_omega_prime} make explicit the two main differences between 
$\widehat{\Omega}$ and $\widehat{\Omega}'$. First, $\widehat{\Omega}$ contains duplicates of certain pairs of polyads. A closer inspection of the proof reveals that these duplicates arise precisely on pairs that share strictly more than one edge. This clarifies the role of Assumption~\ref{ass:second_order_geometry} in Theorem~\ref{theo:CLT_polyads_estim_Graham}: for the projection strategy to be valid, the covariance of the Hájek projection of 
$\nabla \widehat{Q}^{(k)}(\beta_\star)$ --- which is approximately $\widehat{N}_a^{-2}\mathbb{E}\,\widehat{\Omega}$ --- and the true covariance --- approximately $\widehat{N}_a^{-2}\mathbb{E}\,\widehat{\Omega}'$ --- must converge to each other. Second, computing $\widehat{\Omega}$ is less costly than computing $\widehat{\Omega}'$. 
For each $\mathbf{i} \in \mathcal{I}_a$, evaluating $\widehat{\Omega}$ requires only a single pass over each $\mathbf{i}'$ such that $(\mathbf{i}, \mathbf{i}') \in\Xia$. In contrast, computing $\widehat{\Omega}'$ requires an extra loop over $\bi''$ such that $(\mathbf{i}, \mathbf{i}'') \in\Xia$.

We now use equations \eqref{eq:formula_omega} and \eqref{eq:formula_omega_prime} to obtain an algorithm for computing $\widehat{\Sigma}$ and $\widehat{\Sigma}'$. We need to be able to loop over all $\bi \in \sI_a$ and, given $\bi$, to efficiently loop over all $\bi'$ such that $(\bi, \bi') \in \Xia$. Recall that $\Xia^\star$ contains exactly one permutation of each active polyad $\xi \in \Xia$, in fact, $\widehat{N}_a = |\Xia| = 2^D|\Xia^\star|$. One can loop over each $\xi \in \Xia^\star$ and compute all permutations of $\xi$. By updating a dictionary containing for each key $\bi$ the corresponding set of $\bi'$s we can easily construct the data structure needed to evaluate the covariances. In practice we also store a pointer to the original $\xi$ so that we can profit from the already evaluated $\widehat{X}_\xi$ and $\{Y_\bi : \bi \in \sE(\xi)\}$. The following result gives the computational complexity of computing each variance alternative.

\begin{theorem}
    Computing $\widehat{\Sigma}$ requires $O(\widehat{N}_a)$ operations and computing $\widehat{\Sigma}'$ requires $O(\widehat{N}_a m_a )$ operations, where $m_a$ is the maximum over all edges $\bi$ of the number of active polyads $\xi$ such that $\bi \in \sE(\xi)$.
\end{theorem}
\begin{proof}
    Since the cost of updating and consulting a dictionary is constant one can construct the dictionary that maps $\bi$ to $\bi'$ in $O(\widehat{N}_a)$. For evaluating $\widehat{\Omega}$ one goes through each key $\bi$ and each set of $\bi'$s once, thus yielding $O(\widehat{N}_a)$. Evaluating $\widehat{\Omega}$ requires a double loop over the $\bi'$s, thus $O(\widehat{N}_a m_a)$. The corresponding $\widehat{\Gamma}$ is just the Hessian of the loss, which is evaluated in $O(\widehat{N}_a)$. 
\end{proof}

\subsection{Final cost analysis and practical considerations}\label{subsec:cost_considerations}

As a consequence of this section's discussion, we can provide a complete computational cost analysis of our method:

\begin{theorem} Assume that $m_a$ is bounded. When $D$ is odd, assume also that there exists $c \geq 1$ that $|E_{i_1}| < c\frac{|E|}{n_1}$ for all $i_d \in [n_d]$. Thus, evaluating $\widehat{\beta}_\Xi$ and estimating its variance requires $O(T|E|^2)$ operations, where $T$ is the number of iterations of Newton's method.
\end{theorem}
\begin{proof}
    First notice that constructing $\Xia^\star$ has cost $O(|E|^2)$, thus the size of $\Xia$, i.e. $\widehat N_a$, must be of order at most $|E|^2$. Each Newton's method update has cost $O(\widehat{N}_a)$ and, under the assumption that $m_a$ is bounded, both variance estimates have cost $O(\widehat{N}_a)$. Thus, the total cost is driven by the number of iterations of Newton's method times $|E|^2$.
\end{proof}

In practice, Newton's method converges in less than $10$ iterations, yielding $O(|E|^2)$ operations. This quantity is to be compared with the fast implementation of PPML from \cite{correia2019ppmlhdfe}, which is $O(n)$. Our analysis suggests that our method is faster when $|E| \ll \sqrt{n}$ and competitive when $|E|$ is of order $\sqrt{n}$. Although being the standard practice when reporting the complexity of algorithms, the big-O notation hides a constant that matters to practitioners. In the next section, we empirically verify that, as predicted by our cost analysis, our method outperforms PPML in terms of computational time when $|E|$ is smaller than $\sqrt{n}$ and remains competitive as $|E|$ grows. Indeed, in our computational setup (see Section \ref{sec:experiments}) the running time of our method is shorter than that of PPML as long as $|E| \leq 15 \sqrt{n}$. This is, for example, the case of a bipartite network with $n_1=n_2$ such that the average degree of a node is at most $15$.

\begin{remark}\label{rem:jochmans}
We also note that the improvements developed in this section, and particularly the construction of $\Xia^\star$, may be useful for other polyads-based methods. For example, \cite{jochmans2017two} proposes a generalized method of moments for two-way models with $n_1 = n_2 = \sqrt{n}$. His tetrad-based estimator leverages matrix-multiplication tricks and achieves a computational complexity of order $O(n^{1.1877})$ when using the best available theoretical matrix-multiplication algorithm. In contrast, our $O(|E|^2)$ complexity can yield substantial gains in sparse regimes.
\end{remark}

\section{Experiments}
\label{sec:experiments}

We provide experiments comparing our method with PPML and with the analytical debias proposed by \cite{zylkin2024bootstrap}. We consider artificial and real data. With artificial data we investigate the impact of the incidental parameter problem while knowing the correct value of $\beta_\star$. With real data we display evidence of the incidental parameter bias and show how it may lead towards wrong conclusions in inference. We also discuss the computational time and the effect of sparsity for all methods.

Our Polyads estimator is implemented as described in Section~\ref{sec:experiments}, and confidence intervals use the covariance approximation $\widehat{\Sigma}'$ from Section~\ref{subsec:computational_variance}. The experiments with artificial data were performed on a Mac Mini M4 with 16~GB of RAM with $10$ executions in parallel. The experiments with real data took place in a controlled environment for sensitve data access equipped with an Intel Xeon Gold~6444Y and 3~TB of RAM. 

\subsection{Artificial data}
\label{sec:experiments:artificial}

We conduct computational experiments using a three-way data-generating process inspired by \cite{weidner2021bias}; see Example~\ref{ex:3:way}.  
The dimensions are $n_1 = n_2$ (varied) and $n_3 = 5$ (fixed).  
The fixed effects $u_{ij}, w_{it}, v_{jt}$ are i.i.d.\ $\mathcal{N}(0,1/16)$, and $\beta_\star = 1$.  
The covariates $X_{ijt}$ are correlated with the fixed effects and, along the third axis, with their own past values:
\[
X_{ijt} =
\begin{cases}
\tfrac{1}{2} X_{ij(t-1)} + w_{it} + v_{jt} + \tfrac{1}{4}\mathcal{N}(0,1), & t>1,\\[3pt]
w_{it} + v_{jt} + \tfrac{1}{4}\mathcal{N}(0,1), & t=1.
\end{cases}
\]
The mean of the $\bi=(i,j,t)$ edge's weight satisfies
\[
\mathbb{E}(Y_{ijt}) = \exp\!\left(c + \beta_\star X_{ijt} + u_{ij} + w_{it} + v_{jt}\right),
\]
where the constant $c$ can be selected to control the density $|E|/n$ of the graph. We generate both Poisson data (satisfying Assumption~\ref{assump:model} with intensity $\lambda_\bi = \lambda_{ijt}$) and non-Poisson data.  
To obtain non-Poisson outcomes, we generate $Y_{\bi}$ as Negative Binomial via a Gamma--Poisson mixture: the rate of the Gamma controls the variance, while its shape is scaled to match the desired mean $\lambda_{\bi}$.  
Setting the rate to $\infty$ recovers the Poisson model; for experiments with overdispersion, we take the rate equal to $0.1$. We executed $600$ replications of each configuration.

We compare three estimators: PPML, PPML (Debiased), and our Polyads estimator.  
For PPML we use the fast implementation of \cite{correia2019ppmlhdfe}; for PPML (Debiased) we use the analytical correction of \cite{weidner2021bias} via their Stata package. Our first experiment varies the graph density with 
\[
|E| \in \{0.02n,\ 0.03n,\ 0.04n,\ 0.05n,\ 0.1n\},
\quad\text{and}\quad
n_1 = n_2 \in \{50,100\}.
\]
Each configuration is replicated $600$ times.  
Figures~\ref{fig:ipp_visualize_densities} and \ref{fig:ipp_visualize_boxplots} summarize the results.  
The first one shows the distributions of the normalized errors $\sqrt{n}\,(\widehat{\beta}_\Xi - \beta_\star)$ at densities 2\%, 5\%, and 10\%.  
The second reports, for $n_1=n_2=50$, from left to right, the empirical coverage of the 95\% confidence interval (which should be close to 95\%), the convergence rate, and the running time.  
We observe:

\begin{itemize}
    \item \textbf{Incidental parameter bias.}  
    PPML exhibits a clear incidental parameter problem (blue curves shifted to the right), especially at low densities.  
    Our Polyads estimator eliminate this bias and the debiased PPML, at larger sample sizes $n$, do not display bias.

    \item \textbf{Instability at low density.}  
    For small $n$ and sparse graphs, PPML and PPML (Debiased) often fail to converge and produce poor coverage.

    \item \textbf{Coverage issues.}  
    Even at larger samples, PPML intervals remain unreliable (best coverage $\approx 90\%$ for $n_1=n_2=100$ and $|E|=0.1n$).  
    The debiased method improves coverage but still undercovers at low densities and $n$, and overcovers at high densities.

    \item \textbf{Computational efficiency.}  
    Our estimator is substantially faster at low densities and remains competitive as density increases.
\end{itemize}

\begin{figure}[ht!]
    \centering
    \includegraphics[width=\linewidth]{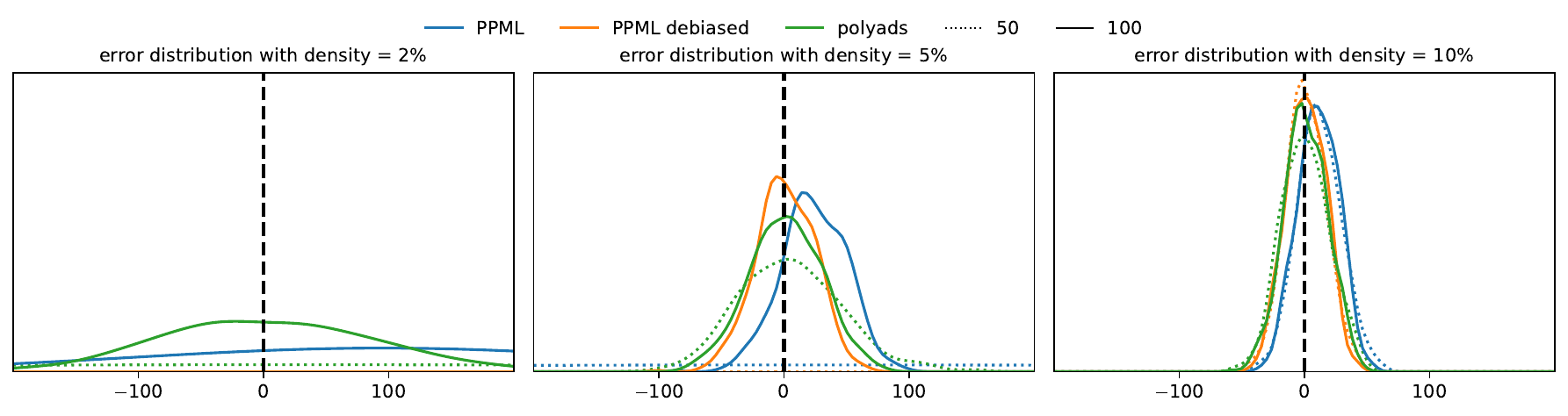}
    \caption{Distributions of the normalized errors of PPML, PPML (Debiased), and the Polyads estimator across different graph densities.}
    \label{fig:ipp_visualize_densities}
\end{figure}

\begin{figure}[ht!]
    \centering
    \includegraphics[width=\linewidth]{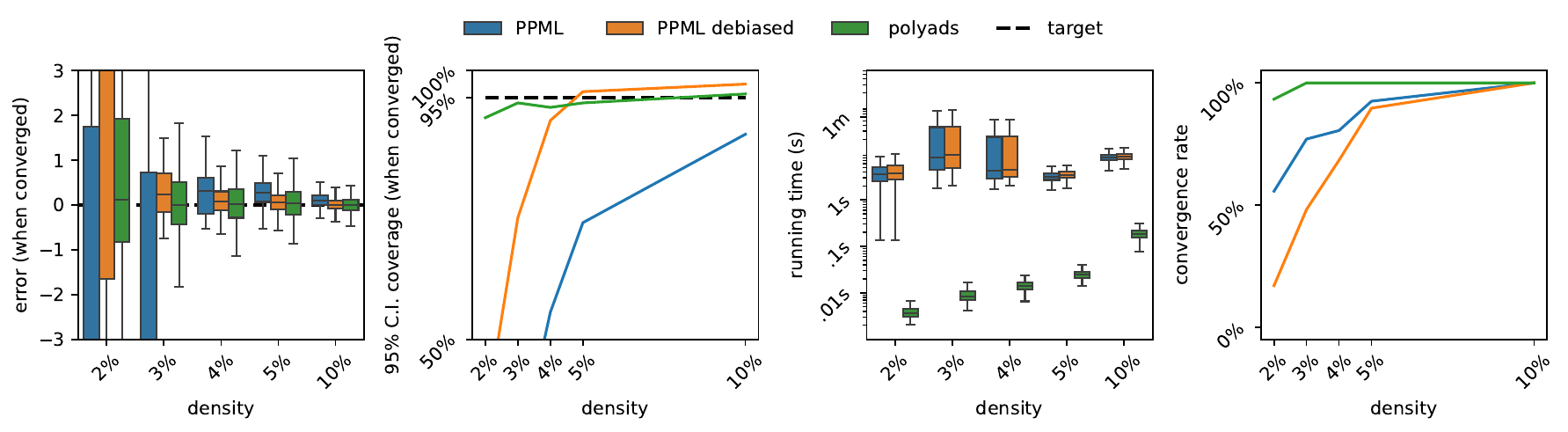}
    \caption{Comparison of PPML, PPML (Debiased), and the Polyads estimator across different graph densities while keeping $n_1=n_2=50$.}
    \label{fig:ipp_visualize_boxplots}
\end{figure}

We next consider a second experiment aimed at evaluating performance in the sparse regime. Here we let $n$ vary from $2\times 10^{5}$ to $32\times 10^{5}$ and choose the constant $c$ so that
\[
|E| \approx 4\sqrt{n},
\]
thus forcing the density $\frac{|E|}{n}$ to shrink towards zero as the sample size grows. Figure~\ref{fig:artificial_poisson_sparse} reports the results for the Poisson case and Figure~\ref{fig:artificial_nb_sparse} reports the results for the Negative Binomial case. The qualitative conclusions of both cases are similar and consistent with those found in the low-density examples from Figure \ref{fig:ipp_visualize_densities}, but become even more pronounced under sparsity:

\begin{figure}[ht!]
    \centering
    \includegraphics[width=\linewidth]{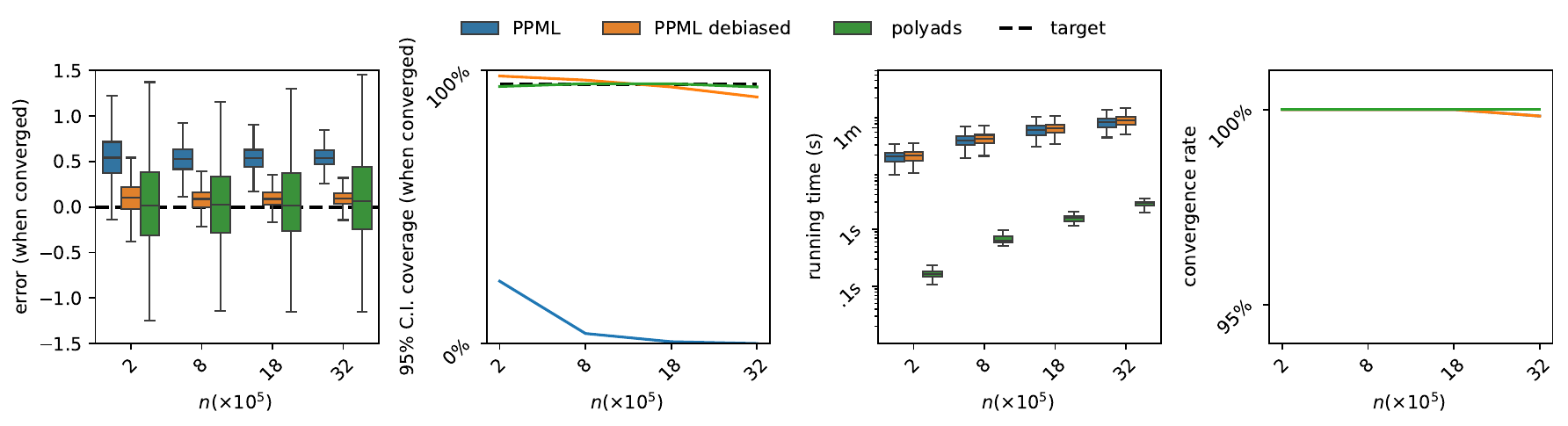}
    \caption{Sparse case ($|E| = 4\sqrt{n}$, Poisson model): Polyads estimator maintains good performance in sparse settings.}
    \label{fig:artificial_poisson_sparse}
\end{figure}

\begin{itemize}
    \item \textbf{Estimation error.} The first panel presents boxplots of the non-scaled estimation error $\widehat{\beta}_\Xi-\beta_\star$ (conditional on convergence).  
    PPML again displays a pronounced incidental parameter bias, with its median shifted upward across all sample sizes.  
    The debiased PPML estimator reduces, but does not eliminate—this distortion.  
    In contrast, the Polyads estimator remains concentrated around $\beta_\star$ for every $n$. More interestingly, as $n$ increases the variance of the PPML (Debiased) approach shrinks, but the incidental parameter bias remains relevant in a way that its coverage starts to decay. 

    \item \textbf{Coverage of the 95\% confidence interval.}
    The second panel highlights the severity of the incidental parameter problem in sparse settings.  
    PPML confidence intervals exhibit near-zero coverage throughout: the bias places the estimator far outside the nominal interval in almost all replications.  
    The debiased PPML method improves coverage but still undercovers for smaller $n$, where sparsity is most acute.  
    The Polyads estimator, by contrast, maintains coverage close to the nominal 95\% level uniformly across all sample sizes.

    \item \textbf{Running time.}
    The third panel reports computation times.  
    Both PPML and debiased PPML become increasingly expensive as $n$ grows, despite the sparsity of the graph.  
    The Polyads estimator directly leverages sparsity and is substantially faster, especially for moderate and large $n$.
\end{itemize}

Across bias, coverage, computation time, and convergence, the results in this sparse regime reinforce the findings from the low-density experiment: PPML exhibits severe incidental parameter bias and essentially zero inferential validity; the debiased PPML estimator improves upon PPML but continues to suffer from miscoverage under sparsity; the Polyads estimator remains accurate, fast, and statistically reliable, even when sparsity increases with sample size.

\begin{figure}[ht!]
    \centering
    \includegraphics[width=\linewidth]{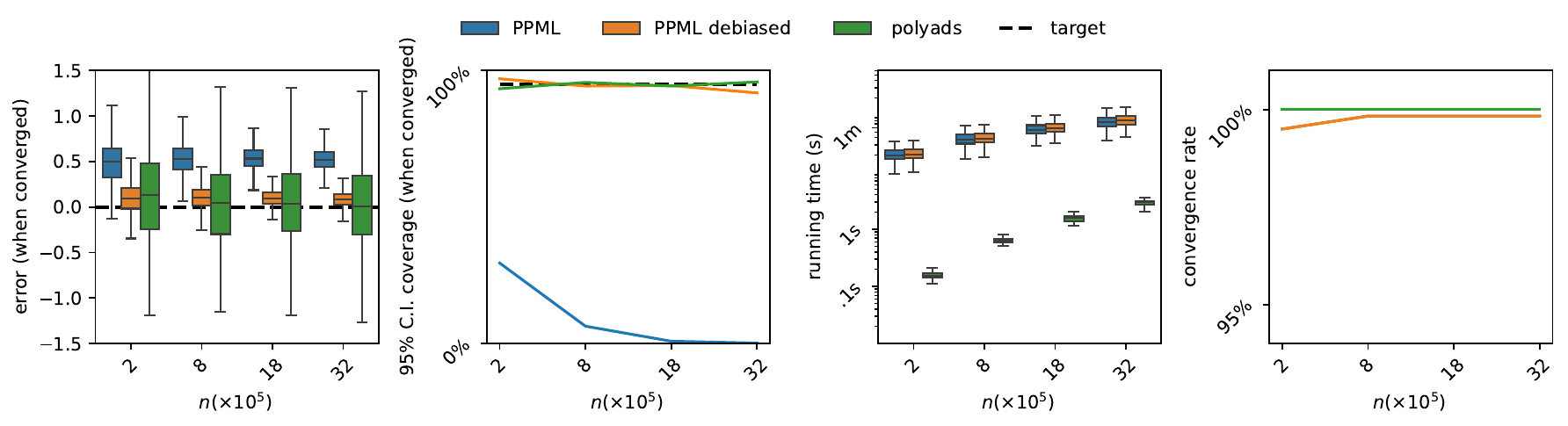}
    \caption{Sparse case ($|E| = 4\sqrt{n}$, Negative Binomial model): Polyads estimator remains robust under Negative Binomial noise.}
    \label{fig:artificial_nb_sparse}
\end{figure}

Comparing Figures \ref{fig:artificial_poisson_sparse} and~\ref{fig:artificial_nb_sparse} we notice that the Negative Binomial case closely match the Poisson one: PPML remains biased, while both the debiased PPML and our Polyads estimator remove the bias and achieve better coverage.  
Thus, although our theory assumes a Poisson model, these results indicate that the Polyads estimator is empirically robust to overdispersion and model misspecification.

\begin{figure}[ht!]
    \centering
    \includegraphics[width=\linewidth]{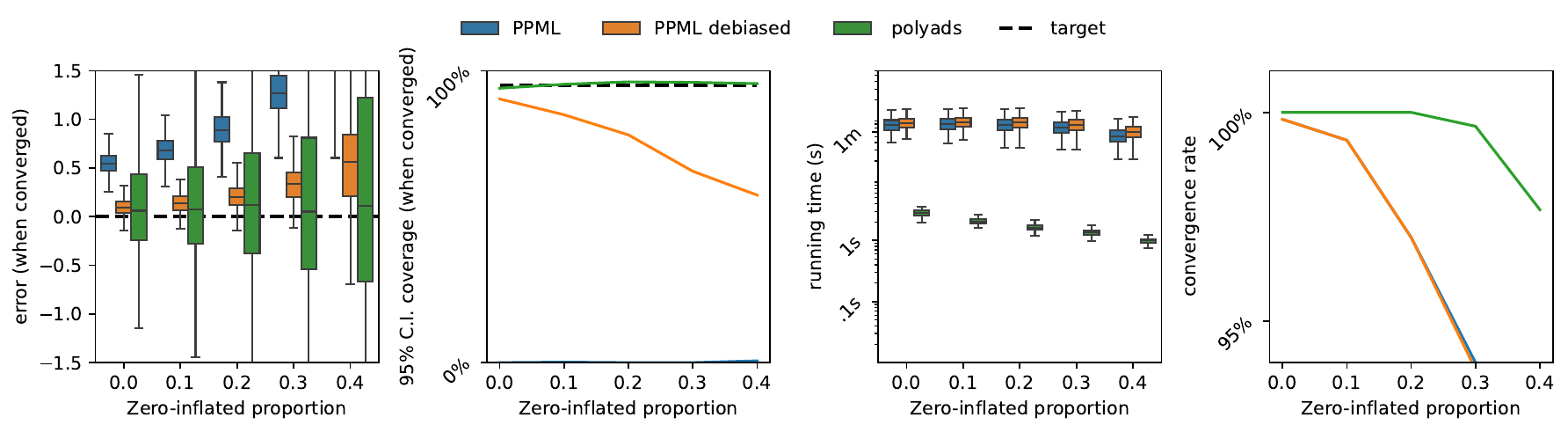}
    \caption{Sparse case ($|E| = 4\sqrt{n}$, Poisson model with inflated zeros): Polyads estimator remains robust under deletions.}
    \label{fig:zero_inflated}
\end{figure}

Finally, we also consider the robustness of the method to zero-inflation. In this case, we start with networks generated following the sparse setup with $n=32 \times 10^5$ and Poisson distribution. We proceed by deleting non-zero entries at random with different deletion probabilities ranging between $0$ and $0.4$. The results are displayed in Figure \ref{fig:zero_inflated}, which shows that both PPML and PPML (Debiased) are highly sensitive to deletions, with performance degrading as the probability of deletion increases. Meanwhile, our method continues to deliver valid confidence intervals, although with increasing variance.

Experiments about  four-way networks are  presented in Appendix~\ref{sec:four_way}. 

\subsection{Real data}
\label{sec:experiments:real}

We exploit health insurance claims data that cover the universe of physician consultations in France over the years 2016 to 2018.
The vast majority of French generalist practitioners (GPs) are subject to fee regulation.
In May 2017, the government increased the regulatory fee level by 8.7\%.
Our goal is to assess how this reform has affected the network of doctors-patients connections.
A difference-in-differences analysis shows that the stronger financial incentives have caused 
physician activity (as measured by number of visits) to rise by approximately 10\%. 
See Appendix~\ref{app:health} for  detail about the data, the considered control groups and the evaluation method.


Given the strong policy concern about geographic access to physician services, it is important to understand how the reform has transformed the patient-doctor network in the spatial dimension. 
Has the reform generated encounters of  doctors and patients located further apart? 
Using a  three-way model and controlling for dyads fixed effects allows to estimate how the reform has affected the geography of doctor-patient connections. 
We compare below 
the standard PPML estimator, the analytical correction  of \cite{weidner2021bias}, hereafter ``PPML (Debiased)'', and our Polyads estimator. As explained above, we highlight the role of geographic distance (spatial accessibility). To illustrate the asymptotic bias of PPML, we also examine whether the reform has affected gender homophily between patients and doctors. 

The outcome~$Y_{ijt}$ is the number of visits by patient~$i$ to doctor~$j$ on month~$t$. Given our two dimensions of interest (geography and gender), we aggregate data at the city-sex level. The index~$i$ (resp. $j$) thus stands thus for the set of patients (resp. doctors) in a given municipality with given gender. The high number of  potential patients in each city-sex group makes to the Poisson assumption plausible.\footnote{The aggregate number of consultations in each group is the sum of (possibly heterogeneous) Bernoulli distributions that represent the occurrence of a consultation for all potential patients in the group. This sum  follows approximately a Poisson distribution 
under the conditions exposed in \cite{le1960approximation}. The  assumption that the individual occurrences of a consultation are independent across potential patients is relaxed in \cite{galambos1973general} and \cite{serfling1978some}.
} The treatment~$T_{j}$ is a binary variable equal to 1 for sector 1 GPs, and to 0 for our control group (direct access specialist physicians, see Appendix~\ref{app:health}). The reform has been implemented from May 2017 onward, hence the definition of $\text{Post}_t$, a dummy variable equal to 1 after that date. On top of the 
interaction between~$\text{Post}_t \times T_j$ (as in any difference-in-differences approach), the model includes three features that account for homophily preferences in the gender and spatial dimensions:  (i) a dummy variable equals to 1 if patients and doctors have the same sex; (ii) a dummy variable that is equal to 1 if the doctor's practice is located in the same municipality as the patient's home; and (iii) travel time between patient's home and doctor's practice.\footnote{Travel time is measured in minutes between the centroids of municipalities.}

We consider the three-way Poisson model $Y_{ijt}\sim\mathcal{P}(\lambda_{ijt})$ with intensity given by
\begin{equation}
    \ln\lambda_{ijt}=\big(\beta_\texttt{d}  d_{ij}+\beta_\texttt{sc} \mathds{1}\{\text{city}_i=\text{city}_j\}+\beta_\texttt{ss} \mathds{1}\{\text{sex}_i=\text{sex}_j\}\big) \times \text{Post}_t \times T_j  +u_{ij}+v_{jt}+w_{it}.
    \label{eq:Poisson:health:reform}
\end{equation}

Because the full dataset contains $n_1 = 69{,}265$ patient groups, $n_2 = 16{,}941$ doctors, and $n_3 = 34$ months --- corresponding to roughly $n \approx 40$ billion edges and $|E| = 56{,}034{,}015$ --- direct estimation on the full graph is computationally infeasible for all methods. We therefore adopt a subsampling strategy combined with meta-analysis. This approach allows us to assess the robustness of the various estimators to subsampling.

\begin{table}[ht!]
\centering
\begin{tabular}{l l c c c}
\toprule
&  &  & \textbf{Subsample proportion} &  \\
\textbf{Parameter} & \textbf{Method} & \textbf{2\%} & \textbf{3\%} & \textbf{4\%} \\
\midrule
 & PPML & 1.94 {\small (1.43, 2.44)} & 1.87 {\small (1.56, 2.18)} & 1.66 {\small (1.48, 1.84)} \\
$\beta_\texttt{d} (\times 10^{4})$ & PPML (Debiased) & 15.14 {\small (-8.33, 38.61)} & 1.33 {\small (0.57, 2.09)} & 1.44 {\small (0.99, 1.89)} \\
 & Polyads & 1.23 {\small (0.26, 2.20)} & 1.73 {\small (1.18, 2.28)} & 1.45 {\small (1.03, 1.86)} \\
\midrule
& PPML & -6.54 {\small (-11.45, -1.63)} & -4.28 {\small (-6.47, -2.08)} & -4.46 {\small (-5.88, -3.04)} \\
$\beta_\texttt{sc} (\times 10^{2})$ & PPML (Debiased) & -62.17 {\small (-174.73, 50.40)} & -0.98 {\small (-5.90, 3.94)} & -3.17 {\small (-6.13, -0.21)} \\
 & Polyads & -7.60 {\small (-13.10, -2.10)} & -3.91 {\small (-6.55, -1.27)} & -4.52 {\small (-6.28, -2.75)} \\
\midrule
 & PPML & -1.54 {\small (-2.84, -0.23)} & -0.91 {\small (-1.67, -0.15)} & -0.78 {\small (-1.29, -0.27)} \\
$\beta_\texttt{ss} (\times 10^{2})$ & PPML (Debiased) & -29.78 {\small (-60.67, 1.11)} & -0.58 {\small (-1.90, 0.74)} & -2.07 {\small (-4.61, 0.47)} \\
 & Polyads & -2.48 {\small (-4.45, -0.50)} & -0.35 {\small (-1.61, 0.92)} & -0.20 {\small (-1.13, 0.73)} \\
\bottomrule
\end{tabular}
\medskip
\caption{Random effects meta-analysis estimates with 95\% confidence intervals (in parentheses) by subsample proportion and method (100 replications). $\beta_\texttt{d}$ estimates are scaled by $10^4$; $\beta_\texttt{sc}$ and $\beta_\texttt{ss}$ by $10^2$.}
\label{tab:meta_analysis}
\end{table}

For each configuration, we sample a proportion $s \in \{2\%, 3\%, 4\%\}$ of patient groups and the same proportion of doctors, while always retaining all $34$ months\footnote{Notice that we sample the groups of patients, the doctors and then get all edges through the sampled nodes, keeping all times. This sampling strategy is consistent with Assumption \ref{assump:model} and also with the sampling assumptions in \cite{fernandez2016individual, weidner2021bias, graham2017econometric, jochmans2018semiparametric}. One could also propose sampling directly the edges and not the nodes, but there is evidence that this procedure can lead to bias, see \cite{subsamplegraphs} for details.}. Each subsample consists of independent random draws of patient and doctor groups; doctors in the treatment and control groups are sampled independently to ensure comparability. We repeat the procedure independently to obtain $100$ subsamples. The final estimates are obtained via a random-effects meta-analysis using the default implementation in the \texttt{statsmodels} Python library, which is based on the iterated method of \cite{paule1982consensus}, a refinement of the classical method of \cite{dersimonian1986meta}.

\begin{figure}
    \centering
    \includegraphics[width=\linewidth]{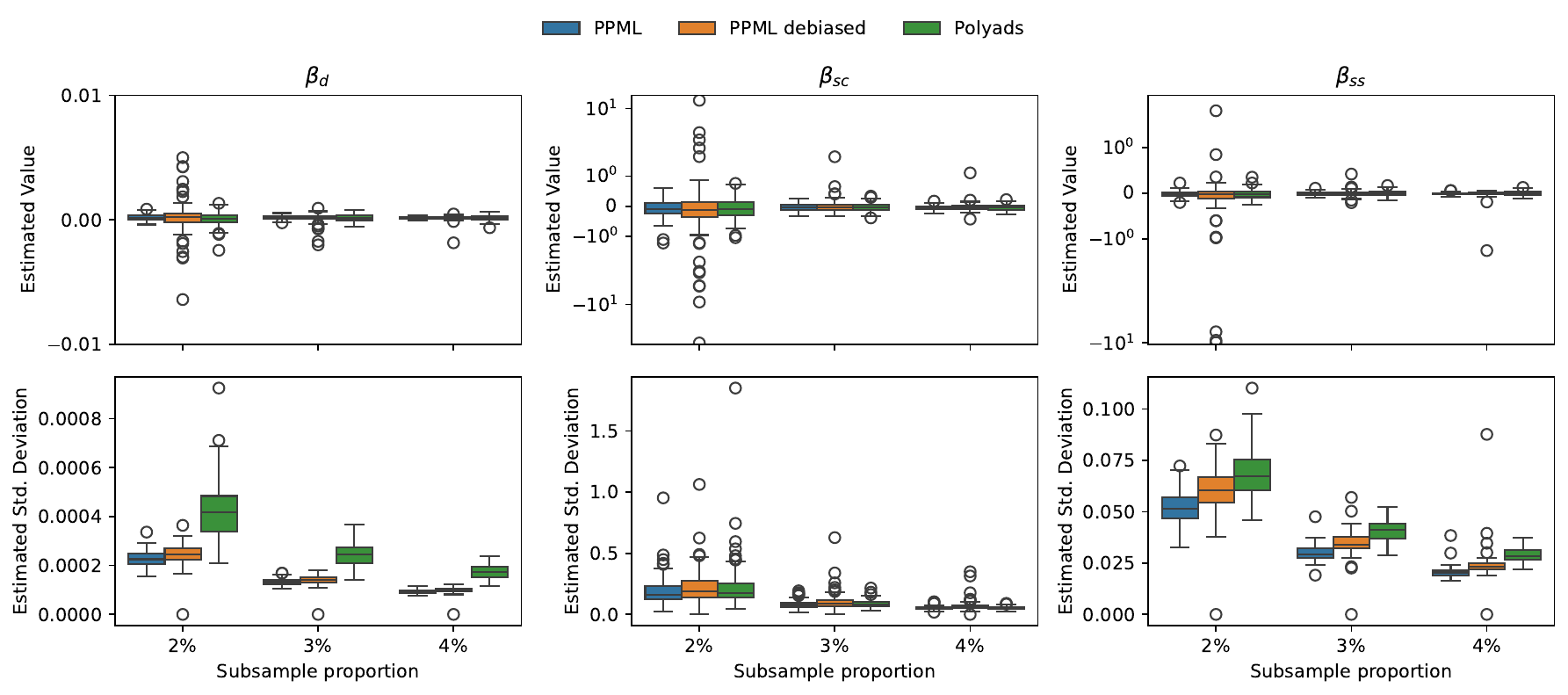}
    \caption{Estimated coefficients and standard deviations by subsample proportion and estimation method (100 replications each).}
    \label{fig:real_beta_boxplots}
\end{figure}

The results from the meta-analysis are reported in Table~\ref{tab:meta_analysis}. 
The three methods yield similar conclusions about the effect of the policy reform on the geography of the network. 
The estimates of the distance and same city parameters, $\hat\beta_\texttt{d} \approx 1.4\ 10^{-4}$
and $\hat\beta_\texttt{sc} \approx -4.5\ 10^{-2}$, show that the policy reform has attenuated the (negative) effect  of the patient-doctor distance  and the (positive) effect of patient and doctor being located in the same city.\footnote{To get a sense of the baseline values of $\beta_\texttt{d}$ and $\beta_\texttt{sc}$, we estimate a cross-sectional version of equation~\eqref{eq:Poisson:health:reform}, namely $\ln\lambda_{ij}=\beta^0_\texttt{d} d_{ij}+\beta^0_\texttt{sc}\mathds{1}\{\text{city}_i=\text{city}_j\}+\nu_i+\xi_j$ and find $\hat\beta^0_\texttt{d}\approx-0.03,\hat\beta^0_\texttt{sc}\approx7.7$. (The estimation is  based on a 3\% subsample of fee-regulated GPs in 2016.) In other words, the reform has caused these two parameters to decrease in absolute value by approximately 0,5\%.} In other words, the reform has caused the concerned doctors to attract patients located further away from their practice.
 In particular, as a result of the reform, consultations with doctors in a different city from the patient rose more rapidly (about +4 percentage points) than those with doctors in the same city
: spatial homophily (the tendency for patients and doctors to be located in the same city) 
has decreased due to the increase in doctors' fees.
Regarding gender homophily, the PPML method  disagrees 
with both PPML (Debiased) and Polyads. With a 3\% or 4\%  subsample proportion, the latter two methods do not find evidence that the policy reform has affected the degree of gender homophily, while PPML suggests reduced homophily. This is consistent with the PPML estimator of $\beta_\texttt{ss}$ being asymptotically biased.

\smallskip

Figure~\ref{fig:real_beta_boxplots} presents the results for each of the 100 replications. The top line shows that the Polyads estimator has much fewer outlier values than PPML and PPML (Debiased), the phenomenon being particularly pronounced for the 2\% and 3\% sampling rates.
The bottom line shows the distribution across the 100 subsamples of the estimated standard deviation of each estimator.
We observe that the Polyads method  yields larger standard errors, which is consistent with the experimental results of Section~\ref{sec:experiments:artificial}. 

The variance of the meta-analysis is a combination of the inter-study variance and the outer-study variance. 
The Polyads estimator has higher inter-study variance. However, it  has much smaller outer-study variance
because it has almost no outliers.
Even with PPML/PPML (Debiased) yielding smaller CIs for each run, the final aggregation yields smaller confidence intervals for the Polyads estimator. 
For instance, with a 4\% subsample, 
Table~\ref{tab:meta_analysis} reports  confidence intervals for $\beta_\texttt{sc}$ and $\beta_\texttt{ss}$ respectively 44\% and 73\% smaller under the Polyads method than under PPML (Debiased).
Because of the greater precision, the parameter $\beta_\texttt{sc}$ estimated with the Polyads methods appears significantly negative at the 5\% confidence level for the 2\%, 3\% and 4\% subsample sizes, while under PPML (Debiased) statistical significance occurs only for the 4\% subsample proportion. 

For the 2\% subsample proportion, PPML (Debiased) yields very  large large point estimates and confidence intervals, with $\hat\beta_\texttt{d}$ and $\hat\beta_\texttt{sc}$ appearing as non statistically different from zero. This is true even though
the number of non-zero observations $Y_{ijt}$ is already substantial ($|E|\approx 20,000$) for this sampling proportion.
The phenomenon is much less pronounced for the Polyads method. With this method, the confidence interval for $\hat\beta_\texttt{d}$ and $\hat\beta_\texttt{sc}$ do not contain zero.
The poorer precision and greater  instability of PPML (Debiased) comes from the existence of many outliers, see Figure~\ref{fig:real_beta_boxplots}.

\smallskip

\begin{table}[h]
\centering
\begin{tabular}{l c c c c c c}
\toprule
 
& & & & & \textbf{Time (min)} &  \\
\textbf{Subsample (\%)} 
& \textbf{$n$} & \textbf{$|E|$} & \textbf{$|\Xia^\star|$} & \textbf{PPML} & \textbf{PPML (Debiased)} & \textbf{Polyads} \\
\midrule
2\% & 15,958,490 & 22,644 & 151,743 & 19 & 53 & 5 \\
3\% & 35,906,602 & 50,617 & 739,271 & 126 & 199 & 30 \\
4\% & 63,833,959 & 89,520 & 2,301,766 & 203 & 315 & 111 \\
\bottomrule
\end{tabular}
\medskip
\caption{Computational time: average number of edges, average number of observed positive edges, average number of active polyads (without counting permutations), and average running times (in minutes) for each method across $100$ replications.}
\label{tab:computational}
\end{table}

Finally, as reported on Table~\ref{tab:computational}, the running time of PPML and debiased PPML scales with $n$ while the polyads time scales with $|E|^2$.

\bigskip

\section*{Code resources}

A fast python-based implementation of the polyads method here presented is available at the Github repository \texttt{lucasresenderc/polyads}, available at \href{https://github.com/lucasresenderc/polyads}{https://github.com/lucasresenderc/polyads}. The scripts to run all experiments and produce all figures are available in the Supplementary Material.

\section*{Acknowledgments}

We are grateful to \'Aureo de Paula, St\'ephane Bonhomme, Cl\'ement de Chaisemartin, Laurent Davezies, Yannick Guyonvarch, Xavier D'Haultf\oe uille, Koen Jochmans, Francis Kramarz, Thierry Magnac, Chris Muris, and Cavit Pakel  for insightful comments. We thank the Agence Nationale de la Recherche for financial support (ANR-23-CE36-0014).

\bibliography{references}  

\newpage

\numberwithin{equation}{section}

\setcounter{equation}{0}  
\setcounter{table}{0} 
\setcounter{figure}{0}  

\numberwithin{figure}{section}
\numberwithin{lemma}{section}
\numberwithin{table}{section}

\appendix
\numberwithin{equation}{section}
\setcounter{equation}{0}  
\setcounter{figure}{0}  
\setcounter{proposition}{0}  
\setcounter{lemma}{0}  
\setcounter{theorem}{0}  
\numberwithin{figure}{section}
\numberwithin{lemma}{section}
\numberwithin{corollary}{section}
\numberwithin{proposition}{section}
\numberwithin{theorem}{section}

\section{Information loss}\label{sec:loss_info}

Throughout this section, we specialize to the two-way ($D=2$) bipartite case for clarity of exposition; polyads are then tetrads, indexed by pairs of rows and pairs of columns. The construction extends to general $D$-way arrays exactly as described in the main text, with tetrads replaced by higher-order polyads, but the two-way case already contains the essential ideas and carries lighter notation.

The Polyads estimator circumvents the incidental parameter problem and remains computationally tractable for large sparse networks, as shown in the previous sections. Since the estimator is consistent and computationally efficient, this convenience must be paid for, if at all, through a loss of statistical efficiency. This section quantifies that loss.\footnote{We thank St\'ephane Bonhomme for raising this question.}

Proposition~\ref{prop:from:degrees:to:polyads} already suggests why such a loss should be expected. That proposition shows that the set of all graphs sharing the observed degree sequence (the conditioning set defining the CMLE) is exactly the set reachable from $Y$ by composing elementary polyad transformations $T_\xi^r$, chained in arbitrary order and arbitrary length, $y' = T_{\xi_1}^{r_1}\circ\cdots\circ T_{\xi_m}^{r_m}(y)$. The CMLE conditions on this entire reachable set at once. By contrast, the Polyads estimator conditions on each link of the chain separately, and combines them by summing scores as if these links were independent, rather than conditioning on the chain as a whole. This simplification is expected to lead to some loss of statistical efficiency, and the remainder of this section makes this loss precise.

The loss decomposes into two steps. From the unconditional MLE (UMLE, which estimates both $\beta$ and the fixed effects $\theta$, \textit{i.e.}, the classical PPML) to the Conditional MLE (CMLE, which conditions on the full degree sequence), there is no loss of information; this is a classical property of the Poisson model, analogous to the elimination of incidental parameters in conditional logit \cite{andersen1970asymptotic}. From the CMLE to the Polyads estimator, there is a loss of statistical efficiency, and this second step is the object of this appendix.

We formalize this by comparing the asymptotic covariance matrices of the UMLE, CMLE, and Polyads estimator under the Poisson model: $Y = (Y_\bi)_{\bi\in\cI}$ for $\cI=[n_1]\times [n_2]$, where the $Y_\bi$ are independent (conditionally on the design) with $Y_\bi|X\sim \cP(\lambda_\bi)$ and
\begin{equation}\label{eq:2_way_model}
    \lambda_{\bi} = \lambda_{ij} = \exp(X_{ij}^\top \beta_\star + \alpha_i + \psi_j).
\end{equation}

Two structural parameters govern the magnitude of the information loss, corresponding to the two manners in which the Polyads estimator objective departs from the CMLE. First, the CMLE extracts information edge by edge, each $Y_\bi$ contributing independently; the Polyads estimator instead extracts information polyad by polyad, from the contrast between edges entering a given tetrad with a positive and a negative sign, and it is the aggregate score across polyads, rather than across edges, that identifies $\beta$. Treating the resulting polyad-level scores as if they were mutually independent, exactly as the CMLE treats edges, gives the following two informational quantities: the sum $\sum_{\bi} \text{Var}_{\beta_\star}(Y_\bi)\tilde X_\bi^2$, the CMLE's information if computed edge by edge, and the sum $\sum_{\xi\in\Xi}\mathbb{E}_{\beta_\star}[V_\xi(Y)]\tilde X_\xi^2$, the polyad information if its terms were likewise independent across polyads. Their ratio defines the \textbf{feature concentration factor},
\begin{equation}\label{eq:feature_concen_factor}
    \kappa_X := \frac{\sum_{\bi} \text{Var}_{{\beta_\star}}(Y_\bi) \tilde{X}_\bi^2}{\sum_{\xi \in \Xi } \mathbb{E}_{{\beta_\star}}[V_\xi(Y)] \tilde{X}_\xi^2},
\end{equation}
where $V_\xi(y) = \bV_{{\beta_\star}}\left(m_\xi(Y) \mid Y\in\cO_\xi(y)\right)$, $\tilde{X}_\xi$ is the polyad (difference-in-differences) feature introduced above, and $\tilde{X}_\bi$ is the residualized edge feature obtained after partialling out the fixed effects (Theorem~\ref{theo:no_loss_info_UMLE_CMLE} gives the exact definition).

But polyads are not, in fact, independent of one another: two polyads that share an edge are statistically dependent, since that edge's realization enters both scores. This is precisely what the second quantity, the \textbf{effective expected local overlap degree}, controls: for a realization of the graph $Y$ and a polyad $\xi$, let $d_\xi(Y)$ denote the number of active polyads sharing at least one edge with $\xi$, and define
\begin{equation}\label{eq:eff_exp_loca_overlap_degree}
    \bar{m}_a := \max_{\xi \in \Xi} \frac{\mathbb{E}_{{\beta_\star}} \left[ d_\xi(Y) (\nabla l_\xi)^2 \right]}{\mathbb{E}_{{\beta_\star}} \left[ (\nabla l_\xi)^2 \right]}.
\end{equation}
Equivalently, in the auxiliary graph whose nodes are active polyads and whose edges link polyads sharing an edge, $d_\xi(Y)$ is the degree of node $\xi$, and $\bar m_a$ is an information-weighted version of this degree. It is $\bar m_a$, then, that measures how much the independence assumption implicit in $\kappa_X$ can be trusted: $\bar m_a$ stays small when active polyads are sparsely connected, as occurs when the underlying network is itself sparse, and grows as overlap among polyads increases.

The main result of this section, obtained by combining Theorems~\ref{theo:no_loss_info_UMLE_CMLE}, \ref{theo:info_loss_LB}, and~\ref{theo:UB_var_poly}, is the following bound. For the one-dimensional case ($p=1$),
\begin{equation}
    \label{eq:UandL_bound_var_polyads_UMLE}
    v_{\text{UMLE}} \leq v_{\text{Polyads}} \le (\bar{m}_a \kappa_X)\, v_{\text{UMLE}},
\end{equation}
where $v_{\text{UMLE}}$ and $v_{\text{Polyads}}$ denote the asymptotic variances of the UMLE and Polyads estimator. The left-hand inequality is the usual Cram\'er-Rao bound, since the UMLE is efficient and the Polyads estimator is a regular Z-estimator. The right-hand inequality shows that the factor $\bar{m}_a \kappa_X$ governs the variance inflation.

Section~\ref{sec:uppe_bound_info_loss} works out $\bar m_a$ and $\kappa_X$ in two contrasting examples, one sparse and one dense. In the sparse example, active polyads rarely share edges, so $\bar m_a$ is small, while $\kappa_X$ is large. In the dense example, $\bar m_a$ grows, as each active polyad now overlaps with many others, but $\kappa_X$ falls correspondingly, since the number of active polyads is then large enough that the denominator grows. In both cases the product $\bar m_a \kappa_X$ turns out to be $\Theta(1)$, suggesting that the Polyads estimator can remain efficient up to a constant in different network topologies.

We now turn to the proof of \eqref{eq:UandL_bound_var_polyads_UMLE}. We analyze the two transitions, UMLE to CMLE and CMLE to the Polyads estimator, in turn. To streamline notation, we omit the explicit conditioning on $X$ in what follows.

\subsection{No loss of information from UMLE to CMLE in the Poisson model}
In this section, we recall that there is zero loss of statistical information when estimating the structural parameters of a Poisson network using the Conditional Maximum Likelihood approach (with respect to the full vector of degrees, i.e. the exhaustive statistics of the fixed effects) versus the Unconditional Maximum Likelihood approach (estimating both structural and fixed effects parameters). Specifically, profiling out or conditioning on sufficient statistics of the (high-dimensional) fixed effects results in no efficiency loss for the structural parameters. This is a (rare) mathematical property of the Poisson model that the Polyads estimator will also benefits.


The econometric foundation for conditioning on the sum of counts (the degree sequence in a network context) to eliminate fixed effects without compromising structural parameter estimation was pioneered by \cite{hausman1984econometric}. \cite{andersen1970asymptotic} showed that the CMLE is both consistent and asymptotically efficient when conditioning on sufficient statistics. Furthermore, \cite{palmgren1981fisher} proved that the Fisher Information matrix in log-linear models becomes orthogonal under this conditioning. This foundational result guarantees that the asymptotic covariance matrices of the conditional maximum likelihood estimator and the unconditional ML estimator of the structural parameter are identical (see \cite{cameron2013regression} and \cite{wooldridge2010econometric}). For the sake of completeness, we recall this result because it controls our first potential source of information leak and provides an explicit formulae for the asymptotic covariance matrice of the UMLE and CMLE to which we compare later the one of the Polyads estimator.

We introduce the following notation: let $W = \text{Diag}(\lambda)$ be the $(n_1 n_2)\times(n_1 n_2)$ diagonal matrix with diagonal vector $\lambda = (\lambda_\bi)_\bi$. Let
\begin{equation*}
    D:\theta = \begin{pmatrix} \alpha \\ \psi \end{pmatrix} \in\bR^{n_1+n_2} \to (\alpha_i + \psi_j)_{ij} \in \bR^{n_1 n_2}
\end{equation*} 
be the fixed effects ``design matrix'', and let $X = \left( X_\bi^\top \right)_{\bi\in\cI}:\beta\in\bR^p\to (\inr{X_\bi, \beta})_{i\in\cI}$ be the design matrix of the dyads' features.

\begin{theorem}[\cite{cameron2013regression,wooldridge2010econometric}] \label{theo:no_loss_info_UMLE_CMLE}
We denote by $\text{Var}(\hat{\beta}_{\text{UMLE}})$ and $\text{Var}(\hat{\beta}_{\text{CMLE}})$ the asymptotic covariance matrices of the unconditional and conditional MLE, respectively.\footnote{In the 2-way Poisson model, as $n_1, n_2\to \infty$, they satisfy $\text{Var}(\hat{\beta}_{\text{UMLE}})^{-1/2}(\hat{\beta}_{\text{UMLE}} - \beta_\star)\to \cN(0,I_p)$ and  $\text{Var}(\hat{\beta}_{\text{CMLE}})^{-1/2}(\hat{\beta}_{\text{CMLE}} - \beta_\star)\to \cN(0,I_p)$, see \cite{fernandez2016individual, weidner2021bias,  andersen1970asymptotic}. Strict exogeneity is assumed in our model~\eqref{assump:model}.} We have 
\begin{equation*}
    \text{Var}(\hat{\beta}_{\text{UMLE}}) = \text{Var}(\hat{\beta}_{\text{CMLE}}) = \left(\tilde X^\top W \tilde X \right)^{-1},
\end{equation*} 
where $\tilde X = W^{-1/2}(I - P_{W^{1/2}D})W^{1/2}X$ and $P_{W^{1/2}D}$ is the projection operator onto the columns space of $W^{1/2}D$, i.e. $P_{W^{1/2}D} = W^{1/2}D(D^\top W D)^{-1} D^\top W^{1/2}$.
\end{theorem}

In other words, we have:
\begin{equation*}
    \text{Var}(\hat{\beta}_{\text{UMLE}}) = \text{Var}(\hat{\beta}_{\text{CMLE}}) = \left(\sum_{\bi} \bV_{{\beta_\star}}(Y_\bi) \tilde X_\bi \tilde X_\bi^\top \right)^{-1},
\end{equation*}


where $\tilde X = (\tilde X_\bi^\top)_\bi$. The information loss by estimating the fixed effects in the UMLE approach  is identical to the variance removed by conditioning on the sufficient statistics in the CMLE. This property can be slightly extended beyond the Poisson model to the linear exponential family.

\subsection{The loss of  information from CMLE to the Polyads estimator}
Instead of conditioning the likelihood on the set of \textit{all} (multi-way) networks with the same degrees as the observed network, the Polyads estimator $\hat{\beta}_{\Xi}$ operates by conditioning, for each polyad, on the subset (i.e. the orbit) it induces. Following the 'no loss of information' result from Theorem~\ref{theo:no_loss_info_UMLE_CMLE}, the only reason of information loss is due to this 'local/polyads' conditioning. In this section, we assess the information lost when transitioning from the CMLE (i.e. full degrees conditioning) to the polyads approach (i.e. polyads' orbit conditioning). 

\paragraph{The covariance matrix gap.} We recall that $\hat{\beta}_{\text{CMLE}}$ denotes the Conditional Maximum Likelihood Estimator (CMLE), which conditions on the combinatorial set of all graphs sharing the exact same degrees vector as the observed data. As shown in standard Maximum Likelihood theory \cite{andersen1970asymptotic}, the Information Matrix Equality holds, meaning the variance of the score perfectly matches the expected Hessian. Consequently, the asymptotic variance of CMLE is the inverse of the Fisher Information:
\begin{equation*}
    \text{Var}(\hat{\beta}_{\text{CMLE}}) = \mathcal{I}_{\text{CMLE}}^{-1}.
\end{equation*}

In contrast, our Polyads estimator is fundamentally a composite likelihood (or pseudo-likelihood) estimator. It is defined by averaging the loss functions $\beta\to l_\xi(Y|X,\beta)$ across all (active) polyads. Because this aggregation does not perfectly capture the global joint probability structure of overlapping edges, the Information Matrix Equality (IME) breaks down. As we established, the asymptotic covariance matrix of the Polyads estimator takes a sandwich form:
\begin{equation}\label{eq:cov_mat_polyads}
    \Sigma_{\text{Polyads}} = H_0^{-1} V_0 H_0^{-1},
\end{equation}
where $H_0 = \mathbb{E}_{{\beta_\star}}[-\nabla_\beta U_N(Y, {\beta_\star})]$ is the expected Hessian, $V_0 = \text{Var}_{{\beta_\star}}(U_N(Y, {\beta_\star}))$ is the variance of the score\footnote{We proved directional a.n. of the Polyads estimator in Theorem~\ref{theo:CLT_polyads_estim_Graham} where $V_0$ is replaced by the covariance matrix of the H{\'a}jek projection of the score function. However, the two quantities are asymptotically equivalent as proved in Proposition~\ref{prop:projection_principl}.} and $U_N(y, \beta) := \frac{1}{N_a} \sum_{\xi} \nabla l_\xi(y|X, \beta)$ is the gradient of our loss function. Because the Polyads method drops the complex global dependencies between disjoint polyads, $H_0 \neq V_0$, and the IME does not hold for the Polyads estimator's objective function.

To formally prove the loss of information incurred by the Polyads estimator relative to the CMLE, we rely on the theory of optimal estimating equations, pioneered by Godambe (1960) and prove a Cramér-Rao type of Lower Bound dedicated to the Polyads estimator looked as a regular $Z$-estimator.

\begin{theorem}[Necessary information loss] \label{theo:info_loss_LB}
Let $\mathcal{I}_{\text{CMLE}}$ be the Fisher information matrix of the CMLE, and let $\Sigma_{\text{Polyads}}$ be the asymptotic covariance matrix of the Polyads estimator introduced in \eqref{eq:cov_mat_polyads}. We have $\Sigma_{\text{Polyads}} \succeq \mathcal{I}_{\text{CMLE}}^{-1}$.
\end{theorem}

Theorem~\ref{theo:info_loss_LB} shows that there is a necessary loss of information of the Polyads estimator compare to the CMLE (resulting in larger confidence intervals). For the sake of completeness we recall the proof from \cite{MR1461808, MR999014} adapted to our setup.

\begin{proof}[Proof of Theorem~\ref{theo:info_loss_LB}] Let $S_c(\beta) = \nabla_\beta \ell_c(\beta)$ denote the score function derived from the conditional log-likelihood. We have
\begin{equation*}
    \mathbb{E}_{{\beta_\star}}[S_c({\beta_\star})] = 0 \quad \mbox{and} \quad \text{Var}_{{\beta_\star}}(S_c({\beta_\star})) = \mathbb{E}_{{\beta_\star}}[S_c({\beta_\star}) S_c({\beta_\star})^\top] = \mathcal{I}_{\text{CMLE}}.
\end{equation*}

We recall that $U_N(y, \beta) = \frac{1}{N_a} \sum_{\xi} \nabla l_\xi(y|X, \beta)$, $V_0 = \text{Var}_{{\beta_\star}}(U_N(Y, {\beta_\star}))$ and the expected Hessian is $H_0 = \mathbb{E}_{{\beta_\star}}[-\nabla_\beta U_N(Y, {\beta_\star})]$. We first observe that for all $\beta$,
\begin{equation}\label{eq:critica_eq_polyads}
    \bE_\beta\left[ U_N(Y, \beta) \mid S \right] = 0.
\end{equation}
Indeed, the Polyads pseudo-score is the averaged sum of individual polyad scores:
\begin{equation*}
     U_N(Y,\beta) = \frac{1}{N_a} \sum_{\xi \in \Xi} \nabla_\beta l_\xi(Y|X, \beta) = \frac{1}{N_a} \sum_{\xi \in \Xi} \nabla_\beta l_\xi(\beta).
\end{equation*}
For fixed polyad $\xi$ and graph $y$, the individual score is the negative gradient of the log-likelihood conditional on the orbit:
\begin{equation*}
    -\nabla_\beta l_\xi(y|X, \beta) = \nabla_\beta \log \mathbb{P}_\beta(Y=y \mid X, Y \in \mathcal{O}_\xi(y)),
\end{equation*}
and its expectation over this specific orbit is identically zero: for all $\xi$, $\beta$, and $y$,
\begin{equation*}
    \mathbb{E}_\beta \big[ \nabla_\beta l_\xi(Y|X, \beta) \mid Y \in \mathcal{O}_\xi(y) \big] = \sum_{y' \in \mathcal{O}_\xi(y)} \nabla_\beta l_\xi(y'|X, \beta) \mathbb{P}_\beta(Y=y' \mid Y \in \mathcal{O}_\xi(y)) = 0.
\end{equation*}
As established in Proposition~\ref{prop:from:degrees:to:polyads}, polyad transformations preserve the network degrees: $\delta(y') = \delta(y)$ for all $y' \in \mathcal{O}_\xi(y)$. Furthermore, by Lemma~\ref{lem:equiv_polyads}, given a polyad $\xi$, the orbits form disjoint equivalence classes, i.e. the relation $y \sim_\xi y^\prime$ iff $y\in\cO_\xi(y^\prime)$ is an equivalence relation. Therefore, for any given degrees vector $s$, the set of all valid graphs $\mathcal{Y}_s = \{y : \delta(y) = s\}$ can be partitioned into a disjoint union of orbits associated with $\xi$. Let $\mathfrak{O}_{\xi, s}:= \mathcal{Y}_s/\sim_\xi$ denote this set of disjoint orbits (= equivalence classes) inside $\mathcal{Y}_s$. We can now apply the law of total expectation to evaluate the expectation of the polyad score conditional on the degrees $S=s$:
\begin{align*}
    \mathbb{E}_\beta \big[ \nabla_\beta l_\xi(Y|X, \beta) \mid S=s \big] &= \sum_{\mathcal{O} \in \mathfrak{O}_{\xi, s}} \mathbb{E}_\beta \big[ \nabla_\beta l_\xi(Y|X, \beta) \mid Y \in \mathcal{O} \big] \mathbb{P}_\beta(Y \in \mathcal{O} \mid S=s) \\
    &= \sum_{\mathcal{O} \in \mathfrak{O}_{\xi, s}} 0 \times  \mathbb{P}_\beta(Y \in \mathcal{O} \mid S=s) = 0.
\end{align*} 
Because this conditional centering holds for every individual polyad $\xi$, it also holds for their linear combination. Therefore, the Polyads estimating equation is centered given the sufficient statistics of the fixed effects:
\begin{equation*}
    \mathbb{E}_\beta \big[ U_N(Y, \beta) \mid S \big] = \frac{1}{N_a} \sum_{\xi \in \Xi} \mathbb{E}_\beta \big[ \nabla_\beta l_\xi(Y|X, \beta) \mid S \big] = 0.
\end{equation*}Next, we differentiate the latter inequality to show the generalized Information Matrix Equality $H_0 = \text{Cov}(U_N, S_c)$.

Differentiating both sides of \eqref{eq:critica_eq_polyads} with respect to $\beta$ under the integral sign yields:
\begin{equation*}
    \int \big[ \nabla_\beta U_N(y, \beta) \big] \bP_{{\beta}}^{Y \mid S=s}(y) dy + \int U_N(y, \beta) \big[ \nabla_\beta \bP_{\beta}^{Y \mid S=s}(y)\big]^\top dy = 0.
\end{equation*}
We can rewrite the gradient of the density using the score function: 
\begin{equation*}
    S_c(\beta) = \nabla_\beta \log \bP_{\beta}^{Y \mid S=s}(Y) = \frac{\nabla_\beta \bP_{\beta}^{Y \mid S=s}(Y) }{\bP_{\beta}^{Y \mid S=s}(Y) },
\end{equation*} 
meaning $\nabla_\beta \bP_{{\beta}}^{Y \mid S=s}(Y) = S_c({\beta}) \bP_{{\beta}}^{Y \mid S=s}(Y)$. Substituting this into the second integral gives at $\beta_\star$:
\begin{equation*}
    \mathbb{E}_{{\beta_\star}}[\nabla_\beta U_N({\beta_\star})] + \mathbb{E}_{{\beta_\star}}[U_N({\beta_\star}) S_c({\beta_\star})^\top] = 0.
\end{equation*}
Rearranging this identity, this shows that the expected Hessian of the Polyads loss function equals to the covariance between the Polyads score and the CMLE score:
\begin{equation*}
    H_0 = \mathbb{E}_{{\beta_\star}}[-\nabla_\beta U_N({\beta_\star})] = \text{Cov}_{{\beta_\star}}(U_N({\beta_\star}), S_c({\beta_\star})).
\end{equation*}

Finally, we consider the stacked random vector composed of the Polyads score and the CMLE score: $Z = \begin{pmatrix} U_N({\beta_\star}) \\ S_c({\beta_\star}) \end{pmatrix}$. 
The joint covariance matrix of $Z$ must be positive semi-definite (PSD). Using the properties established above, this block matrix is:
\begin{equation*}
    \text{Var}_{{\beta_\star}}(Z) = \begin{pmatrix} \text{Var}_{{\beta_\star}}(U_N) & \text{Cov}_{{\beta_\star}}(U_N, S_c) \\ \text{Cov}_{{\beta_\star}}(S_c, U_N) & \text{Var}_{{\beta_\star}}(S_c) \end{pmatrix} = \begin{pmatrix} V_0 & H_0 \\ H_0^\top & \mathcal{I}_{\text{CMLE}} \end{pmatrix} \succeq 0.
\end{equation*}
Because $\mathcal{I}_{\text{CMLE}}$ is positive definite, the block matrix inversion formulae of PSD matrices dictates that its Schur complement in $\text{Var}_{{\beta_\star}}(Z) $ must also be positive semi-definite:
\begin{equation*}
    V_0 - H_0 \mathcal{I}_{\text{CMLE}}^{-1} H_0^\top \succeq 0.
\end{equation*}
The conclusion follows by multiplying by $H_0^{-1}$ on both sides. \hfill 
\end{proof}

\subsection{Controlling the information loss from the Polyads estimator to the CMLE.}\label{sec:uppe_bound_info_loss}


To fully understand the relative efficiency of the Polyads estimator compared to the CMLE, we quantify the magnitude of the information loss in the case of a $1D$ structural parameter, i.e. for $p=1$ --- which is particularly relevant for gravity models where $X_{ij}$ is a distance between $i$ and $j$. The main result of this section links the variance of the two estimators (the Polyads estimator and the CMLE) via two parameters $\bar{m}_a$ and $\kappa_X$ introduced in \eqref{eq:eff_exp_loca_overlap_degree} and \eqref{eq:feature_concen_factor}.


\begin{theorem}[General 1D variance control]
    \label{theo:UB_var_poly}
The asymptotic variance of the Polyads estimator (denoted by $v_{\text{Polyads}} $) compared to the one of CMLE (denoted by $v_{\text{CMLE}}$) satisfies $v_{\text{CMLE}} \leq v_{\text{Polyads}} \le (\bar{m}_a \kappa_X) v_{\text{CMLE}}$.
\end{theorem}
\begin{proof}[Proof of Theorem~\ref{theo:UB_var_poly}] Inequality $v_{\text{CMLE}} \leq v_{\text{Polyads}}$ follows from Theorem~\ref{theo:info_loss_LB}. Let us now move to the variance control of the Polyads estimator. We recall that $N_a$ is  the expected number of active polyads and $U_N(\beta) = \frac{1}{N_a} \sum_{\xi \in \Xi} \nabla l_\xi(\beta)$. For $p=1$, the asymptotic variance is $v_{\text{Polyads}} = V_0 / H_0^2$, where $H_0 = \mathbb{E}_{{\beta_\star}}[-\nabla_\beta U_N({\beta_\star})]$ and $V_0 = \text{Var}_{{\beta_\star}}(U_N({\beta_\star}))$. The expected Hessian is the normalized sum over all possible polyads:
\begin{equation*}
    H_0 = \frac{1}{N_a} \sum_{\xi \in \Xi} \mathbb{E}_{{\beta_\star}}[V_\xi(Y)] \tilde{X}_\xi^2
\end{equation*} and the variance of the normalized score is
\begin{equation*}
    V_0 = \frac{1}{N_a^2} \mathbb{E}_{{\beta_\star}} \left[ \left( \sum_{\xi \in \Xi} \nabla l_\xi \right)^2 \right] = \frac{1}{N_a^2} \sum_{\xi \in \Xi} \sum_{\xi' \in \Xi} \mathbb{E}_{{\beta_\star}} \big[ \nabla l_\xi \nabla l_{\xi'} \big].
\end{equation*}

Next, using that edge fluctuations are conditionally independent, $\mathbb{E}_{{\beta_\star}}[\nabla l_\xi \nabla l_{\xi'}] = 0$ unless $\xi$ and $\xi'$ share at least one edge (denoted $\xi' \sim \xi$), the Cauchy-Schwarz inequality, $\mathbb{E}[AB] \le \frac{1}{2}\mathbb{E}[A^2 + B^2]$ and that the score $\nabla l_{\xi'}$ is zero unless $\xi'$ is active (denoted by the indicator $\mathbf{1}_{\xi' \in \Xi_a}$), we obtain
\begin{align*}
    \sum_{\xi \in \Xi} \sum_{\xi' \sim \xi} \mathbb{E}_{{\beta_\star}} \big[ \nabla l_\xi \nabla l_{\xi'} \big] &\le \sum_{\xi \in \Xi} \sum_{\xi' \sim \xi} \mathbb{E}_{{\beta_\star}} \left[ (\nabla l_\xi)^2 \mathbf{1}_{\xi' \in \Xi_a} \right] = \sum_{\xi \in \Xi} \mathbb{E}_{{\beta_\star}} \left[ (\nabla l_\xi)^2 \sum_{\xi' \sim \xi} \mathbf{1}_{\xi' \in \Xi_a} \right].
\end{align*}

The inner sum $\sum_{\xi' \sim \xi} \mathbf{1}_{\xi' \in \Xi_a}$ represents the total number of active polyads that share at least one edge with $\xi$, it is exactly given by $d_\xi(Y)$ and so
\begin{equation*}
    V_0 \le \frac{1}{N_a^2} \sum_{\xi \in \Xi} \mathbb{E}_{{\beta_\star}} \left[ d_\xi(Y) (\nabla l_\xi)^2 \right].
\end{equation*}

Using our definition of the effective expected local overlap degree $\bar{m}_a$, we bound the expectation:
\begin{equation*}
    V_0 \le \frac{ \bar{m}_a}{N_a^2} \sum_{\xi \in \Xi} \mathbb{E}_{{\beta_\star}} \left[ (\nabla l_\xi)^2 \right].
\end{equation*} We proved in Lemma~\ref{lem:derivatives}  that  $\mathbb{E}_{{\beta_\star}}[(\nabla l_\xi)^2] = \mathbb{E}_{{\beta_\star}}[V_\xi(Y)] \tilde{X}_\xi^2$. Thus, the sum exactly equals $N_a H_0$:
\begin{equation*}
    V_0 \le \frac{\bar{m}_a}{N_a^2} (N_a H_0) = \frac{ \bar{m}_a}{N_a} H_0.
\end{equation*}As a consequence, the  variance of the Polyads estimator $v_{\text{Polyads}}$ satisfies
\begin{equation*}
    v_{\text{Polyads}}= V_0 / H_0^2 \le \frac{\bar{m}_a}{N_a H_0}.
\end{equation*}

To relate this to the exact CMLE variance, we recall that the CMLE asymptotic variance is $v_{\text{CMLE}} = \mathcal{I}_{\text{CMLE}}^{-1}$, where $\mathcal{I}_{\text{CMLE}} = \sum_{\bi} \text{Var}(Y_\bi) \tilde{X}_\bi^2$. Hence, we obtain
\begin{equation*}
    v_{\text{Polyads}} \le  \bar{m}_a \left( \frac{\mathcal{I}_{\text{CMLE}}}{N_a H_0} \right) \mathcal{I}_{\text{CMLE}}^{-1}.
\end{equation*}Finally, because $N_a H_0 = \sum_{\xi \in \Xi} \mathbb{E}_{{\beta_\star}}[V_\xi(Y)] \tilde{X}_\xi^2$, the term in parentheses is exactly the feature concentration factor $\kappa_X$ introduced above. Noting that $\mathcal{I}_{\text{CMLE}}^{-1} = v_{\text{CMLE}}$, we obtain the final bound $v_{\text{Polyads}} \le (\bar{m}_a \kappa_X) v_{\text{CMLE}}$. \hfill

\end{proof}


The next two examples give an overall idea of the order of magnitude of the effective expected local overlap degree ($\bar{m}_a$) and the feature concentration factor ($\kappa_X$). The first example is in a sparse network and the second in a dense one.

\paragraph{Example in the sparse homophily design.} Consider a bipartite graph with $n_1 = n_2 = N$. The structural parameter $\beta \in \mathbb{R}$ captures pure homophily (the ``main diagonal'' effect), such that the dyad feature is $X_{ij} = \mathbf{1}_{i=j}$. We assume constant fixed effects $\alpha_i = \alpha$ and $\psi_j = 0$. We also choose $\alpha$ and $\beta$ such that the Poisson intensities are
\begin{equation}\label{eq:def_lambda_sparse}
    \lambda_{ij} = \exp(\beta \mathbf{1}_{i=j} + \alpha) = 
    \begin{cases} 
      \lambda_{\text{in}} = \exp(\beta + \alpha) = 1 & \text{if } i=j \\
      \lambda_{\text{out}} = \exp(\alpha) = \frac{2}{N-1} & \text{if } i \ne j 
   \end{cases}
\end{equation}so that $(Y_{ij})_{ij}$ is essentially a diagonal matrix (\textit{i.e.} a very sparse setting).

Let us first give an order of magnitude of the  feature concentration factor $\kappa_X$ defined as 
\begin{equation*}
    \kappa_X = \frac{\sum_{\bi} \text{Var}_{{\beta_\star}}(Y_\bi) \tilde{X}_\bi^2}{\sum_{\xi \in \Xi } \mathbb{E}_{{\beta_\star}}[V_\xi(Y)] \tilde{X}_\xi^2},
\end{equation*}where $V_\xi(y) = \bV_{{\beta_\star}}\left(m_\xi(Y) \mid Y\in\cO_\xi(y)\right)$. The numerator is the CMLE Fisher Information. We have $X = (\mathbf{1}_{i=j})_{i,j}\in\bR^{N^2}$ and $\lambda_{ij}$ can take only two values in \eqref{eq:def_lambda_sparse}. In this context, we can compute exactly the CMLE Fisher Information:
\begin{equation*}
    \mathcal{I}_{\text{CMLE}} = \sum_{i,j} \lambda_{ij} (\tilde{X}_{ij})^2 = \lambda_{\text{in}} N(1+a) + \lambda_{\text{out}} N(N-1)a^2  = 11 N
\end{equation*} where $a:= 2\lambda_{\text{in}}/((N-1)\lambda_{\text{out}} - \lambda_{\text{in}}) = 2$. The denominator in $\kappa_X$ is the Hessian of the Polyads estimator's loss function. Many of its terms equal zero because the polyads' features are for most of them zero. Indeed, let $\xi = \begin{pmatrix}
    j_1 & j_2\\
    j_1^\prime & j_2^\prime
\end{pmatrix}$ be a given tetrad (i.e. a polyads in a $D=2$-way network). The tetrad  (difference-in-differences) feature is:
\begin{align*}
    \tilde{X}_\xi & = X_{j_1j_2} - X_{j_1j_2^\prime} - X_{j_1^\prime j_2} + X_{j_1^\prime j_2^\prime} = \mathbf{1}_{j_1 = j_2} - \mathbf{1}_{j_1=j_2^\prime} - \mathbf{1}_{j_1^\prime = j_2} + \mathbf{1}_{j_1^\prime = j_2^\prime}.
\end{align*}
In particular, we observe that $\tilde{X}_\xi \ne 0$ \textit{if and only if} the tetrad intersects the main diagonal. Tetrads strictly off the diagonal have $\tilde{X}_\xi = 0$ and contribute zero to the total information.  

The conditional variance $V_\xi(Y)$ is also zero unless the polyad is active which mostly requires to have two diagonal edges in it, i.e. only tetrads of the shape $\xi = \begin{pmatrix}
    j_1 & j_1\\
    j_1^\prime & j_1^\prime
\end{pmatrix}$ (up to permutation) bring information. In that case, $\mathbb{E}_{{\beta_\star}}[V_\xi(Y)] \approx \bP[\xi \mbox{ is active }] \approx 1$. Since there are $N(N-1)$ tetrads of the shape $\xi = \begin{pmatrix}
    j_1 & j_1\\
    j_1^\prime & j_1^\prime
\end{pmatrix}$ the numerator is of the order of $N^2$ and so $\kappa_X$ is of the order of $1/N$.

Let us now handle the effective expected local overlap degree 
    \begin{equation*}
    \bar{m}_a := \max_{\xi \in \Xi} \frac{\mathbb{E}_{{\beta_\star}} \left[ d_\xi(Y) (\nabla l_\xi)^2 \right]}{\mathbb{E}_{{\beta_\star}} \left[ (\nabla l_\xi)^2 \right]}.
\end{equation*}Given the 'mostly diagonal shape' of $Y=(Y_{ij})$, the only polyads that are expected to be active are permutations of polyads of the form $\xi = \begin{pmatrix}
    j_1 & j_1\\
    j_1^\prime & j_1^\prime
\end{pmatrix}$. Considering such a polyad then $\xi^\prime = \begin{pmatrix}
    j_1 & j_1\\
    j_1^{\prime\prime} & j_1^{\prime\prime}
\end{pmatrix}$ is also likely to be active and share an edge with $\xi$.  There are $N-1$ other possible choices of $j_1^{\prime \prime}$ for which this holds. Hence, the number of active polyads that share an edge with $\xi$ is of the order of $N$, i.e. $d_\xi(Y)\approx N$. As a consequence, we will have $\bar{m}_a \approx N$. 

Finally, applying Theorem~\ref{theo:UB_var_poly} in this sparse homophily design setting provide an upper on the loss of information of the order of a constant: 
\begin{equation*}
    v_{\text{CMLE}} \leq v_{\text{Polyads}} \lesssim v_{\text{CMLE}}.
\end{equation*}

\paragraph{Example in a dense network.} Consider a bipartite graph with $n_1=n_2=N$ and a design with no homophily structure, where $X_{ij}$, $i,j\in[N]$, are i.i.d.\ standard Gaussian random variables. We take constant fixed effects $\alpha_i=\alpha=1$ and $\psi_j=0$, and set the structural parameter to
\begin{equation}\label{eq:def_beta_dense}
    \beta = \frac{1}{\sqrt{2\ln N}},
\end{equation}
so that the Poisson intensities are
\begin{equation}\label{eq:def_lambda_dense}
    \lambda_{ij} = \exp(\beta X_{ij} + 1).
\end{equation}
The scaling \eqref{eq:def_beta_dense} is chosen so that $\beta \max_{i,j}|X_{ij}| = O(1)$, since $\max_{i,j}|X_{ij}| \approx \sqrt{2\ln(N^2)} = \Theta(\sqrt{\ln N})$ for an array of i.i.d.\ Gaussians, while for a typical pair $(i,j)$ we have $\beta X_{ij} = o(1)$. In particular, $\lambda_{ij} = e \cdot (1+o(1))$ for the bulk of the $N^2$ edges, and all intensities cell are simultaneously bounded above and below by constants with probability tending to one: this is a genuinely dense network, in contrast with the sparse homophily design above.

Let us first evaluate the feature concentration factor $\kappa_X$. Because the fixed effects are constant, the row and column means of $(X_{ij})$ used to profile them out vanish at rate $1/\sqrt{N}$, so that $\tilde{X}_{ij} = X_{ij}(1+o(1))$ to leading order. The CMLE Fisher Information is therefore
\begin{equation*}
    \mathcal{I}_{\text{CMLE}} = \sum_{i,j} \lambda_{ij} \tilde{X}_{ij}^2 \approx e \sum_{i,j} X_{ij}^2 = \Theta(N^2),
\end{equation*}
using $\mathbb{E}[X_{ij}^2]=1$ and $\lambda_{ij}\approx e$ for the typical pair.

For the denominator, consider a tetrad $\xi = \begin{pmatrix} j_1 & j_2 \\ j_1' & j_2' \end{pmatrix}$. Its DiD feature is
\begin{equation*}
    \tilde{X}_\xi = X_{j_1 j_2} - X_{j_1 j_2'} - X_{j_1' j_2} + X_{j_1' j_2'},
\end{equation*}
a sum of four independent standard Gaussians, so $\tilde{X}_\xi \sim \mathcal{N}(0,4)$ and $\mathbb{E}[\tilde{X}_\xi^2] = 4$ regardless of $N$. Since the network is dense, all four corner intensities of a generic tetrad are $\Theta(1)$ (of order $e$), so each tetrad is active with probability bounded away from $0$ uniformly in $N$, and $\mathbb{E}_{\beta_\star}[V_\xi(Y)] = \Theta(1)$ for essentially all of the $\Theta(N^4)$ tetrads $\xi \in \Xi$. Hence
\begin{equation*}
    \sum_{\xi \in \Xi} \mathbb{E}_{\beta_\star}[V_\xi(Y)] \tilde{X}_\xi^2 = \Theta(N^4) \times \Theta(1) = \Theta(N^4),
\end{equation*}
and so, recalling the definition \eqref{eq:feature_concen_factor},
\begin{equation*}
    \kappa_X = \frac{\mathcal{I}_{\text{CMLE}}}{\sum_{\xi \in \Xi} \mathbb{E}_{\beta_\star}[V_\xi(Y)] \tilde{X}_\xi^2} = \frac{\Theta(N^2)}{\Theta(N^4)} = \Theta(N^{-2}).
\end{equation*}
In words, $\kappa_X$ is small in this dense design: there are far more tetrads ($\Theta(N^4)$) than edges ($\Theta(N^2)$), and each tetrad carries an $O(1)$ amount of information.

Let us now turn to the effective expected local overlap degree $\bar{m}_a$. Fix an active tetrad $\xi = \begin{pmatrix} j_1 & j_2 \\ j_1' & j_2' \end{pmatrix}$ and consider the tetrads sharing the edge $(j_1,j_2)$, i.e.\ of the form $\xi' = \begin{pmatrix} j_1 & j_2 \\ k & \ell \end{pmatrix}$ for $k \neq j_1$, $\ell \neq j_2$. There are $\Theta(N^2)$ such choices of $(k,\ell)$, and, because the network is dense, each of the corresponding tetrads is active with probability bounded away from zero, so the number of active tetrads sharing this single edge with $\xi$ is itself $\Theta(N^2)$. Summing over the (fixed number of) edges of $\xi$ does not change the order, so $d_\xi(Y) = \Theta(N^2)$ for essentially every active tetrad, and therefore
\begin{equation*}
    \bar{m}_a = \Theta(N^2),
\end{equation*}
matching the general remark that $\bar{m}_a$ approaches $\mathcal{O}(n_1 n_2)$ in dense networks.

Combining the two rates, we obtain
\begin{equation*}
    \bar{m}_a \kappa_X = \Theta(N^2) \times \Theta(N^{-2}) = \Theta(1),
\end{equation*}
so that, applying Theorem~\ref{theo:UB_var_poly}, our upper bound remains of the order of a constant even in this dense design:
\begin{equation*}
    v_{\text{CMLE}} \leq v_{\text{Polyads}} \lesssim v_{\text{CMLE}}.
\end{equation*}

\section{Health data application}
\label{app:health}

We use claims data from the French comprehensive \textit{Système National des Données de Santé} (SNDS) database.  The data contain information on each encounter between patients and physicians. 

In this appendix, we evaluate the reform by using doctor-level data only. In other words, we aggregate the visits at the doctor level for each month between January 2016 and October 2018 and do not consider patient-doctor dyads as in the main text.\footnote{Our observation period runs from 2016 to 2018, but we exclude November and December 2018 because of right-censoring. Indeed, the processing time for claim files may last up to 40 days, and the date of record may well fall after the end of 2018. April and May 2017 are also removed from the sample to neutralize anticipation effects.} 
We restrict our attention to doctors aged 30 to 79 practicing in mainland France (at the exclusion of overseas and Corsica).

\begin{table}[ht]\centering
\def\sym#1{\ifmmode^{#1}\else\(^{#1}\)\fi}
\caption{Effect of the fee increase on physician activity}
\label{tab:activity:jt}
\begin{tabular*}{1.0\hsize}{@{\hskip\tabcolsep\extracolsep\fill}l*{2}{c}}
\toprule
Specification                   &\multicolumn{2}{c}{PPML}\\
\midrule
Control group                    &\multicolumn{1}{c}{Direct access specialists}&\multicolumn{1}{c}{GPs in unregulated sector 2}\\
\midrule
Post $\times$ Treatment     &  0.104\sym{***}&      0.131\sym{***}\\
                    &   (0.00654)         &   (0.00671)         \\
\midrule
Physician FE & Yes& Yes\\
\addlinespace
Month-year FE& Yes& Yes\\
\midrule
Observations        &     2159008         &     2070592         \\
\addlinespace
(Pseudo) R$^2$                 &     0.754           &        0.740             \\
\bottomrule
\multicolumn{3}{l}{\footnotesize \textit{Source:} French SNDS claims data aggregated at the doctor-month level.}\\
\multicolumn{3}{l}{\footnotesize \textit{Sample:} GPs and direct access specialists, mainland France.}\\
\multicolumn{3}{l}{\footnotesize Period: January 2016-October 2018, at the exclusion of April-May 2017.}\\
\multicolumn{3}{l}{\footnotesize Treatment group: GPs in regulated sector 1.}\\
\multicolumn{3}{l}{\footnotesize Dependent variable: Number of visits per GP and per month.}\\
\multicolumn{3}{l}{\footnotesize N.B. The data include 0 activity at physician-month level.}\\
\multicolumn{3}{l}{\footnotesize \sym{*} \(p<0.1\), \sym{**} \(p<0.05\), \sym{***} \(p<0.01\)}\\
\end{tabular*}
\end{table}

To identify the effect of the reform, we perform  a difference-in-differences exercise in which fee-regulated GPs are the  treated group. We consider two control groups: direct access specialists (stomatologists,  ophthalmologists and gynecologists\footnote{In section~\ref{sec:experiments:real}, we exclude gynecologists because we seek to assess how gender homophily changed after the fee increase. Patients aged less than 18 are also excluded from the analysis conducted in section~\ref{sec:experiments:real}.}); free-billing GPs in the unregulated sector.\footnote{In January 2016, 8.6\% of GPs were allowed to charge fees in excess of the regulated level.} The outcome~$Y_{jt}$ is the number of medical visits by doctor~$j$ on month~$t$ and we assume that it follows a Poisson model:~$Y_{jt}\ove{\sim}{ind.}\mathcal{P}(\lambda_{jt})$ with
\begin{equation}
    \label{eq:DinD}
    \ln \lambda_{jt}=\beta \text{Post}_t\times T_j + \alpha_j+\gamma_t,
\end{equation}
where $\text{Post}$ is a dummy equal to 1 after May 2017,~$T_j$ is a binary variable that accounts for the treatment group status, and~$(\alpha,\gamma)$ are physician and month-year fixed effects.  

Results reported in Table~\ref{tab:activity:jt} show that this reform led physicians to increase their activity by 11--14\%, depending on the specification and the control group considered. \cite{wilner2025physician} further show that the effect is more pronounced for young doctors, and they also find  suggestive evidence that access to healthcare improved in the sense that doctors actually admitted new patients following the fee increase. We did not use the Polyads estimator for the results in Table~\ref{tab:activity:jt} because the network density exceeds $90\%$, a regime where standard PPML is computationally more efficient.


\section{Experiments on the four-way case}
\label{sec:four_way}

To demonstrate the scalability and robustness of our Polyads estimator in a more demanding multi-way settings, we extend our simulation framework to a four-way case. This setup replicates an empirical scenario where an analyst tracks flows between origins and destinations across multiple sectors over time, facing a severe incidental parameter problem along four separate axes simultaneously. We let connections be driven mainly by homophily.

We consider a network structured around four indices: importing (origin) cities $i$, exporting (destination) cities $j$, sectors $k$, and time periods $t$. The dimensions are as $i \in \{1, \dots, n_{\text{cities}}\}$, $j \in \{1, \dots, n_{\text{cities}}\}$, $k \in \{1, \dots, n_{\text{sectors}}\}$, and $t \in \{1, \dots, n_{\text{ts}}\}$. The covariate $X_{ijkt}$ is structured to reflect localized, policy-induced interactions across specific sectors and time horizons. It is constructed via:
\[
X_{ijkt} = \mathds{1}\{|i - j| \le r\} \times \mathds{1}\{k \ge s_0\} \times \mathds{1}\{t \ge t_0\} \times U_{ijkt},
\]
where $\mathds{1}\{\cdot\}$ is the indicator function, $r$ represents a geographic radius neighborhood, $s_0$ and $t_0$ represent structural thresholds for the sector and temporal policy shocks respectively, and $U_{ijkt} {\ove{\sim}{i.i.d.}} \mathcal{U}(0, 1)$ introduces independent localized variation. The outcome $Y_{ijkt}$ is generated from a Poisson distribution:
\[
Y_{ijkt} \sim \mathcal{P}(\lambda_{ijkt}),
\]
where the log-intensity incorporates the true treatment effect $\beta_\star$, a baseline density constant $c$, and four distinct sets of high-dimensional fixed effects that capture overlapping three-way interactions:
\[
\ln \lambda_{ijkt} = \beta_\star X_{ijkt} + c + \theta^{(1)}_{jkt} + \theta^{(2)}_{ikt} + \theta^{(3)}_{ijt} + \theta^{(4)}_{ijk}.
\]
The unobserved fixed effects capture comprehensive group-level heterogeneity and are drawn independently from a standard normal distribution:
\[
\theta^{(1)}_{jkt}, \theta^{(2)}_{ikt}, \theta^{(3)}_{ijt}, \theta^{(4)}_{ijk} \sim \mathcal{N}(0, 1).
\]

We evaluate the performance of the estimators by varying the overall network size via the number of cities, keeping the time and sector dimensions fixed to simulate a typical sparse short-panel. Specifically, the number of origin/destination cities ($n_{\text{cities}}$) varies across $\{200, 300, 400\}$, while the number of sectors ($n_{\text{sectors}}$) is fixed at $3$ and the number of time periods ($n_{\text{ts}}$) is fixed at $5$. The structural constraints for the policy shock feature a sector threshold of $s_0 = 1$, a temporal threshold of $t_0 = 3$, and a localized geographic radius neighborhood of $r = 5$. The true parameter value is set to $\beta_\star = 20$, and we evaluate each configuration over $600$ independent Monte Carlo replications. The baseline intercept constant $c = -7$ is explicitly chosen to ensure the network remains highly sparse as $n_{\text{cities}}$ grows, artificially aggravating the incidental parameter bias for traditional maximum likelihood frameworks.

\begin{figure}[h!]
    \centering
    \includegraphics[width=\linewidth]{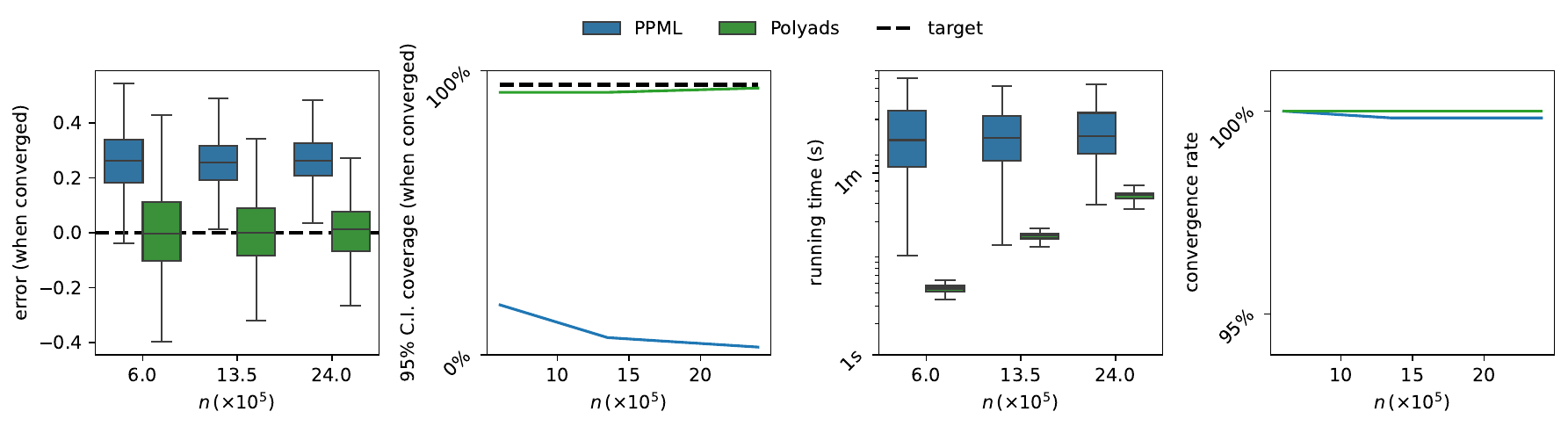}
    \caption{Comparison between PPML and Polyads on a four-way case.}
    \label{fig:4w}
\end{figure}

Figure \ref{fig:4w} displays the results. Notice that the analytical debias procedure from \cite{weidner2021bias} is not present there since it is available only for $D=3$. Similarly to the experiments with $D=3$ we observe a clear incidental parameter bias and poor coverages for PPML.

\section{Extensive-margin binary outcomes}
\label{sec:binary}

The Polyads estimator extends naturally from weighted networks to extensive-margin models in which only the existence of a link is observed. This extension follows exactly the construction of Section~\ref{sec:beta_estimation}. The generalized degrees remain sufficient statistics for the fixed effects, the polyad transformations remain degree-preserving, the conditional likelihood has the same form, and the resulting estimator is unchanged. The only modification concerns the construction of the polyad orbits. Since outcomes now belong to $\{0,1\}$ rather than $\mathbb N$, the admissible polyad transformations differ from the count case, leading to a different notion of active polyads. Once these active polyads have been identified, all subsequent derivations, including the loss function, optimization problem, and computational implementation, remain unchanged. In the case $D=2$, this estimator is closely related to the one developed by \cite{graham2017econometric}, with ours working for bipartite graphs. For $D=3$, our algorithm provides a fast implementation of the hexad logit studied by \cite{Muris_Pakel_WP_2025}.

\paragraph{Bernoulli model and sufficient statistics.}

Replace Assumption~1 by the Bernoulli specification
\begin{equation}
Y_\bi\mid X\sim\mathrm{Bernoulli}(\lambda_\bi),
\qquad
\log\frac{\lambda_\bi}{1-\lambda_\bi}
=
\beta_\star^\top X_\bi
+
\sum_{g\in\mathcal G}\theta^g_{g(\bi)}.
\label{eq:bernoulli}
\end{equation}

Exactly as in Section~\ref{sec:degree:sufficient:stat}, the generalized degrees
\[
\delta(y)=\left(\delta^g(y):g\in\mathcal G\right)
\]
remain sufficient statistics for the fixed effects, since the Bernoulli log-likelihood is still linear in the fixed effects. Consequently, conditioning on $\delta(Y)$ removes all nuisance parameters exactly as in
\eqref{eq:conditional_likelihood_degree}. Moreover, Proposition~\ref{prop:from:degrees:to:polyads} continues to hold without modification: the polyad transformations preserve all generalized degrees and, under the fixed effects structure $\mathcal G_{\max}$, generate every graph sharing the same degrees.

\paragraph{Binary polyad orbits.}

The polyad transformation
\[
T_\xi(y)=y+s_\xi,
\]
defined in Section~\ref{sec:degree:sufficient:stat}, is unchanged. The difference lies only in the set of admissible transformed graphs.

For count data, the admissible transformations are determined by the requirement that every transformed edge remains nonnegative (recall definition \eqref{eq:original-Mm}). For Bernoulli outcomes, the transformed graph must instead remain in the hypercube $\{0,1\}^{\mathcal I}$. Up to the permutations discussed in Lemma \ref{lem:tools_polyads}, positive-sign edges may therefore only be decremented when they all equal one, whereas negative-sign edges may only be incremented when they all equal zero. Accordingly, we let
\begin{equation}
m_\xi(y)
=
\bigwedge_{i:s_\xi(i)=+1}y_i \qquad \mbox{ and }
\qquad
M_\xi(y)
=
\bigvee_{i:s_\xi(i)=-1}y_i.
\label{eq:mM-binary}
\end{equation}

An informative (active) polyad therefore satisfies (again, up to a permutation discussed in Lemma \ref{lem:tools_polyads})
\begin{equation}
m_\xi(y)=1 \qquad \mbox{ and }
\qquad
M_\xi(y)=0,
\label{eq:active-binary}
\end{equation}
in which case the orbit consists of exactly two feasible graphs related by a single polyad transformation. Geometrically, every edge with positive sign must be present whereas every edge with negative sign must be absent.

\paragraph{Conditional likelihood and estimator.}

Once the active polyads have been identified, every derivation of
Section~\ref{sub:a_classification_problem_on_polyad_and_the_associated_estimator_of_}
remains valid. Indeed, the conditional likelihood over an orbit is still obtained by conditioning on
$Y\in\mathcal O_\xi(y)$, giving exactly the same loss\footnote{Note that since $Y_\bi$ takes values on $\{0,1\}$ for all $\bi$, the extra term $
\sum_{i\in\mathcal E(\xi)}
\log
\frac{Y_i!}{Y_i^{\,r}!}$ present in the original loss cancels to zero, so, one can keep the exact same definition of $\ell_\xi$.}
\[
\ell_\xi(Y\mid X,\beta)
=
\log
\left(
\sum_{r=-m_\xi(Y)}^{M_\xi(Y)}
\exp
\left\{
r\beta^\top\widetilde X_\xi
\right\}
\right),
\]
and therefore the same objective function and same Polyads estimator
\[
\widehat{\mathcal L}_\Xi(Y\mid X,\beta)
=
\sum_{\xi\in\Xi_a}
\ell_\xi(Y\mid X,\beta),
\qquad
\widehat\beta_\Xi
=
\arg\min_\beta
\widehat{\mathcal L}_\Xi(Y\mid X,\beta).
\]

Every active orbit therefore reduces to a binary conditional logit whose linear predictor is simply the generalized difference-in-differences feature
$\widetilde X_\xi$. Lemma~\ref{lem:derivatives} remains valid without modification, implying that each polyad loss is convex and that its gradient and Hessian are still given by
\eqref{eq:foc}--\eqref{eq:soc}. Although outside of the scope of this paper, all theoretical results established in Section~4 seem to extend immediately to the extensive-margin setting.

\paragraph{Computational implementation.}

The algorithms of Section~5 require only one modification. Candidate polyads are generated exactly as before using pairs of realized edges and the same tie-breaking rule defining $\Xi_a^\star$. The sole difference lies in the activity test. For count data, a candidate polyad is active whenever its orbit contains more than one nonnegative graph, which is determined through the extrema defining $m_\xi$ and $M_\xi$. In the Bernoulli model, the candidate is retained only if every positive-sign edge equals one and every negative-sign edge equals zero, namely if \eqref{eq:active-binary} holds. No other part of the algorithm changes. In particular, the computational complexity remains exactly as in the count model.

\begin{figure}[h!]
    \centering
    \includegraphics[width=\linewidth]{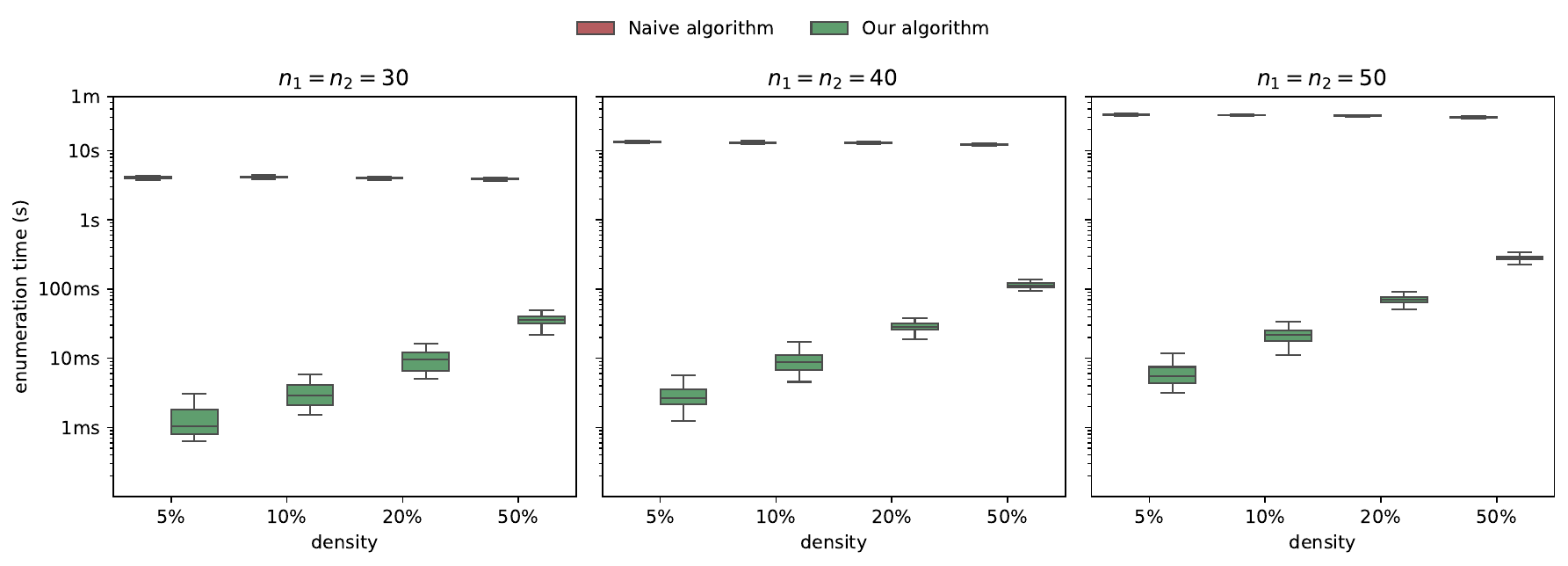}
    \caption{Comparison, in the Bernoulli case, between the naive algorithm for listing active polyads and the method we propose.}
    \label{fig:extensive-enum-times}
\end{figure}

\paragraph{Computational performance.}
The discussion of Section~5 suggests that the computational advantage of the sparse implementation should carry over unchanged to the extensive-margin setting, since only the activity test differs from the count model. We verify this claim through a Monte Carlo experiment under the Bernoulli model~\eqref{eq:bernoulli} with two-way fixed effects ($\mathcal G=\mathcal G_{\max}$).

We let $n_1=n_2\in\{30,40,50\}$ and target densities $|E|/n \in\{0.05,0.10,0.20,0.50\}$. We perform $100$ replications and compare two implementations for constructing the canonical active set $\widehat\Xi_a^\star$. The first is a \emph{naive} implementation that enumerates every ordered polyads in $\Xi$. The second is the sparse implementation of Algorithm~1, which generates candidates only from pairs of realized edges and applies exactly the same activity test. Both implementations return the same canonical collection of active polyads and therefore produce identical values of the objective function and identical estimates of $\widehat\beta_\Xi$. The experiment thus isolates the computational cost of polyads enumeration.

Figure~\ref{fig:extensive-enum-times} reports wall-clock enumeration times over the $100$ replications. Three conclusions emerge. First, the running time of the naive implementation depends almost exclusively on the total number of polyads $|\Xi|$ and is therefore nearly insensitive to network density. Second, the sparse implementation scales with the number of realized edges, yielding running times that increase with the density. Finally, the computational gains are higher for sparse networks, where the sparse implementation can be 3 orders of magnitude faster. Even for relatively dense networks, sparse enumeration remains considerably faster.

\newpage 

\begin{center}
    {\bf \large SUPPLEMENTARY MATERIAL} 
\end{center}

\bigskip

\section{Proofs of main results}

\subsection{Notation}
In several places in the proof of consistency or asymptotic normality of the Polyads estimator, we use a notation that involves a random polyad and or a random index.  In that case, we denote by $\bxi$ (in bold) a random variable with values in the set of all polyads that is uniformly distributed over this set --- whereas a deterministic polyad is denoted by $\xi$. We also denote by $\bar \bi$ a random variable that is uniformly distributed over $\cI$, the set of all edges. We denote by $\bbE_{\beta_\star}$ (resp. $\bbP_{\beta_\star}$) the expectation (resp. the probability) w.r.t. both the random polyads/edge and the $Y_\bi$'s. For instance, we will use several time the following quantity  
\begin{align*}
    \Delta_{q,N}(X) &= \bbE_{\beta_\star}\left[\inr{\nabla \ell_{\bxi}(\beta_\star),c}\inr{\nabla \ell_{\bxi^\prime}(\beta_\star),c}I\left(|\cE(\bxi) \cap \cE(\bxi^\prime)|=q\right)|X\right]\\
    &  = \frac{1}{N^2}\sum_{\xi, \xi^\prime}\bE_{\beta_\star}\left[\inr{\nabla \ell_{\xi}(\beta_\star),c}\inr{\nabla \ell_{\xi^\prime}(\beta_\star),c}I\left(|\cE(\xi) \cap \cE(\xi^\prime)|=q\right)|X\right]
\end{align*}where $\bxi$ and $\bxi^\prime$ are iid uniformly distributed over $\Xi$.    

We will also use $B_2$ as the unit Euclidean ball of $\bR^d$ or $\bR^p$ depending on the context and $B_2(\theta, \eps) = \theta+\eps B_2$. The intensity of the Poisson variable $Y_\bi| X_\bi$ under model assumption~\ref{eq:Poisson_model} is denoted by
\begin{equation*}
    \lambda_\bi(\beta_\star, \theta^\sG) = \exp\left(\beta_\star^\top X_\bi + \sum_{g\in\cG}\theta_{g(\bi)}^g\right),
\end{equation*}the expected number of active polyads (conditionally on $X$) is denoted by $N_a$. The operator norm of a matrice (ie its largest singular value) is denoted by $\norm{\cdot}_{op}$. To simplify notation we drop the $k$ index everywhere.

\subsection{Consistency of $\widehat{\beta}_\Xi$ in the Poisson model} 
\label{sub:consistency_of_}

The proof of the consistency result uses the convexity of the loss function as a key ingredient. It is based on a slight modification of Theorem 2.7 of \cite{newey1994large} that requires to revisit some classical results in convex analysis (such as those from Chapter~10 in \cite{rockafellar-1970a}) as well as some results in asymptotic statistics from \cite{newey1994large} and \cite{Andersen1982}. All these revisited versions of these classical results may be found in Section~\ref{sec:appendix_a_proof_of_lemma_lem:suf_beta}. Theorem~\ref{theo:consistency_tetrads} is a consequence of Lemma~\ref{lem:suf_beta} below whose proof may be found also in Section~\ref{sec:appendix_a_proof_of_lemma_lem:suf_beta}.

\begin{lemma}\label{lem:suf_beta} Let \( \Theta \) be a non-empty open and convex set in $\bR^d$.  Let $(\hat Q_n)_n$ be a sequence of random convex functions and $(Q_n)_n$ be a sequence of (deterministic) convex functions all defined on $\Theta$.  We assume that there exists $\theta_0\in\Theta$ and $\eps_0>0$ such that $\theta_0+3\eps_0 B_2\subset \Theta$ and the following holds:
\begin{enumerate}
    \item for $n$ large enough,  $\hat Q_n$ is uniquely minimized over $\Theta$ by $\hat \theta_n$;
    \item for all $0<\eps\leq \eps_0$ there exists $\eta>0$ and $n_0$ such that for all $n\geq n_0$ and all $\theta\in\Theta$, if $\norm{\theta - \theta_0}_2= \eps$ then $Q_n(\theta) - Q_n(\theta_0)\geq \eta$;
    \item there exists  $\theta_1\in \theta_0 + 3 \eps_0 B_2$ and $L_1\in\bR$ such that $\inf_n Q_n(\theta_1)\geq L_1$ and  for every $\theta\in \theta_0+3 \eps_0 B_2$, there exists $L>0$ such that $\sup_{n}Q_n(\theta)\leq L$;
    \item for all \( \theta \in \theta_0+3 \eps_0 B_2 \), as $n$ tends to infinity, \( \hat Q_n(\theta) - Q_n(\theta) \overset{p}{\to}  0 \).
\end{enumerate}
Then $ \hat{\theta}_n \overset{p}{\to} \theta_0$.
\end{lemma}

They are two main advantages of Lemma~\ref{lem:suf_beta}: first, as in Theorem~2.7 from \cite{newey1994large}, we don't need to assume that $\Theta$ is compact (that will be useful for us since we are minimizing over all $\beta\in\bR^p$ to define the Polyads estimator); second, we don't need to assume that the sequence of  risk functions $(Q_N)_N$ has a limit or that it is the same function for all $N$. The latter assumption is problematic when the fixed effects are considered as parameters (or when they are random variables and we work conditionally to them); the same remark holds for the co-variable vectors $(X_{\bi})_\bi$. This is our main motivation to use Lemma~\ref{lem:suf_beta} in place of the classical Theorem~2.7 from \cite{newey1994large} where a limiting function is assumed to exists.

\begin{proof}[Proof of Theorem~\ref{theo:consistency_tetrads}]
 To prove Theorem~\ref{theo:consistency_tetrads}, we apply Lemma~\ref{lem:suf_beta}  to $\Theta = \bR^d$, $n = N$ the number of polyads, for all $\beta\in\bR^d$,   $\hat Q_N(\beta) = N_a^{-1} L(Y|X, \beta)$ and $Q_N(\beta) = N_a^{-1} \bE_{\beta_\star} [L(Y|X, \beta)|X]$ so that $\theta_0= \beta_\star$ and $\eps_0=1/3$. Under the assumptions from Theorem~\ref{theo:consistency_tetrads}, the first point from Lemma~\ref{lem:suf_beta} is satisfied. The second point follows from strong convexity of $Q_N$ at $\beta_\star$ due to Assumption~\ref{ass:consistency} and the computation of the Hessian in Lemma~\ref{lem:derivatives}. It therefore only remains to show the third and fourth items of Lemma~\ref{lem:suf_beta}.

We start with the third item from Lemma~\ref{lem:suf_beta}. We introduce some tools that will be useful to check this condition. We recall that the cross entropy between two probability measures $P,Q$ on $\bN^\cI$ is defined as 
$$
H(P,Q) = -\bE_P \ln Q = - \sum_{y\in\bN^\cI} P(\{y\}) \ln Q(\{y\}).
$$ We know that $Q\to H(P,Q)$ is minimal at $Q=P$ and that  $H(P,P)\geq0$. We also have $H(P,Q) - H(P,P)=KL(P,Q)$ where
$$
KL(P,Q)=\sum_{y} P(\{y\}) \ln\left(\frac{P(\{y\})}{Q(\{y\})}\right)
$$is the Kulback-Leiber divergence between $P$ and $Q$. In particular, if there exists $C>0$ such that $P(\{y\})\leq C Q(\{y\}), \forall y$ then $0\leq KL(P,Q) \leq \ln(C)$. Next, given $\xi\in\Xi$, $\beta\in\bR^p$ and $z\in\bN^{\cI}$, we define a probability measure on $\bN^\cI$ by
\begin{equation*}
    \bP_{\beta, \cO_\xi(z)}(\{y\}|X) = \left\{ 
    \begin{array}{cc}
        \bP_\beta[Y=y|Y\in\cO_\xi(z), X] & \mbox{ if } y\in\cO_\xi(z)\\
        0 & \mbox{ otherwise.}
    \end{array}
    \right.
\end{equation*}It follows from Lemma~\ref{lem:equiv_polyads} from Section~\ref{sec:appendix_a} that $y\in\cO_\xi(z)$ iff $\cO_\xi(y) = \cO_\xi(z)$. As a consequence, when $y\in\cO_\xi(z)$, we also have
\begin{equation*}
    -\ln \bP_{\beta, \cO_\xi(z)}(\{y\}|X) = -\ln\left[\frac{\bP_\beta[Y=y|X]}{\bP_\beta[Y\in\cO_\xi(y)|X]}\right] = \ell_\xi(y|X, \beta).
\end{equation*}

\begin{proposition}\label{prop:check_assum_consistency_risk_fct}
     We have for all $\beta\in\bR^p$, 
    \begin{equation*}
        Q_N(\beta) = \frac{1}{N_a} \sum_{\xi\in\Xi} \bE_{\beta_\star}\left[ H(\bP_{\beta_\star, \cO_\xi(Y)}(\cdot|X), \bP_{\beta, \cO_\xi(Y)}(\cdot|X))|X\right].
    \end{equation*}
    As a consequence, $Q_N(\beta_\star)\geq0$ and so one can take $L_1=0$ in \textit{item~3} from Lemma~\ref{lem:suf_beta} without further assumption. We also have for all $\beta\in\bR^p$,
    \begin{equation*}
        Q_N(\beta)\leq c_0\left( 1+ \bar\lambda\right)\left( 1 + L_N\norm{\beta - \beta_\star}_2\right).
    \end{equation*} As a consequence, one can take $L$ as some absolute in \textit{item~3} from Lemma~\ref{lem:suf_beta} as long as $L_N:=\max_\xi\norm{\tilde X_\xi}$ and the fixed effects are uniformly bounded from above as $k\to+\infty$.
\end{proposition}

\begin{proof}[Proof of Proposition~\ref{prop:check_assum_consistency_risk_fct}]

Let $\beta\in\bR^p$ and $\xi$ be a polyads. We have
    \begin{align*}
        \bE_{\beta_\star}\ell_\xi(Y|X, \beta) = -\sum_{y\in\bN^\cI} \bP_{\beta_\star}[Y=y] \ln\left[\frac{\bP_\beta[Y=y|X]}{\bP_\beta[Y\in\cO_\xi(y)|X]}\right].
    \end{align*}As a consequence, if we denote by $\mathfrak{O}_\xi$ the set of all orbits associated with polyads $\xi$, ie $\mathfrak{O}_\xi=\{\cO_\xi(z):z\in\bN^{\cI}\}$ and if we apply Lemma~\ref{lem:equiv_polyads} from Section~\ref{sec:appendix_a} to get that $y\in\cO$ iff $\cO_\xi(y) = \cO$ for all $\cO\in\mathfrak{O}$, we obtain
\begin{align*}
  &\bE_{\beta_\star}\ell_\xi(Y|X, \beta) = -\sum_{\cO\in\mathfrak{O}_\xi} \sum_{y\in\cO} \bP_{\beta_\star}[Y=y|X] \ln\left[\frac{\bP_\beta[Y=y|X]}{\bP_\beta[Y\in\cO |X]}\right]\\ 
  &=  -\sum_{\cO\in\mathfrak{O}_\xi} \bP_{\beta_\star}[Y\in\cO|X] \sum_{y\in\cO} \frac{\bP_{\beta_\star}[Y=y|X]}{\bP_{\beta_\star}[Y\in\cO|X]} \ln\left[\frac{\bP_\beta[Y=y|X]}{\bP_\beta[Y\in\cO |X]}\right]\\
  &  = \sum_{\cO\in\mathfrak{O}_\xi} \bP_{\beta_\star}[Y\in\cO|X] \sum_{y\in\cO} H\left(\bP_{\beta_\star, \cO}(\cdot|X), \bP_{\beta, \cO}(\cdot|X)\right)\\ 
  &= \sum_{\cO\in\mathfrak{O}_\xi} \sum_{y\in\cO} \bP_{\beta_\star}[Y = y|X]  H\left(\bP_{\beta_\star, \cO_\xi(y)}(\cdot|X), \bP_{\beta, \cO_\xi(y)}(\cdot|X)\right)\\
  & = \bE_{\beta_\star} H\left(\bP_{\beta_\star, \cO_\xi(Y)}(\cdot|X), \bP_{\beta, \cO_\xi(Y)}(\cdot|X)\right)
\end{align*}where Lemma~\ref{lem:equiv_polyads} from Section~\ref{sec:appendix_a} has been used in the last but one equality.

It follows from the definition of the KL-divergence based on the cross entropy recalled above that
\begin{equation*}
        Q_N(\beta) - Q_N(\beta_\star) = \frac{1}{N_a} \sum_{\xi\in\Xi} \bE_{\beta_\star}\left[ KL(\bP_{\beta_\star, \cO_\xi(Y)}(\cdot|X), \bP_{\beta, \cO_\xi(Y)}(\cdot|X))|X\right].
    \end{equation*} Next, for all $y\in\bN^\cI$ and all $\cO\in\mathfrak{O}_\xi$ such that $y\in\cO$, we have
\begin{equation*}
 \frac{\bP_{\beta_\star, \cO}(y|X)}{\bP_{\beta, \cO}(y|X)}= \frac{\bP_{\beta_\star}[Y=y]}{\bP_{\beta}[Y=y]}\frac{\bP_{\beta}[Y\in\cO]}{\bP_{\beta_\star}[Y\in\cO]} = \frac{\sum_{r=-m}^M \exp(r\inr{\tilde X_\xi, \beta_\star}) \prod_{\bi}(y_\bi^r!)^{-1}}{\sum_{r=-m}^M \exp(r\inr{\tilde X_\xi, \beta}) \prod_{\bi}(y_\bi^r!)^{-1}}\leq \max_{r=-m,\ldots, M} \exp\left(r\inr{\tilde X_\xi, \beta_\star - \beta}\right)
\end{equation*}where $m = m_\xi(y)$ and $M=M_\xi(y)$. As a consequence, for $m = m_\xi(Y)$ and $M=M_\xi(Y)$, it follows from Lemma~\ref{lem:moment_min_Poisson} that
\begin{equation*}
        Q_N(\beta) - Q_N(\beta_\star) \leq \frac{1}{N_a} \sum_{\xi\in\Xi} \bE_{\beta_\star}\left[(m+M)|X\right] |\inr{\tilde X_\xi, \beta_\star - \beta}| \leq c_0(\bar{\lambda} + 1) L_N \norm{\beta_\star - \beta}_2
    \end{equation*}where $L_N = \max_\xi \norm{\tilde X_\xi}_2$. Hence, it only remains to show that $\sup_{N}Q_N(\beta_\star)$ is finite. We have that 
  \begin{equation*}
        Q_N(\beta_\star) = \frac{1}{N_a} \sum_{\xi} \bE_{\beta_\star} \left[Ent(\bP_{\beta_\star, \cO_\xi(Y)}(\cdot|X))|X\right]
    \end{equation*} where $Ent(P) = H(P,P)$ is the entropy of a probability distribution $P$. We know that the maximal entropy is achieved by the uniform distribution and that it is equal to $\ln(n)$ for distributions supported over a set of cardinality $n$. In our case, the probability distributions $\bP_{\beta, \cO_\xi(z)}(\cdot|X)$ are supported on $\cO_\xi(z)$, hence,  for almost all $Y$, we have 
    $$
    Ent(\bP_{\beta_\star, \cO_\xi(Y)}(\cdot|X))\leq \ln |\cO_\xi(Y)| = \ln (m_\xi(Y)+M_\xi(Y)).
    $$Therefore, it follows from Lemma~\ref{lem:moment_min_Poisson} that there exists an absolute constant such that
    \begin{equation*}
      Q_N(\beta_\star) \leq \frac{1}{N_a} \sum_{\xi}  \bE_{\beta_\star} \left[\ln (m_\xi(Y)+M_\xi(Y))|X\right]\leq c_0(\bar{\lambda} + 1).
    \end{equation*} 
\end{proof}

Let us now turn to the fourth items of Lemma~\ref{lem:suf_beta},  which is the pointwise convergence in probability of $\hat Q_N-Q_N$ to $0$ for $\beta \in \beta_\star + B_2$. Let $\beta\in\bR^d$ be such that $\norm{\beta_\star - \beta}_2\leq1$.  By Chebyshev's inequality we only need to show that $N_a^{-2}\bV_{\beta_\star}(L(Y|X,\beta)|X)\to 0$ as $N$ tends to infinity. Since, two polyads $\xi$ and $\xi^\prime$ that do not have an edge in common are independent we obtain:
\begin{align*}
 \bV_{\beta_\star}(L(Y|X,\beta)|X)  = \sum_{\substack{\xi, \xi^\prime\in\Xi\\\cE(\xi)\cap \cE(\xi^\prime)\neq \emptyset}}\bE_{\beta_\star}\left[(\ell_\xi(Y | X, \beta)-\bE_{\beta_\star}\ell_\xi(Y | X, \beta))(\ell_{\xi^\prime}(Y | X, \beta)-\bE_{\beta_\star}\ell_{\xi^\prime}(Y | X, \beta)) |X\right]. 
\end{align*}

It follows from Lemma~\ref{lem:sampling_device_for_consistency} that there exists some constant $C_1>0$  such that  for all $\xi, \xi^\prime\in\Xi$, we have
$$
\bE_{\beta_\star}\left[(\ell_\xi(Y | X, \beta)-\bE_{\beta_\star}\ell_\xi(Y | X, \beta))(\ell_{\xi^\prime}(Y | X, \beta)-\bE_{\beta_\star}\ell_{\xi^\prime}(Y | X, \beta)) |X\right] \leq C_1 \bP_{\beta_\star}\left[\xi \mbox{ and  } \xi^\prime  \mbox{ are both active}|X\right].
$$ According to Lemma~\ref{lem:sampling_device_for_consistency} we can choose
\begin{equation}\label{eq:def_C1}
    C_1 = c_0 \left(1+L_N^2(\norm{\beta_\star}_2^2+1) + \ln^2(\bar \lambda+1)\right)(\bar\lambda+1)^2.
\end{equation}where $L_N = \max_\xi \norm{\tilde X_\xi}_2$ and  
$$
\bar \lambda = \max_{\bi\in\cI}\left(\exp\left(\inr{X_\bi, \beta_\star} + L_N + \sum_{g\in\cG}\theta^g_{g(\bi)}\right)\right)
$$(note that we used $\norm{\beta_\star - \beta}_2\leq1$ to define $\bar \lambda $).

As a consequence, we get 
\begin{align}\label{eq:control_var_consistency_1}
  \nonumber  &\bE_{\beta_\star}\left[(\hat Q_N(\beta) - Q_N(\beta))^2|X\right] \leq \frac{C_1}{N_a^2} \sum_{\substack{\xi, \xi^\prime\in\Xi\\\cE(\xi)\cap \cE(\xi^\prime)\neq \emptyset}}\bP_{\beta_\star}\left[\xi \mbox{ and  } \xi^\prime  \mbox{ are both active}|X\right]\\
  & \leq C_1\frac{ N^2\bbP_{\beta_\star}[\bxi \text{ and  } \bxi^\prime \text{ are active and they share at least one edge in common }] }{N_a^2}\\
  &\leq C_1\frac{ \bbP_{\beta_\star}[\bxi \text{ and  } \bxi^\prime \text{ are active and they share at least one edge in common }] }{\bbP_{\beta_\star}[\bxi \text{ is active}]^2}
  \end{align}then we conclude using Assumption~\ref{ass:first_order_geometry}.

\end{proof}

\subsection{Asymptotic normality of the Polyads estimator} 
\label{sec:asymptotic_normality_of_the_polyads_estimator}
In this section we prove the asymptotic normality property of the Polyads estimator as stated in Theorem~\ref{theo:CLT_polyads_estim_Graham}. Classical results such as the one in \cite{newey1994large} apply when the parameter space is compact and when there exists some limit risk function; two assumptions we want to avoid. As a consequence, we first state a general asymptotic normality result and then we apply it to obtain  Theorem~\ref{theo:CLT_polyads_estim_Graham}. 




%


\subsubsection{Main theorem for asymptotic normality under convexity assumption} 
\label{sub:main_theorem_for_asymptotic_normality_under_convexity_assumption}
The following result is adapted from Victor-Emmanuel Brunel lecture notes \cite{victor} who cites \cite{MR1026303} and \cite{MR1186263} as classical references for asymptotic statistics under convexity. However, all the later results use a loss function which is a sum of independent variables. This is not our case here since our loss function sums over all polyads that are not necessarily independent, moreover, as mentioned previously in the section on the  consistency result above, we do not assume the existence of a limiting risk function; we only require the existence of a limiting covariance matrix of the risk in $\beta_\star$. We therefore need to adapt the classical proof of asymptotic normality under convexity assumption to our setup. The proof of asymptotic normality of the Polyads estimator will rely on the following general result.

\begin{theorem}
    \label{theo:general_CLT} Let \( \Theta \) be a non-empty, open and convex set in $\bR^d$, let $\theta_0$ be in $\Theta$ and denote $\eps_0>0$ such that $B_2(\theta_0, 3\eps_0)\subset \Theta$.   Let $(\hat Q_n)_n$ be a sequence of random functions defined on $\Theta$ and denote $Q_n = \bE \hat Q_n$. We assume that
\begin{enumerate}
    \item for $n$ large enough,  $\hat Q_n$ is convex, twice differentiable and $\inf_{\theta\in\Theta}\hat Q_n(\theta)$ is achieved at $\hat\theta_n\in\Theta$;
    \item for $n$ large enough, $Q_n$ is twice differentiable, it is minimized at $\theta_0$ over $\Theta$, as $n$ tends to infinity, $\sup\left(\norm{\nabla^2 Q_n(\theta) - \nabla^2 Q_n(\theta_0)}_{op}:\norm{\theta_0- \theta}_2\leq 2 \eps_0/\sqrt{n} \right)\to 0$ and there exists $H_0\succ0$ such that, as $n$ tends to infinity, $\nabla^2 Q_n(\theta_0) \to H_0$;
    
    \item as $n$ tends to infinity, $\sup\left(\norm{\nabla^2 \hat Q_n(\theta) - \nabla^2 Q_n(\theta)}_{op}: \norm{\theta_0 - \theta}_2\leq 2 \eps_0/\sqrt{n} \right)\ove{\to}{p}0$;
    \item there exists $V_0\succ0$ such that, as $n$ tends to infinity, $\sqrt{n} \nabla \hat Q_n(\theta_0)\overset{d}{\to} \cN(0, V_0)$.
\end{enumerate}
Then, as $n$ tends to infinity, $\sqrt{n}(\hat{\theta}_n  - \theta_0)\overset{d}{\to} \cN(0, \Sigma_0)$ where $\Sigma_0 =  H_0^{-1} V_0 H_0^{-1}$.
\end{theorem}


\begin{proof}
    For all $t\in\bR^d$ such that $\theta_0+t\in\Theta$, we define
    \begin{equation*}
        Z_n(t):= \hat Q_n\left(\theta_0 + \frac{t}{\sqrt{n}}\right) - \hat Q_n(\theta_0) - \inr{\nabla \hat Q_n(\theta_0), \frac{t}{\sqrt{n}}}.
    \end{equation*} Let $t\in\bR^d$ be such that $\norm{t}_2\leq 2 \eps_0$. It follows from  a second order Taylor expansion of $\hat Q_n$ at $\theta_0$ that there exists $\hat\theta_1\in[\theta_0, \theta_0+t/\sqrt{n}]$ such that
       \begin{align*}
   n\left|Z_n(t) - \bE Z_n(t)\right| & = \frac{1}{2}\left|t^\top \left(\nabla^2 \hat Q_n(\hat \theta_1) - \bE\nabla^2 \hat Q_n(\hat \theta_1)\right) t \right| \leq   \frac{\norm{t}_2^2}{2}\sup_{\theta_1\in [\theta_0, \theta_0+t/\sqrt{n}]}\left\|\nabla^2 Q_n(\theta_1) - \nabla^2 \hat Q_n(\theta_1)\right\|_{op}\\
   &\leq 2\eps^2_0 \sup\left(\norm{\nabla^2 \hat Q_n(\theta) - \nabla^2 Q_n(\theta)}_{op}: \norm{\theta_0 - \theta}_2\leq 2 \eps_0/\sqrt{n} \right) 
    \end{align*}and so by \textit{item~3} above, we obtain
    \begin{equation}\label{eq:conv_compact_set}
    n \sup_{t:\norm{t}_2\leq 2 \eps_0} \left|Z_n(t) - \bE Z_n(t) \right| \overset{p}{\to} 0.
\end{equation}

Next, we define for all $t\in\bR^d$ such that $\theta_0+t\in\Theta$, $\hat F_n(t) = \hat Q_n(\theta_0+t/\sqrt{n}) - \hat Q_n(\theta_0)$ and 
\begin{equation*}
    \hat H_n(t) = \inr{\nabla \hat Q_n(\theta_0), \frac{t}{\sqrt{n}}} + \frac{1}{2}\left(\frac{t}{\sqrt{n}}\right)^\top H_0 \left(\frac{t}{\sqrt{n}}\right).
\end{equation*}It follows from a second order Taylor expansion  of $Q_n$ at $\theta_0$ that
\begin{align*}
  & n \sup_{t:\norm{t}_2\leq 2 \eps_0}|\hat F_n(t) - \hat H_n(t)| = n \sup_{t:\norm{t}_2\leq 2 \eps_0}\left| Z_n(t) - \frac{1}{2}\left(\frac{t}{\sqrt{n}}\right)^\top H_0 \left(\frac{t}{\sqrt{n}}\right) \right|\\
 & \leq  n \sup_{t:\norm{t}_2\leq 2 \eps_0} \left|Z_n(t) - \bE Z_n(t) \right| + n \sup_{t:\norm{t}_2\leq 2 \eps_0} \left|\bE Z_n(t) - \frac{1}{2}\frac{t}{\sqrt{n}}\nabla^2 Q_n(\theta_0)\frac{t}{\sqrt{n}} \right|  \\
 & + \frac{1}{2} \sup_{t:\norm{t}_2\leq 2 \eps_0} \left|t^\top \left(\nabla^2 Q_n(\theta_0) - H_0\right)t\right|\\
 & \leq  n \sup_{t:\norm{t}_2\leq 2 \eps_0} \left|Z_n(t) - \bE Z_n(t) \right| + \sup_{\norm{\theta_1- \theta_0}_2\leq 2\eps_0/\sqrt{n}} \norm{\nabla^2 Q_n(\theta_1) - \nabla^2 Q_n(\theta_0))}_{op}  + 2\eps_0^2 \norm{\nabla^2 Q_n(\theta_0) - H_0}_{op}.
\end{align*} Hence, it follows from \eqref{eq:conv_compact_set} and \textit{item~2} above that  
\begin{equation}
    \label{eq:cv_compact_set_F_H}
    n \sup_{t:\norm{t}_2\leq 2 \eps_0}|\hat F_n(t) - \hat H_n(t)|  \overset{p}{\to} 0.
\end{equation}

By definition, $\hat t_n = \sqrt{n}(\hat\theta_n - \theta_0)$ is a minimizer of $\hat F_n$ and $\tilde t_n = -H_0^{-1} (\sqrt{n} \nabla \hat Q_n(\theta_0))$ is a minimizer of $\hat H_n$. It follows from the convergence assumption on the sequence of gradient that $\tilde t_n \overset{d}{\to}\cN(0, \Sigma_0)$; hence, by Slutsky's Lemma, it only remains to show that $\hat t_n - \tilde t_n \overset{p}{\to} 0$ to get the result. 

Let $0<\eps\leq \eps_0$. Let us prove that $\bP\left[\norm{\hat t_n - \tilde t_n}_2\geq \eps\right]\to 0$. Let $0<\eps_1<1$. Let us show that there exists $n_0\in\bN$ such that for all $n\geq n_0, \bP\left[\norm{\hat t_n - \tilde t_n}_2\geq \eps\right]\leq \eps_1$.  
We have $\tilde t_n = \cO_p(1)$, hence, there exists $n_1$ such that $\bP\left[\norm{\tilde t_n}_2\leq \eps_0 \right]\geq 1-\eps_1/2$ for all $n\geq n_1$. It follows from \eqref{eq:cv_compact_set_F_H} that there exists $n_2$ such that for all $n\geq n_2$, $\bP\left[n\sup_{t:\norm{t}_2\leq 2 \eps_0}|\hat F_n(t) - \hat H_n(t)|\leq \eps^2\lambda_{min}/12\right]\geq 1-\eps_1/2$ where $\lambda_{min}$ is the smallest singular value of $H_0$ (it is positive by assumption). Next, we check that for all $t\in\bR^d$ such that $\theta_0+t\in\Theta$, 
\begin{equation*}
    \hat H_n(t) - \hat H_n(\tilde t_n) = \frac{1}{2n}\norm{H_0^{1/2}(t-\tilde t_n)}_2^2 \geq \frac{\lambda_{min}}{2n}\norm{t - \tilde t_n}_2^2.
\end{equation*}

Denote by $\Omega_0$ the event onto which $n\sup_{t:\norm{t}_2\leq 2 \eps_0}|\hat F_n(t) - \hat H_n(t)|\leq \eps^2\lambda_{min}/12$ and $\norm{\tilde t_n}_2\leq \eps_0$. Let $n$ be larger than $\max(n_1, n_2)$ and denote by $S_2(\tilde t_n, \eps) = \{t:\norm{t - \tilde t_n}_2 =  \eps\}$. On $\Omega_0$, we have $S_2(\tilde t_n, \eps)\subset B_2(0, 2\eps_0)$ and so for all $t$ such that $\norm{t-\tilde t_n}_2 = \eps$, 
\begin{align*}
     \hat F_n(t) & \geq \hat H_n(t) - \frac{\eps^2\lambda_{min}}{12n} > \hat H_n(\tilde t_n) + \frac{\lambda_{min}}{2n}\norm{t - \tilde t_n}_2^2 - \frac{\eps^2\lambda_{min}}{12n} = \hat H_n(\tilde t_n) + \frac{5\eps^2\lambda_{min}}{12n}\\
     & \geq \hat F_n(\tilde t_n) - \frac{\eps^2\lambda_{min}}{12n} +\frac{5\eps^2\lambda_{min}}{12n} = \hat F_n(\tilde t_n) + \frac{\eps^2\lambda_{min}}{3n}. 
 \end{align*} We proved that, on the event $\Omega_0$, for all $t$ such that $\norm{t-\tilde t_n}_2 = \eps$, we have $\hat F_n(t) > \hat F_n(\tilde t_n)$. It follows from the convexity of $\hat F_n$ that this results extends to all $t$ such that $\norm{t-\tilde t_n}_2 \geq \eps$. As a consequence, we conclude that the minimizer $\hat t_n$ of $\hat F_n$ necessary lies in $B_2(\tilde t_n, \eps)$. In other words, we proved that for all $n\geq \max(n_1, n_2)$, $\bP\left[\norm{\hat t_n - \tilde t_n}_2>\eps\right]\leq \bP[\Omega_0^c]\leq \eps_1.$
\end{proof}

There are situations where only directional CLTs for the gradient of the loss function at $\theta_0$ are available and where the Cr{\'a}mer-Wold device does not apply. In that case, we may also prove directional CLTs for $(\hat \theta_n)_n$.

\begin{theorem}\label{theo:general_directional_CLT}Let $u\in\bR^d$.
We consider the same setup as Theorem~\ref{theo:general_CLT} except that \textit{item~4} is replaced by:
\begin{itemize}
    \item[$4^\prime$.] There exists a sequence of PSD matrices $(\hat V_n)_n$ such that,  
    $$
    \frac{\inr{\sqrt{n} \nabla \hat Q_n(\theta_0), H_0^{-1}u}}{\norm{\hat V_n^{1/2}H_0^{-1}u}_2}\overset{d}{\to} \cN(0, 1)
    $$and $1= \cO_\bP\left(\norm{\hat V_n^{1/2} H_0^{-1}u}_2\right)$.
\end{itemize}
Then  as $n$ tends to infinity,  
$$
\frac{\inr{\sqrt{n}(\hat{\theta}_n  - \theta_0),u}}{ \norm{\hat V_n^{1/2} H_0^{-1}u}_2}\overset{d}{\to} \cN(0, 1).$$   
\end{theorem}


\begin{proof}[Proof of Theorem~\ref{theo:general_directional_CLT}]
    The proof of Theorem~\ref{theo:general_directional_CLT} follows the same line as the one of Theorem~\ref{theo:general_CLT} except that we have 
   \begin{equation}\label{eq:decomp_direct_CLT}
        \frac{\inr{\sqrt{n}(\hat{\theta}_n  - \theta_0),u}}{ \norm{\hat V_n^{1/2} H_0^{-1}u}_2} = \frac{\inr{\hat t_n - \tilde t_n,u}}{ \norm{\hat V_n^{1/2} H_0^{-1}u}_2} - \frac{\inr{\sqrt{n} \nabla \hat Q_n(\theta_0), H_0^{-1}u}}{ \norm{\hat V_n^{1/2} H_0^{-1}u}_2}. 
    \end{equation} By the assumption of the directional asymptotic normality of the gradient (i.e. \textit{item~$4^\prime$}), we have
    \begin{equation*}
        \frac{\inr{\sqrt{n} \nabla \hat Q_n(\theta_0), H_0^{-1}u}}{ \norm{\hat V_n^{1/2} H_0^{-1}u}_2} \to \cN(0,1)
    \end{equation*}hence by Slutsky it only remains to show that the first term in the right-hand side equality of \eqref{eq:decomp_direct_CLT} tends to $0$ in probability. This can be proved by using that 
    \begin{equation*}
        \frac{\left|\inr{\hat t_n - \tilde t_n,u}\right|}{ \norm{\hat V_n^{1/2} H_0^{-1}u}_2} \leq \cO_\bP(1)\norm{\hat t_n - \tilde t_n}_2,
    \end{equation*}together with the argument in the proof of Theorem~\ref{theo:general_CLT} that shows that $\hat t_n - \tilde t_n\overset{p}{\to}0$. 
\end{proof}

\subsubsection{Proof of Theorem~\ref{theo:CLT_polyads_estim_Graham}}\label{sec:proof_of_CLT_polyads}



To show a.n. of the Polyads estimator, we apply Theorem~\ref{theo:general_directional_CLT}. We therefore need to check the assumptions from this theorem for the choice of loss and risk functions from \eqref{eq:def_risk}.

 First, convexity and the $\cC^2$-regularity of $\hat Q_N$ follow from Lemma~\ref{lem:derivatives} as well as for the existence of the Polyads estimator. It remains to show the properties of the risk functions $Q_N$ (given Assumption \ref{ass:consistency}, we only need to check the uniform continuity of the Hessian), the uniform  convergence over a compact set around $\beta_\star$ in probability of Hessian matrices and the directional a.n. of the gradient of the loss function at $\beta_\star$. The next two sections are devoted to this latter task. It goes through the directional a.n. of the H{\'a}jek projection of the gradient.

\subsubsection{H{\'a}jek projection of the gradient of the loss function at $\beta_\star$}
\label{sub:hajeck}


We use the notation $\beta\to\ell_\xi(\beta) = \ell_{\xi}(Y|X, \beta)$ for all $\xi\in\Xi$ and
\begin{equation*}
    U_N:=\nabla \hat Q_N(\beta_\star) = \frac{1}{{N_a}}\sum_{\xi\in\Xi} \nabla \ell_\xi(\beta_\star).
\end{equation*}



The main difficulty in proving a CLT for $U_N$ is that it is not a sum of independent variables. However, we  prove in this section, as in \cite{graham2017econometric,jochmans2018semiparametric}, that it can be well approximated (in a sense given in Proposition~\ref{prop:projection_principl} below) by its H{\'a}jek projection - which is a sum of independent variables.

We first show that 
\begin{equation*}
     U_N^* := \sum_{\bi} \bE_{\beta_\star} [U_N|X, Y_\bi]
 \end{equation*} is the H{\'a}jeck projection of $U_N$ onto the linear sub-space $\cS$ of $\bR^p$-valued random variables in $L^2(\bP_{\beta^*}[\cdot|X])$ defined by
 \begin{equation*}
     \cS:= \left\{\sum_{\bi}Z_{\bi}: Z_{\bi} {\mbox{ is }} \sigma(X, Y_\bi)-{\mbox{measurable and }} \bE[\norm{Z_{\bi}}_2^2|X]<\infty \right\}.
 \end{equation*} In other word, we want to show that 
 \begin{equation}\label{eq:min_pb_hajek_proj}
     \min\left(\bE_{\beta_\star}[\norm{U_N - Z}_2^2|X]: Z\in \cS\right)
 \end{equation}is achieved by $U_N^*$. Note that the $(\sigma(X, Y_\bi))_{\bi}$ are not independent in general, however, they are independent conditionally on $X$ because all edges weights are independent conditionally on $X$ in model assumption \eqref{assump:model}. That is the reason why we are considering the H{\'a}jeck projection conditionally to $X$. However, because of this conditioning, we cannot directly apply   Lemma~11.10 from \cite{MR1652247} but we can adapt its proof to our context. We therefore follow the proof strategy of  Lemma~11.10 from \cite{MR1652247} to show that $U_N^*$ is solution to \eqref{eq:min_pb_hajek_proj}. 


 First, it is clear that $U_N^*$ belongs to $\cS$ (note that $U_N$ has a second moment because of Lemma~\ref{lem:derivatives} and $m_\xi(Y)$ has a second moment). Next, we want to show that $U_N-U_N^*$ is orthogonal to $\cS$ conditionally on $X$. To prove it, we show that for every $\bi$ and  $Z_{\bi}$ which is $\sigma(X, Y_\bi)$-measurable with a second moment, we have 
\begin{equation*}
    \bE_{\beta_\star}[\inr{Z_{\bi},U_N-U_N^*}|X] = 0.
\end{equation*}Given that $\bE_{\beta_\star}[\inr{Z_{\bi},U_N-U_N^*}|X] =  \bE_{\beta_\star}[\bE_{\beta_\star}[\inr{Z_{\bi},U_N-U_N^*}|X, Y_\bi]|X]$ (see Lemma~\ref{lem:simple_vaart}), it is enough to show that $\bE_{\beta_\star}[\inr{Z_{\bi},\bE[(U_N-U_N^*)|X, Y_\bi]}|X] = 0$. The latter will be true if we show that 
\begin{equation*}
    \bE_{\beta_\star}[U_N|X, Y_\bi] = \bE_{\beta_\star}[U_N^*|X, Y_\bi].
\end{equation*}Given the definition of $U_N^*$, we only need to prove that $\bE_{\beta_\star}[\bE_{\beta_\star}[U_N|X, Y_\bi]|X, Y_{\bi^\prime}] = 0$ for all $\bi^\prime$ different from $\bi$. But for all $\bi^\prime$ different from $\bi$, $Y_{\bi}$ is independent of $Y_\bi$ conditionally on $X$, hence, by Lemma~\ref{lem:simple_cond} and Lemma~\ref{lem:simple_vaart}, we have 
\begin{equation*}
    \bE_{\beta_\star}\left[\bE_{\beta_\star}[ U_N|X, Y_{\bi}|] |X, Y_{\bi^\prime}\right] = \bE_{\beta_\star}\left[\bE_{\beta_\star}[ U_N|X, Y_{\bi}|] |X\right] = \bE_{\beta_\star}[ U_N|X] =0
\end{equation*}by definition of $\beta_\star$. As a consequence, we proved that $U_N-U^*_N$ is orthogonal to $\cS$ conditionally on $X$. This proved that $U_N^*$ is the H{\'a}jeck projection of $U_N$ conditionally on $X$, ie a solution to \eqref{eq:min_pb_hajek_proj}. 




Theorem~\ref{theo:CLT_polyads_estim_Graham} is a directional CLT for the Polyads estimator. To prove it we will apply Theorem~\ref{theo:general_directional_CLT} that requires a directional CLT for $U_N$. Our approach is to show that such a directional CLT for $U_N$ may be derived from a directional CLT for its H{\'a}jek projection. To that end, we need to prove an asymptotic equivalence between the directional projections $\inr{U_n, c}$ and $\inr{U_n^*, c}$ for a given $c\in\bR^p$. The following result show that such a result holds when the variances of $\inr{U_n, c}$ and $\inr{U_n^*, c}$ are asymptotically equivalent. 

 \begin{proposition}\label{prop:projection_principl}
 Let $c\in\bR^p$.
 For almost all $X$, the following statement holds: 
     \begin{equation*}
        \mbox{ if  }  \frac{\bE_{\beta_\star}\left[\inr{U_N,c}^2|X\right]}{\bE_{\beta_\star}\left[\inr{U_N^*,c}^2|X\right]}\to 1 \mbox{ then  }\frac{\inr{U_N,c}}{\sqrt{\bE_{\beta_\star}[\inr{U_N^*,c}^2|X]}} - \frac{\inr{U_N^*,c}}{\sqrt{\bE_{\beta_\star}[\inr{U_N^*,c}^2|X]}} \overset{\bP_{\beta_\star}[\cdot|X]}{\longrightarrow}0.
     \end{equation*}
 \end{proposition}

\begin{proof}[Proof of Proposition~\ref{prop:projection_principl}]
We provide the  proof of Proposition~\ref{prop:projection_principl} for the sake of completeness. We denote $\bE_{\beta_\star} = \bE$. We first show that 
\begin{equation}\label{eq:pythagor_projection_hajek}
    \bE[\inr{U_N-U_N^*,c}^2|X] = \bE[\inr{U_N,c}^2|X] - \bE[\inr{U_N^*,c}^2|X].
\end{equation}We have 
\begin{align*}
    \bE\left[\inr{U_n,c}\inr{U_N^*,c}|X\right] = \sum_\bi \bE\left[\inr{U_N,c} \bE[\inr{U_N,c}|X,Y_\bi] |X\right]  = \sum_\bi \bE\left[ \left(\bE[\inr{U_N,c}|X,Y_\bi]\right)^2|X\right]
\end{align*}and, since $Y_\bi$ and $Y_\bj$ are independent conditionally on $X$ when $\bi\neq \bj$ and $\bE[U_N|X]=0$, we also have 
\begin{align*}
  \bE\left[\inr{U_n^*,c}\inr{U_N^*,c}|X\right] = \sum_{\bi,\bj} \bE\left[ \bE[\inr{U_N,c}|X, Y_\bj] \bE[\inr{U_N,c}|X,Y_\bi]|X\right] =  \sum_{\bi} \bE\left[ \left(\bE[\inr{U_N,c}|X,Y_\bi]\right)^2|X\right]. 
\end{align*}Hence, the latter two quantities are equal and so \eqref{eq:pythagor_projection_hajek} follows. We conclude the proof by using a second order Chebyshev's inequality and by controlling the second order moment with
\begin{equation*}
     \frac{\bE[\inr{U_N - U_N^*,c}^2|X]}{\bE\left[\inr{U_N^*,c}^2|X\right]} = \frac{\bE[\inr{U_N,c}^2|X] - \bE[\inr{U_N^*,c}^2|X]}{\bE[\inr{U_N^*,c}^2|X]}= 1- \frac{\bE[\inr{U_N^*,c}^2|X]}{\bE[\inr{U_N,c}^2|X]}.
 \end{equation*} 
\end{proof}

Following Proposition~\ref{prop:projection_principl}, our next step is to find asymptotic equivalents for $\bE[\inr{U_N,c}^2|X]$ and $\bE[\inr{U_N^*,c}^2|X]$ and show that they are the same. This will show that the condition $\bE[\inr{U_N,c}^2|X] / \bE[\inr{U_N^*,c}^2|X]\to 1$ of Proposition~\ref{prop:projection_principl} holds and so we will be able to apply Proposition~\ref{prop:projection_principl}.  
We have
\begin{align*}
    \bE_{\beta_\star}[\inr{U_N,c}^2|X] &= \frac{1}{N_a^2}\sum_{\xi, \xi^\prime} \bE_{\beta_\star}[\inr{\nabla \ell_{\xi}(\beta_\star),c} \inr{\nabla \ell_{\xi^\prime}(\beta_\star),c}|X] \\
    & = \frac{1}{N_a^2}\sum_{q=0}^{2^D}\sum_{\substack{\xi, \xi^\prime\\ |\cE(\xi)\cap\cE(\xi^\prime)|=q}} \bE_{\beta_\star}\left[\inr{\nabla \ell_{\xi}(\beta_\star),c}\inr{\nabla \ell_{\xi^\prime}(\beta_\star),c}|X\right].
\end{align*} If $\xi$ and $\xi^\prime$ are two polyads with no edge in common then $\nabla \ell_{\xi}(\beta_\star)$ and $\nabla \ell_{\xi^\prime}(\beta_\star)$ are independent conditionally on $X$ and so we have $\bE_{\beta_\star}\left[\inr{\nabla \ell_{\xi}(\beta_\star),c} \inr{\nabla \ell_{\xi^\prime}(\beta_\star),c}|X\right]=0$ in that case because $\bE_{\beta_\star}[\nabla \ell_{\xi}(\beta_\star)|X]=0$ by definition of $\beta_\star$. Moreover, the total number of edges in common of two given polyads lies in $\{0,1,2,2^2,\cdots, 2^D\}$. Hence, we have
\begin{align*}
    \bE_{\beta_\star}[\inr{U_N,c}^2|X] = \frac{1}{N_a^2}\sum_{\xi, \xi^\prime}\Big(&\bE_{\beta_\star}\left[\inr{\nabla \ell_{\xi}(\beta_\star),c}\inr{\nabla \ell_{\xi^\prime}(\beta_\star),c}|X\right]I\left(|\cE(\xi) \cap \cE(\xi^\prime)|=1\right)\\
    & + \sum_{r= 1}^D \bE_{\beta_\star}\left[\inr{\nabla \ell_{\xi}(\beta_\star),c}\inr{\nabla \ell_{\xi^\prime}(\beta_\star),c}|X\right]I\left(|\cE(\xi) \cap \cE(\xi^\prime)|=2^r\right)\Big).
\end{align*}Our aim is now to show that, under the condition of Theorem~\ref{theo:CLT_polyads_estim_Graham}, the dominating term in the sum above is the one with all couples of polyads sharing exactly one edge in common. To that end, we introduce the following notation: for all $q\in[2^D]$,
\begin{equation}\label{eq:Delta_q_N}
    \Delta_{q,N}(X) := \bbE_{\beta_\star}\left[\inr{\nabla \ell_{\bxi}(\beta_\star),c}\inr{\nabla \ell_{\bxi^\prime}(\beta_\star),c}I\left(|\cE(\bxi) \cap \cE(\bxi^\prime)|=q\right)|X\right]
\end{equation}where $\bxi$ and $\bxi^\prime$ are two independent random variables with values in the set of all polyads and uniformly distributed over this set and independent of $Y$ and $X$. Using these notation, we have 
\begin{equation*}
    \bE_{\beta_\star}[\inr{U_N,c}^2|X] = \frac{N^2}{N_a^2} \sum_{q=1}^{2^D} \Delta_{q,N}(X)  = T_1 + \frac{N^2}{N_a^2}\sum_{q=2}^{2^D} \Delta_{q,N}(X)
\end{equation*}where 
\begin{equation}\label{eq:T1}
    T_1 : = \frac{N^2}{N_a^2}\Delta_{1,N}(X)
\end{equation} is the term that we are proving now to be the dominant term in the variance of $\inr{U_N,c}$ as $n$ grows. It follows from Lemma~\ref{lem:Delta} and Assumption \ref{ass:second_order_geometry} that as $n$ grows to $\infty$, 
\begin{align*}
   \frac{N^2}{N_a^2}\sum_{q=2}^{2^D} \Delta_{q,N}(X) \lesssim  \frac{1}{N_a^2}\sum_{\substack{\xi, \xi^\prime\\ |\cE(\xi) \cap \cE(\xi^\prime)|\geq2}} \bP_{\beta_\star}\left[\xi \mbox{ and  } \xi^\prime \mbox{ are active}|X\right] |\inr{\tilde X_\xi,c}\inr{\tilde X_{\xi^\prime},c}| =  o (T_1).
\end{align*}

 Then, regarding the variance of the H{\'a}jek projection in direction $c$, we have
 \begin{equation*}
      \bE_{\beta_\star}[\inr{U_N^*,c}^2|X]  =  \sum_{\bi , \bi^\prime}  \bE_{\beta_\star} \left[ \bE_{\beta_\star} [\inr{U_N,c}|X, Y_\bi] \bE_{\beta_\star} [\inr{U_N,c}|X, Y_{\bi^\prime}] \big|X\right]. 
 \end{equation*}Since $Y_\bi$ and $Y_{\bi^\prime}$ are independent conditionally to $X$ when $\bi\neq \bi^\prime$, we deduce from Lemma~\ref{lem:simple_cond} and Lemma~\ref{lem:simple_vaart} that  
 \begin{align*}
   &\bE_{\beta_\star}  \left[ \bE_{\beta_\star} [\inr{U_N,c}|X, Y_\bi] \bE_{\beta_\star} [\inr{U_N,c}|X, Y_\bi^\prime] \big|X\right] =   \bE_{\beta_\star}  \left[ \bE_{\beta_\star} [\inr{U_N,c}|X, Y_\bi] \big|X\right] \bE_{\beta_\star}  \left[\bE_{\beta_\star} [\inr{U_N,c}|X, Y_\bi^\prime] \big|X\right]\\ 
   &= \left(\bE_{\beta_\star} [\inr{U_N,c}|X]\right)^2=0
 \end{align*}and so 
 \begin{equation*}
     \bE_{\beta_\star}[\inr{U_N^*,c}^2|X] = \sum_{\bi}  \bE_{\beta_\star} \left[ \left(\bE_{\beta_\star} [\inr{U_N,c}|X, Y_\bi] \right)^2 \big|X\right]. 
 \end{equation*}For all $\bi\in\cI$, we have
 \begin{equation*}
     \left(\bE_{\beta_\star} [\inr{U_N,c}|X, Y_\bi]\right)^2  
     = \frac{1}{N_a^2}\sum_{\xi, \xi^\prime} \bE[\inr{\nabla \ell_\xi(\beta_\star),c}|X, Y_\bi]  \bE[\inr{\nabla \ell_{\xi^\prime}(\beta_\star),c}|X, Y_\bi]
 \end{equation*}and so
 \begin{align*}
     \bE\left[\inr{U_N^*,c}^2|X\right] = \frac{1}{N_a^2}\sum_{\xi, \xi^\prime} \sum_{\bi}\bE\left[\bE[\inr{\nabla \ell_\xi(\beta_\star),c}|X, Y_\bi]  \bE[\inr{\nabla \ell_{\xi^\prime}(\beta_\star),c}|X, Y_\bi]|X\right]. 
 \end{align*}We first note that if $\bi$ is not an edge of $\xi$ then $\bE[\nabla \ell_\xi(\beta_\star)|X, Y_\bi]=0$ because in that case $\nabla \ell_\xi(\beta_\star)$ and $Y_\bi$ are independent conditionally to $X$ and we always have $\bE_{\beta_\star}[\nabla \ell_\xi(\beta_\star)|X] = 0$ by definition of $\beta_{\star}$.  We therefore obtain that 
  \begin{align}\label{eq:decomp_U_N_star}
 \nonumber   &\bE\left[\inr{U_N^*,c}^2|X\right] = \frac{1}{N_a^2}\sum_{\xi, \xi^\prime} \left[\sum_{\bi\in \cE(\xi)\cap \cE(\xi^\prime)} \bE\left[\bE[\inr{\nabla \ell_\xi(\beta_\star),c}| X, Y_\bi]  \bE[\inr{\nabla \ell_{\xi^\prime}(\beta_\star),c}| X, Y_\bi]|X\right]\right]\\
 & = \frac{1}{N_a^2} \sum_{q=1}^{2^D}\sum_{\xi, \xi^\prime} \sum_{\bi \in \cE(\xi)\cap \cE(\xi^\prime)} \bE\left[\bE[\inr{\nabla \ell_\xi(\beta_\star),c}|X, Y_\bi]  \bE[\inr{\nabla \ell_{\xi^\prime}(\beta_\star),c}|X, Y_\bi]|X\right] I\left(|\cE(\xi)\cap\cE(\xi^\prime)|=q\right).
  \end{align} Moreover, if $\xi$ and $\xi^\prime$ share exactly one edge, denoted by $\bi$, in common then  $\nabla \ell_\xi(\beta_\star)$ and $\nabla \ell_{\xi^\prime}(\beta_\star)$ are independent conditionally to $\sigma(X, Y_\bi)$, hence 
 \begin{align*}
        &\bE_{\beta_\star} \left[\bE\left[\inr{\nabla \ell_\xi(\beta_\star),c}| X, Y_\bi\right] \bE\left[\inr{\nabla \ell_{\xi^\prime}(\beta_\star),c}| X, Y_\bi\right]|X\right] =\bE\left[\bE\left[\inr{\nabla \ell_\xi(\beta_\star),c} \inr{\nabla \ell_{\xi^\prime}(\beta_\star),c}| X, Y_\bi\right]|X\right]\\ 
        &= \bE\left[\inr{\nabla \ell_\xi(\beta_\star),c} \inr{\nabla \ell_{\xi^\prime}(\beta_\star),c}|X\right] 
  \end{align*}and so for the '$q=1$' term  in \eqref{eq:decomp_U_N_star}, we get
  \begin{align*}
      &\frac{1}{N_a^2}\sum_{\xi, \xi^\prime} \sum_{\bi \in \cE(\xi)\cap \cE(\xi^\prime)} \bE\left[\bE[\inr{\nabla \ell_\xi(\beta_\star),c}|X, Y_\bi]  \bE[\inr{\nabla \ell_{\xi^\prime}(\beta_\star),c}|X, Y_\bi]|X\right] I\left(|\cE(\xi)\cap\cE(\xi^\prime)|=1\right)\\ 
      &= \frac{N^2}{N_a^2} \Delta_{1,N}(X) = T_1
  \end{align*}where $\Delta_{1,N}(X)$ has been introduced in \eqref{eq:Delta_q_N} and $T_1$ in \eqref{eq:T1}. This shows that in the variance decomposition of both $\inr{U_N,c}$ and $\inr{U_N^*,c}$ the term for $q=1$ (i.e. the term coming from all couples of polyads sharing exactly one edge in common) is the same given by $T_1$. The next step is to show, under Assumption \ref{ass:second_order_geometry}, that this term is also the dominant term in $\bE\left[\inr{U_N^*,c}^2|X\right]$ as $n$ grows.
It follows from Lemma~\ref{lem:Delta} that
\begin{align*}
&\frac{1}{N_a^2}\sum_{\substack{\xi, \xi^\prime\\ |\cE(\xi)\cap\cE(\xi^\prime)|\geq2}} \sum_{\bi \in \cE(\xi)\cap \cE(\xi^\prime)} \bE\left[\bE[\inr{\nabla \ell_\xi(\beta_\star),c}|X, Y_\bi]  \bE[\inr{\nabla \ell_{\xi^\prime}(\beta_\star),c}|X, Y_\bi]|X\right]  \\ 
&\lesssim \frac{1}{N_a^2}\sum_{\substack{\xi, \xi^\prime\\ |\cE(\xi)\cap\cE(\xi^\prime)|\geq2}} \sum_{\bi \in \cE(\xi)\cap \cE(\xi^\prime)} \bP_{\beta_\star}\left[\xi \mbox{ and  } \xi^\prime \mbox{ are active}|X\right]) |\inr{\tilde X_\xi,c}\inr{\tilde X_{\xi^\prime},c}|
\end{align*}and so we conclude that $T_1$ is the dominating term in the variance of $\inr{U_N^*,c}$ thanks to Assumption \ref{ass:second_order_geometry}.


We conclude that both $\bE\left[\inr{U_N,c}^2|X\right]$ and $\bE\left[\inr{U_N^*,c}^2|X\right]$ are asymptoticaly equivalent to $T_1$ and so their ratio tends to $1$ as $n$ grows. Then, it follows from Proposition~\ref{prop:projection_principl} that, as $n$ tends to infinity, 
\begin{equation}\label{eq:final_result_a_e_UN_UN_star}
    \frac{\inr{U_N,c}}{\sqrt{\bE_{\beta_\star}[\inr{U_N^*,c}^2|X]}} - \frac{\inr{U_N^*,c}}{\sqrt{\bE_{\beta_\star}[\inr{U_N^*,c}^2|X]}} \overset{\bP_{\beta_\star}[\cdot|X]}{\longrightarrow}0.
\end{equation}

\subsubsection{Asymptotic normality of the H{\'a}jek projection of the gradient} 
\label{sub:asymptotic_normality_of_the_h}
We recall that  $U^*_N$ is the H{\'a}jek projection of $U_N$ where
\begin{equation*}
    U_N^* := \sum_{\bi} \bE_{\beta_\star} [U_N|X, Y_\bi] \mbox{ and } U_N= \frac{1}{{N_a}}\sum_{\xi\in\Xi} \nabla \ell_\xi(\beta_\star). 
\end{equation*} Our aim, in this section, is to show a directional CLT for $U_N^*$ conditionally to $X$. Together with the asymptotic equivalence proved in \eqref{eq:final_result_a_e_UN_UN_star} this will prove a similar directional CLT for $U_N$. 



We know that if $\bi$ is not an edge of $\xi$ then $\bE[\nabla \ell_\xi(\beta_\star) | X, Y_\bi]=0$ by definition of $\beta_\star$. As a consequence, we have
\begin{equation}\label{eq:def_U_N_star_section_5_5}
    U_N^* = \frac{1}{{N_a}}\sum_{\bi} \sum_{\xi: \bi\in\cE(\xi)} \bE[\nabla \ell_\xi(\beta_\star) | X, Y_\bi] = \frac{1}{{N_a}} \sum_{\bi} \bar{s}_{\bi}
\end{equation}where, we recall that for all $\bi\in\cI$
$$
\bar{s}_{\bi} = \sum_{\xi:\bi\in\cE(\xi)} \bE[\nabla \ell_\xi(\beta_\star) | X, Y_\bi]. 
$$

All the point of working on $U_N^*$ and note directly on $U_N$ is that, conditionally on $X$, the $\bar{s}_{\bi}$'s are independent because $\bar{s}_\bi$ is $\sigma(X, Y_\bi)$-measurable and the edge weights $Y_{\bi}$'s are independent conditionally on $X$. To prove a.n. for $U_N^*$ we apply \cite{MR2294976} as in \cite{graham2017econometric}.

To simplify the exposition we denote by $\bE$ the expectation under $\bP_{\beta_\star, \theta}$, i.e. $\bE_{\beta_\star}$ and  $\alpha_{1,N} = \alpha_{1,N}(X)$. Let $c\in\bR^p$.  We first compute the variance of the H{\'a}jek projection using the notation from \eqref{eq:def_U_N_star_section_5_5}: since $\bar{s}_{\bi}$ and $\bar{s}_{\bi^\prime}$ are independent conditionally on $X$  when $\bi\neq \bi^\prime$ and $\bE[\bar s_\bi|X]=0$, we have
\begin{align}\label{eq:var_UN_star_equiv_Delta_N}
  \bE\left[\inr{U_N^*,c}^2|X\right] =  \frac{1}{N_a^2} \sum_{\bi, \bi^\prime} c^\top \bE\left[\bar{s}_{\bi} \bar{s}_{\bi^\prime}^\top|X\right] c = \frac{1}{{N_a}^2} \sum_{\bi} c^\top \bE\left[\bar{s}_{\bi} \bar{s}_{\bi}^\top|X\right] c= \frac{n \alpha_{1,N}}{{N_a}^2} c^\top\tilde \Delta_N c
\end{align} where $n = \prod_d n_d$ and
\begin{equation*}
   \tilde \Delta_N  = \frac{1}{n}\sum_{\bi} \tilde \Delta_{\bi} \mbox{ for }   \tilde\Delta_{\bi} = \alpha_{1, N}^{-1}\bE\left[\bar{s}_{\bi} \bar{s}_{\bi}^\top| X\right].
\end{equation*} Therefore, we want to show that for almost all $X$, conditionally on $X$,
\begin{equation*}
\frac{{N_a}}{\sqrt{c^\top\tilde \Delta_N c}}\inr{\frac{1}{\sqrt{n\alpha_{1, N}}} U_N^*,c} = 
    \frac{1}{\sqrt{c^\top\tilde \Delta_N c}}\inr{\frac{1}{\sqrt{n\alpha_{1, N}}}\sum_{\bi}\bar{s}_{\bi},c}\overset{d}{\to} \cN(0, 1).
\end{equation*}

For all $\bi$, we set
\begin{equation*}
    R_\bi = \frac{1}{\sqrt{\alpha_{1, N}}}\frac{\inr{\bar{s}_\bi, c}}{\sqrt{c^\top\tilde \Delta_N c}}.
\end{equation*}Our aim is to show that for almost all $X$, conditionally on $X$, $n^{-1/2} \sum_{\bi} R_\bi \overset{d}{\to} \cN(0,1)$ as $n\to+\infty$.

Let us first start with the computation of the conditional expectation and variance of $R_\bi$. Since $\tilde \Delta_{N}$ is $\sigma(X)$-measurable, we have
\begin{equation*}
     \bE[R_{\bi}|X] = \frac{1}{\sqrt{\alpha_{1, N}}}\frac{\inr{\bE[\bar{s}_\bi|X], c}}{\sqrt{c^\top\tilde \Delta_N c}}=0
 \end{equation*} because $\bE[\bar{s}_\bi|X]=0$ since $\bE[\nabla \ell_\xi(\beta_\star)|X]=0$ for all polyads $\xi$. To compute the variance of $R_\bi$ conditionally on $X$, we compute its second moment and since $\tilde \Delta_{N}$ is $\sigma(X)$-measurable, it only depends on the second moment of $\inr{\bar{s}_\bi, c}$ conditionally on $X$, which is the quantity appearing in the asymptotic covariance matrix from Theorem~\ref{theo:CLT_polyads_estim_Graham}:
 \begin{align*}
      \bE\left[R_\bi^2|X\right] = \frac{c^\top\tilde \Delta_\bi c}{c^\top\tilde \Delta_N c}:= \sigma_\bi^2.
  \end{align*} 


 Now, as in \cite{MR2294976} and \cite{graham2017econometric}, we consider independent Gaussian variables with the same first and second moments as the $R_\bi$'s: let $(G_\bi)_\bi$ be independent Gaussian variable with mean $0$ and variances $(\sigma_\bi^2)_\bi$. In particular, we observe that, conditionally on $X$, $n^{-1/2} \sum_\bi G_\bi$ is a $\cN(0,1)$ random variable since $n^{-1}\sum_\bi \sigma_\bi^2=1$. As a consequence, in order to show that $n^{-1/2} \sum_{\bi} R_\bi \overset{d}{\to} \cN(0,1)$ conditionally on $X$, it is enough to show that for all functions $f\in\cC^3$ so that $\norm{f^{(r)}}_{\infty}\leq L$ for all $r=0, 1, 2, 3$, we have
 \begin{equation}
     \label{eq:our_aim_CLT_chatterjee}
     \bE\left[f\left(\frac{1}{\sqrt{n}}\sum_{\bi} R_\bi\right) - f\left(\frac{1}{\sqrt{n}}\sum_{\bi} G_\bi\right) \bigg| X\right]\to 0.
 \end{equation}

 Let $f\in\cC^3$ so that $\norm{f^{(r)}}_{\infty}\leq L$ for all $r=0, 1, 2, 3$. Our aim is to show that \eqref{eq:our_aim_CLT_chatterjee} holds. As in \cite{MR2294976} and \cite{graham2017econometric}, we consider
 \begin{equation*}
     Z_\bi = (R_1, R_2, \ldots, R_\bi, G_{\bi+1}, \ldots, G_{n})\mbox{ and }
     Z_\bi = (R_1, R_2, \ldots, R_{\bi-1},0, G_{\bi+1}, \ldots, G_{n})
 \end{equation*}where $n = \prod_d n_d = |\cI|$ is the total number of possible edges. Using a second order Taylor approximation of $f$, we get, as $n$ tends to infinity,
 \begin{align*}
  &f\left(\frac{1}{\sqrt{n}}\sum_{\bi} R_\bi\right) - f\left(\frac{1}{\sqrt{n}}\sum_{\bi} G_\bi\right)  = \sum_{\bi=1}^{n} f\left(\frac{1}{\sqrt{n}}\inr{Z_\bi, \1} \right) - f\left(\frac{1}{\sqrt{n}}\inr{Z_{\bi-1}, \1}\right) \\
  & =  \sum_{\bi=1}^{n} f\left(\frac{1}{\sqrt{n}}\inr{Z_\bi^0, \1} + \frac{R_\bi}{\sqrt{n}} \right) - f\left(\frac{1}{\sqrt{n}}\inr{Z_{\bi}^0, \1} + \frac{G_\bi}{\sqrt{n}}\right) \\
  & = \sum_{\bi=1}^n \left(\frac{R_\bi}{\sqrt{n}} - \frac{G_\bi}{\sqrt{n}}\right)f^\prime\left(\frac{1}{\sqrt{n}}\inr{Z_{\bi}^0, \1}\right) + \frac{1}{2}\left(\frac{R_\bi^2}{n} - \frac{G_\bi^2}{n}\right)f^{\prime\prime}\left(\frac{1}{\sqrt{n}}\inr{Z_{\bi}^0, \1}\right) + \cO\left(\frac{|R_\bi|^3}{n^{3/2}} + \frac{|G_\bi|^3}{n^{3/2}}\right)
 \end{align*}where $\1=(1)_{1}^n$. Next, we observe that $Z_\bi^0$ is independent of $R_\bi$ and $G_\bi$ conditionally on $X$ and since $R_\bi$ and $G_\bi$ have the same first and second moments conditionally on $X$, we obtain that, as $n$ tends to infinity
\begin{equation}\label{eq:UB_R_3_Y_3}
    \left| \bE\left[f\left(\frac{1}{\sqrt{n}}\sum_{\bi} R_\bi\right) - f\left(\frac{1}{\sqrt{n}}\sum_{\bi} G_\bi\right) \bigg| X\right] \right| = \cO\left(\frac{1}{n^{3/2}} \sum_\bi \bE\left[|R_\bi|^3 + |G_\bi|^3|X\right]\right).
\end{equation}Next, we show that the right-hand side from \eqref{eq:UB_R_3_Y_3} goes to $0$ as $n$ goes to infinity. First, we observe that for all $\bi$, $\bE[|G_\bi|^3|X] \lesssim \left(\bE[|G_\bi|^2|X]\right)^{3/2} = \left(\bE[R_\bi^2|X]\right)^{3/2}\leq \bE[|R_\bi|^3|X]$. Hence, we only need to show that $(1/n)\sum_\bi \bE\left[|R_\bi|^3|X\right] = o(n^{1/2})$. We have
\begin{equation*}
    \frac{1}{n}\sum_{\bi} \bE[|R_\bi|^3|X] = \bE\left[ \frac{\frac{1}{n}\sum_\bi |Z_\bi|^3 }{\left(\frac{1}{n}\sum_\bi|Z_\bi|^2\right)^{3/2}}\bigg| X\right] \mbox{ where  } Z_\bi = \inr{\bar s_\bi,c}.
\end{equation*}We note that $Z_\bi$'s are independent conditionally to $X$ because $Z_\bi$ is $\sigma(X, Y_\bi)$-measurable and the $Y_\bi$'s are independent conditionally on $X$.  Then, we apply Proposition~\ref{prop:ratio_L3_L2_V0} to $X_{\bi,n} = Z_\bi$ --- thanks to Assumption~\ref{ass:third_order_and_non_zero_variance}, the condition~\eqref{eq:weak_assum_for_L3_L2} required to apply Proposition~\ref{prop:ratio_L3_L2_V0} is satisfied.  It then follows from Proposition~\ref{prop:ratio_L3_L2_V0}, for $S_{k,n}=\frac{1}{n}\sum_\bi \bE[|Z_\bi|^k|X]$, that when $n\to+\infty$,
\begin{equation*}
  \frac1n\sum_\bi \bE\left[|R_\bi|^3|X\right]  =   \bE\left[ \frac{\frac{1}{n}\sum_\bi |Z_\bi|^3 }{\left(\frac{1}{n}\sum_\bi|Z_\bi|^2\right)^{3/2}}\bigg| X\right] =  o(\sqrt{n}).
\end{equation*}

 In the end, we showed that as $n$ goes to infinity,
 \begin{equation}\label{eq:cv_result_CLT_Hajek_proj}
    \bE\left[f\left(\frac{1}{\sqrt{n}}\sum_{\bi} R_\bi\right)\bigg|X\right] \to \bE f(g) 
 \end{equation}where $g\sim \cN(0,1)$ for all $f\in\cC^3$ so that $\norm{f^{(r)}}_{\infty}\leq L$ for all $r=0, 1, 2, 3$. We conclude that
 for almost all $X$, conditionally on $X$, as $n$ tends to infinity, 
 \begin{equation*}
 \frac{\inr{U_N^*,c}}{\sqrt{\bE_{\beta_\star}[\inr{U_N^*,c}^2|X]}} =    \frac{{N_a}}{\sqrt{c^\top\tilde \Delta_N c}}\inr{\frac{1}{\sqrt{n\alpha_{1, N}}} U_N^*,c} =  \frac{1}{\sqrt{n}}\sum_{\bi} R_\bi \overset{d}{\to} \cN(0,1) 
 \end{equation*}and so, it follows from \eqref{eq:final_result_a_e_UN_UN_star} and Slutsky's lemma that the same result holds for the gradient itself: 
  \begin{equation}
     \label{eq:final_CLT}
    \frac{{N_a}}{\sqrt{c^\top\tilde \Delta_N c}}\inr{\frac{1}{\sqrt{n\alpha_{1, N}}} U_N,c} \overset{d}{\to} \cN(0,1) 
 \end{equation}where this convergence holds w.r.t. $\bP_{\beta_\star}[\cdot|X]$.

\subsubsection{Uniform convergence of the Hessian over the compact set $B_2(\beta_\star, 1)$} 
 \label{sub:uniform_convergence_of_the_hessian_over_compact_sets}
 The final ingredient needed to apply Theorem~\ref{theo:general_directional_CLT} is the uniform convergence in probability over the compact set $K:=B_2(\beta_\star, 1) =\{\beta:\norm{\beta-\beta_\star}_2\leq 1\}$ of the Hessian matrices of the loss functions (it is item~3. from Theorem~\ref{theo:general_CLT} for $\eps_0=1/2$). It is the aim of this section to check the uniform convergence (conditionally on $X$) $\sup_{\beta\in K}\norm{\nabla^2 \hat Q_N(\beta) - \nabla^2 Q_N(\beta)}_{op}\ove{\to}{p}0$ where, we recall that for all $\beta\in\bR^p$,
\begin{equation*}
    \hat Q_N(\beta) = \frac{1}{{N_a}}L(Y | X, \beta)  \mbox{ and } Q_N(\beta) = \bE_{\beta_\star}[\hat Q_N(\beta)|X].
\end{equation*}
 
In \cite{graham2017econometric,jochmans2018semiparametric}, the authors apply Lemma~2.9 from \cite{newey1994large} to prove the uniform convergence of the sequence of Hessian matrices of the loss function. However, the latter result requires the existence of a limit risk function that we don't have. We therefore apply Corollary~2.2 from \cite{newey91}. We recall this result adapted to our setup (we work with Hessian functions with values in $\bR^{p\times p}$ unlike \cite{newey91}).

\begin{proposition}[Corollary~2.2 in \cite{newey91}]
    Let $(\hat H_n)_n$ (resp. $(H_n)_n$) be a sequence of random (resp. deterministic) functions with values in $\bR^{p\times p}$ defined on a compact set $K$ of $\bR^p$. We assume that:
    \begin{itemize}
          \item for all $\beta\in K$, $\hat H_n(\beta) - H_n(\beta) \overset{p}{\to} 0$;
          \item there exists $(B_n)_n$ such that $B_n=\cO_p(1)$ and for all $\beta_0, \beta_1\in K, \norm{\hat H_n(\beta_0) - \hat H_n(\beta_1)}_{op}\leq B_n \norm{\beta_0 - \beta_1}_2$;
          \item $(H_n)_n$ is equicontinuous.
      \end{itemize} Then, we have $\sup_{\beta\in K} \norm{\hat H_n(\beta) - H_n(\beta)}_{op} \overset{p}{\to} 0$. 
\end{proposition}
\begin{proof}
   The proof is a straightforward application of the one-dimensional result given in  Corollary~2.2 in \cite{newey91} since one can write 
   \begin{equation*}
       \sup_{\beta\in K} \norm{\hat H_n(\beta) - H_n(\beta)}_{op} = \sup_{(\beta,A)\in K\times B}|\inr{A,\hat H_n(\beta)} - \inr{A, H_n(\beta)}|  
   \end{equation*}where $B$ is the dual (compact) ball of the operator norm in $\bR^{p\times p}$.
\end{proof}

The equicontinuity of the family of Hessian matrix of the risk functions is assumed in Assumption~\ref{ass:technical_CLT}. We therefore, only have to show the Lipshitz property of the Hessian matrices of the loss functions and their point wise convergence in probability, i.e.:
\begin{itemize}
     \item[(a)] for all $\beta_0, \beta_1\in K$, $\norm{\nabla^2 \hat Q_N(\beta_0) - \nabla^2 \hat Q_N(\beta_1)}_{op} = O_p(1)\norm{\beta_1 - \beta_0}_2$
     \item[(b)] for all $\beta\in K$,  $\nabla^2 \hat Q_N(\beta) - \nabla^2 Q_N(\beta)\overset{p}{\to} 0$
 \end{itemize} 

We start with the point-wise convergence in probability over $K$. To prove this result it is enough to show that all $p^2$ entries of the matrix $\nabla^2 \hat Q_N(\beta) - \nabla^2 Q_N(\beta)$ tend to zero in probability. Let $\beta\in\bR^p$ and $(i,j)\in[p]^2$. It follows from Lemma~\ref{lem:derivatives} that
\begin{align*}
    \left(\nabla^2 \hat Q_N(\beta) - \nabla^2 Q_N(\beta)\right)_{ij} &= \frac{1}{{N_a}} \sum_{\xi}\left( \mathbb{V}_\beta\left[ 
m_\xi(Y^\prime) |X,  Y^\prime \in \sO_\xi(Y) \right] -  \bE_{\beta_\star}\left[ \mathbb{V}_\beta\left[ 
m_\xi(Y^\prime) |X,  Y^\prime \in \sO_\xi(Y) \right]\right]\right)(\tilde X_\xi)_i (\tilde X_\xi^\top)_j 
\end{align*}where $Y^\prime$ and $Y$ are independent random variables with values in $\bN^\cI$ distributed according to $\bP_{\beta, \theta}$ and $\bP_{\beta_\star, \theta}$ respectively. It follows from Chebyshev's inequality that we only have to show that the second moment of $\left(\nabla^2 \hat Q_N(\beta) - \nabla^2 Q_N(\beta)\right)_{ij}$ tends to $0$. As in the proof of consistency, we use that two polyads having no edge in common are independent so that we get
\begin{align*}
&\bE_{\beta_\star}\left[\left(\nabla^2 \hat Q_N(\beta) - \nabla^2 Q_N(\beta)\right)_{ij}^2\bigg| X\right]\\
 &\bE_{\beta_\star}\left[\left(\frac{1}{{N_a}} \sum_{\xi}\left( \mathbb{V}_\beta\left[ 
m_\xi(Y^\prime) |X,  Y^\prime \in \sO_\xi(Y) \right] -  \bE_{\beta_\star}\left[ \mathbb{V}_\beta\left[ 
m_\xi(Y^\prime) |X,  Y^\prime \in \sO_\xi(Y) \right]\right]\right)(\tilde X_\xi)_i (\tilde X_\xi^\top)_j\right)^2 \bigg|X\right]  \\
&= \frac{1}{{N_a}^2}\sum_{\substack{\xi, \xi^\prime\in\Xi\\\cE(\xi)\cap \cE(\xi^\prime)\neq \emptyset}} \bE_{\beta_\star}\left[\left(V_\xi(Y) - \bE_{\beta_\star}V_\xi(Y)\right) \left(V_{\xi^\prime}(Y) - \bE_{\beta_\star}V_{\xi^\prime}(Y)\right) |X \right](\tilde X_\xi)_i (\tilde X_\xi^\top)_j  (\tilde X_{\xi^\prime})_i (\tilde X_{\xi^\prime}^\top)_j 
\end{align*}where we denote for all $y$, $V_\xi(y) = \mathbb{V}_\beta\left[m_\xi(Y^\prime) |X,  Y^\prime \in \sO_\xi(y) \right]$. Let us now take a closer look at the correlation terms in the sum above: let $\xi$ and $\xi^\prime$ be two polyads, we have
\begin{align*}
    \bE_{\beta_\star}\left[\left(V_\xi(Y) - \bE_{\beta_\star}V_\xi(Y)\right) \left(V_{\xi^\prime}(Y) - \bE_{\beta_\star}V_{\xi^\prime}(Y)\right) |X \right]\leq \bE_{\beta_\star}\left[V_\xi(Y)V_{\xi^\prime}(Y) |X \right]
\end{align*}because $V_\xi(Y)\geq 0$ a.s.. Next, we have $m_\xi(y^\prime)\leq m_\xi(Y) + M_\xi(Y)$ for all $y^\prime\in \cO_\xi(Y)$ and so $V_\xi(Y)\leq m_\xi(Y) + M_\xi(Y)$. Hence, it follows from Lemma~\ref{lem:moment_min_Poisson} that
\begin{equation}\label{eq:bound_correlation_variance_for_Hessian_cv}
    \bE_{\beta_\star}\left[V_\xi(Y)V_{\xi^\prime}(Y) |X \right]\leq \bE_{\beta_\star}\left[(m_\xi(Y) + M_\xi(Y))(m_{\xi^\prime}(Y) + M_{\xi^\prime}(Y))|X \right]  \leq c_0 (\bar\lambda+1)^2 \bP_{\beta_\star}\left[ \xi \mbox{ and } \xi^\prime \mbox{ are active}|X \right]
\end{equation}and so under Assumption \ref{ass:first_order_geometry}, we have $\nabla^2 \hat Q_N(\beta) - \nabla^2 Q_N(\beta)\overset{p}{\to} 0$.

Next, we move to the Lipschitz property over $K=B_2(\beta_\star, 1)$ of the Hessian matrices of the loss functions from point \textit{(a)} above. Let $\beta_0, \beta_1\in K$. 
We compute the third order derivative of the loss functions to prove the Lipschitz property of the Hessian. For all $s\in\bN^*, y=(y_\bi)_\bi$ and $\beta\in\bR^p$, we let 
\begin{equation*}
   \kappa_{s, \xi}(\beta, y) := \bE_\beta[ (m_\xi(Y^\prime) - m_\xi(y))^s| X, Y^\prime \in\cO_\xi(y)] = \frac{\sum_{r=-m_\xi(y)}^{M_\xi(y)} r^s \exp(a_r+r \inr{\widetilde{X}_\xi, \beta})}{\sum_{R=-m_\xi(y)}^{M_\xi(y)}\exp(a_R+R \inr{\widetilde{X}_\xi, \beta})}  
\end{equation*}where the last inequality follows from the change of variable $m=r+m_\xi(y)$ and the result on the distribution of $m_\xi(Y)$ given $X$ and $Y\in\cO_\xi(y)$ from Lemma~\ref{lem:derivatives}. Next, it follows from \eqref{eq:computation_Hessian} that for all $\beta$:
\begin{equation*}
    \nabla^3_\beta \ell_\xi(y|X, \beta)  = \left(\kappa_3(\beta, y) - 3\kappa_2(\beta, y)\kappa_1(\beta, y) +2 \kappa_1(\beta, y)^3\right) \widetilde{X}_\xi \widetilde{X}_\xi^\top \widetilde{X}_\xi^\top
\end{equation*}where for any vector $u\in\bR^p$, $uu^\top u^\top$ is the linear operator $h\in\bR^p\to \inr{u,h} uu^\top \in  \bR^{p\times p}$. As a consequence, if we denote by $B_2^p = \{x\in\bR^p:\norm{x}_2\leq 1\}$,  it follows from a Taylor expansion that for all $\beta_0, \beta_1\in\bR^p$,
\begin{align*}
    &\norm{\nabla^2 \hat Q_N(\beta_0) - \nabla^2 \hat Q_N(\beta_1)}_{op}  =  \sup_{x \in B_2^p}|\inr{xx^\top, \nabla^2 \hat Q_N(\beta_0) } - \inr{xx^\top, \nabla^2 \hat Q_N(\beta_1)}|\\  
    &\leq \sup_{x \in B_2^p, \bar\beta\in[\beta_0, \beta_1]} \left|\inr{xx^\top, \frac{1}{{N_a}}\sum_{\xi}\nabla^3_\beta \ell_\xi(y|X, \bar\beta)(\beta_0 - \beta_1)}\right| \\  
    &\leq \sup_{x \in B_2^p, \bar\beta\in[\beta_0, \beta_1]}  \frac{1}{{N_a}}\sum_{\xi} |\inr{xx^\top, \widetilde{X}_\xi \widetilde{X}_\xi^\top}| \left|\kappa_{3, \xi}(\bar\beta, Y) - 3\kappa_{2, \xi}(\bar\beta, Y)\kappa_{1, \xi}(\bar\beta, Y) +2 \kappa_{1, \xi}^3(\bar\beta, Y) \right| \norm{\widetilde{X}_\xi}_2\norm{\beta_0 - \beta_1}_2.
\end{align*}Next, we use that $\inr{xx^\top, \widetilde{X}_\xi \widetilde{X}_\xi^\top} = \inr{x, \widetilde{X}_\xi}^2\leq \norm{\widetilde{X}_\xi}_2^2$ to get
\begin{align*}
   &\norm{\nabla^2 \hat Q_N(\beta_0) - \nabla^2 \hat Q_N(\beta_1)}_{op}\\  
   &\leq \max_{\xi\in\Xi} \norm{\widetilde{X}_\xi}_2^3 \max_{\beta\in K}\left[\frac{1}{{N_a}}\sum_{\xi\in \Xi}\left|\kappa_{3, \xi}(\beta, Y) - 3\kappa_{2, \xi}(\beta, Y)\kappa_{1, \xi}(\beta, Y) +2 \kappa_{1, \xi}^3(\beta, Y) \right| \right] \norm{\beta_0 - \beta_1}_2.
\end{align*} Next, we show that for all $k=1, 2,3$, we have
\begin{equation}\label{eq:uniformly_bounded_terms_for_Hessian_conv}
    \sup_{\beta\in K}\frac{1}{{N_a}}\sum_{\xi\in \Xi}|\kappa_{k, \xi}(\beta, Y)| = \cO_{\bP_{\beta_\star}[\cdot|X]}(1). 
\end{equation}It follows from Lemma~\ref{lem:m_M_orbits} that if $Y^\prime \in\cO_\xi(Y)$ then $|m_\xi(Y^\prime) - m_\xi(Y)|\leq m_\xi(Y) + M_\xi(Y) = |\cO_\xi(Y)|$. As a consequence, we have for all $\beta$ and $\xi$,
\begin{equation*}
    |\kappa_{k, \xi}(\beta, Y)| \leq  \bE_\beta[ |m_\xi(Y^\prime) - m_\xi(Y)|^k| X, Y^\prime \in\cO_\xi(Y), Y]\leq |\cO_\xi(Y)|^k
\end{equation*} and so we have 
\begin{equation*}
    \sup_{\beta\in K}\frac{1}{{N_a}}\sum_{\xi\in \Xi}|\kappa_{k, \xi}(\beta, Y)| \leq \frac{1}{{N_a}}\sum_{\xi\in \Xi} |\cO_\xi(Y)|^k
\end{equation*}By Chebyshev, we only need to show that for all  $k=1, 2,3, \sum_{\xi\in \Xi} \bE_{\beta_\star}[|\cO_\xi(Y)|^k|X] = \cO(N_a)$ as $n$ tends to $\infty$. This result follows from Lemma~\ref{lem:moment_min_Poisson} since we have
\begin{align*}
   &\bE_{\beta_\star}[|\cO_\xi(Y)|^k|X]  \leq 2^{k-1}\bE_{\beta_\star}\left[|m_\xi(Y)|^k + |M_\xi(Y)|^k |X\right] \\ 
   &\leq 2^{k-1} C_{1} \left(\bP_{\beta_\star}[m_\xi(Y)\geq 1|X] + \bP_{\beta_\star}[M_\xi(Y)\geq 1|X]\right)\leq 2^{k-1} C_{1}\bP_{\beta_\star}[\xi \mbox{ is active}|X]
\end{align*}where $C_{1}$ is the constant appearing in Lemma~\ref{lem:moment_min_Poisson} for $k=3$. Next, since we assumed the intensity of the $Y_\bi$'s to be uniformly bounded from above we obtain the Lipschitz property of the Hessian matrices of the loss functions.

\subsubsection{Final step to the proof of Theorem~\ref{theo:CLT_polyads_estim_Graham}: apply Theorem~\ref{theo:general_directional_CLT}} 
\label{sub:final_step_apply_theorem_theo:general_clt}
The statements in this section hold for almost all $X$, contionally on $X$. Both convergence in distribution and in probability are therefore given w.r.t. to the  probability distribution $\bP_{\beta_\star}$ conditionally on $X$ that we denote by $\bP_{\beta_\star}[\cdot|X]$. Let us now gather all the results we obtained previously and apply Theorem~\ref{theo:general_directional_CLT}. Let $c\in\bR^p$.

We obtained in \eqref{eq:final_CLT} (applied to $c = \Gamma^{-1}c$) that 
\begin{equation}\label{eq:final_conv_gradient}
    \frac{{N_a}}{\sqrt{n\alpha_{1, N}}}\frac{\inr{ \nabla \hat Q_N(\beta_\star), \Gamma^{-1} c}}{\sqrt{c^\top\Gamma^{-1}\tilde \Delta_N \Gamma^{-1}c}}  \overset{d}{\to} \cN(0,1).
\end{equation}As a consequence, the directional CLT for the gradient granted in \textit{item~$4^\prime$} of Theorem~\ref{theo:general_directional_CLT} is satisfied with $\hat V_n$ defined as the deterministic matrix
\begin{equation*}
    \hat V_n := \frac{n \alpha_{1,N}}{{N_a}} \tilde \Delta_N  = \frac{1}{N_a} \sum_i \bE\left[\bar{s}_{\bi}\bar{s}_{\bi}^\top | X\right]
\end{equation*}and where $N_a$ plays the role of the number of data (called $n$ in Theorem~\ref{theo:general_directional_CLT}).

In order to apply Theorem~\ref{theo:general_directional_CLT}, we need to check that there exists some absolute constant $c_0>0$ and $n_0$ such  that for all $n\geq n_0$,
\begin{equation}\label{eq:assumption_on_smallest_sing_var}
    \norm{\hat V_n^{1/2}\Gamma^{-1}c}_2 = \sqrt{\frac{ n\alpha_{1,N}}{{N_a}}} \norm{\tilde \Delta_N^{1/2} \Gamma^{-1}c}_2\geq c_0.
\end{equation}In \eqref{eq:var_UN_star_equiv_Delta_N}, we proved that
\begin{equation*}
  \bE\left[\inr{U_N^*,\Gamma^{-1} c}^2|X\right] = \frac{n \alpha_{1,N}}{{N_a}^2} \norm{\tilde \Delta_N^{1/2} \Gamma^{-1}c}_2^2 = \frac{1}{N_a}\norm{\hat V_n^{1/2} \Gamma^{-1}c}_2^2
\end{equation*} and we showed in Section~\ref{sub:hajeck} that $\bE\left[\inr{U_N^*,\Gamma^{-1} c}^2|X\right] $ is asymptotically equivalent to $T_1$ defined in \eqref{eq:T1} as $T_1 = N^2 \Delta_{1,N}(X)/N_a^2$ where $\Delta_{1,N}$ is defined in \eqref{eq:Delta_q_N} where $c$ is taken equal to $\Gamma^{-1} c$ in this equation.  Finally, under Assumption \ref{ass:second_order_geometry}, we have $N_a T_1$ which is lower bounded by an absolute constant for $n$ large enough and so \eqref{eq:assumption_on_smallest_sing_var} follows.

We apply Theorem~\ref{theo:general_directional_CLT} to the convex (random) loss functions $\hat Q_N: \beta\in\bR^p \to (1/{N_a}) L(Y|X,\beta)$ and its associated risk function $\beta \to Q_N(\beta) = \bE_{\beta_\star}[\hat Q_N(\beta)|X]$. 

It follows from Lemma~\ref{lem:derivatives} that for $V_\xi(y)= \bV_{\beta_\star}\left(m_\xi(Y^\prime)|X, Y^\prime\in \cO_\xi(y)\right)$, the Hessian of the risk function is 
\begin{equation*}
    \nabla^2 Q_N(\beta) = \frac{1}{N_a}\sum_{\xi\in\Xi}\bE_{\beta_\star} \left[V_\xi(Y)|X \right] \tilde X_\xi \tilde X_\xi^\top.
\end{equation*}Hence, the risk function is convex and twice differentiable. It follows from Assumption~\ref{ass:consistency} that  $\beta_\star$ is the unique minimum of $Q_N$. Therefore, Assumption~\ref{ass:consistency} together with Assumption~\ref{ass:technical_CLT} shows that the first item of Theorem~\ref{theo:general_CLT} is satisfied. The second item of Theorem~\ref{theo:general_CLT} follows from Lemma~\ref{lem:derivatives}  and Assumption~\ref{ass:consistency}. The third item of Theorem~\ref{theo:general_CLT} has been proved in Section~\ref{sub:uniform_convergence_of_the_hessian_over_compact_sets}. Finally item $4^\prime$ from Theorem~\ref{theo:general_directional_CLT} was proved right above. It follows from Theorem~\ref{theo:general_directional_CLT} that 
\begin{equation*}
     \frac{N_a}{\sqrt{n\alpha_{1, N}}}\frac{\inr{\hat \beta_\Xi - \beta_\star,c}}{\sqrt{c^\top \Gamma^{-1}\tilde \Delta_N \Gamma^{-1}c}}=\frac{\inr{ \sqrt{N_a}\left(\hat \beta_\Xi - \beta_\star\right),c}}{\norm{V_n^{1/2}\Gamma^{-1}c}_2}   \overset{d}{\to} \cN(0,1).
\end{equation*}








\section{Auxiliary results}
In this section, we collect several technical tools used to prove the consistency result from Theorem~\ref{theo:consistency_tetrads} and the asymptotic normality of Theorem~\ref{theo:CLT_polyads_estim_Graham}.

\subsection{Characterization of degree-preserving transformations}\label{subsec:degree_preserving}

The aim of this subsection is to proof the counterpart of Proposition \ref{prop:from:degrees:to:polyads}, which assumes $\sG = \sG^{\max}$. We need to show that if two graphs $y,y' \in \ZZ^\sI$ satisfy $\delta(y) = \delta(y')$, then it is possible to find a finite number $m$, polyads $\xi_1, \dots, \xi_m$ and integers $r_1,\dots,r_m$ such that
\[ y' = T_{\xi_1}^{r_1} \circ \cdots \circ T_{\xi_m}^{r_m}(y). \]

We will prove by induction on both the dimension $D$ and the component sizes $n_1,\dots,n_D$. We start with the following lemma, that is the base case:

\begin{lemma}
    If two graphs $y \neq y' \in\ \ZZ^\sI$ with $n_1=n_2= \dots = n_D = 2$ satisfy $\delta(y) = \delta(y')$, then there exist a polyad $\xi$ and an integer $r$ such that $y' = T_\xi^r(y)$.
\end{lemma}
\begin{proof}Let $r = y'_{1\dots 1} - y_{1\dots 1}$. In a graph with $n_1 = n_2 = \dots = n_D = 2$ there are only $2^D$ polyads. Take
\[ \xi = \begin{pmatrix}
    1 & 1 & \cdots & 1\\ 
    2 & 2 & \cdots & 2\\ 
\end{pmatrix}.\]

We now check that $y' = T_\xi^r(y)$. First, notice that $T_\xi^r(y)_{1\dots 1} = y_{1\dots 1} + s_\xi(1\dots1) r = y_{1\dots 1} + r = y'_{1\dots 1}$. Now we look at $\bi = 1\dots 1 2$ (i.e., the last index is $2$). We know that $T_\xi^r$ preserves degrees. Looking at the degree associated with $i_1=\dots=i_{D-1}=1$ we have
\[ y'_{1\dots 1 1} + y'_{1\dots 1 2} = y_{1\dots 1 1} + y_{1\dots 1 2} = T_\xi^r(y)_{1\dots 1 1} + T_\xi^r(y)_{1\dots 1 2} = y'_{1\dots 1 1} + T_\xi^r(y)_{1\dots 1 2} \]
and so $T_\xi^r(y)_{1\dots 1 2} = y'_{1\dots 1 2}$. The same argument holds to proof the equalities for all $\bi$'s that differ from $1\dots1$ in one index. Once those are proved they can be used to ensure equality for the $\bi$'s that differ in two indices from $1\dots1$, but only by one index from the already shown. And so on, until we equality is proved for $2\dots2$.
\end{proof}

Now we show that the proposition holds for any sizes $n_1,n_2 \geq 2$ when $D=2$. We start noticing that if $n_1=n_2=2$ it is proven by the previous lemma. By induction we assume it is proven for $n_1, n_2 \geq 2$ and we show for $n_1+1, n_2$ (the case $n_1, n_2+1$ follows permuting the order of the indices). The strategy will be to make $y_\bi = y'_\bi$ for all $\bi = (i_1, i_2)$ such that $i_1 = n_1+1$, once this is done we can get back to the induction assumption since the degrees of the $n_1 \times n_1$ subgraphs obtained removing the indices with $i_1 = n_1+1$ from $y'$ and from the transformed $y$ must be the same. Let $i_2 \in [n_2]$, $i_2 \neq 1$, define the polyad
\[ \xi = \begin{pmatrix}
    n_1+1 & i_2\\ 
    1 & 1\\ 
\end{pmatrix},\]
notice that it can be applied $y'_{n_1+1, i_2} - y_{n_1+1, i_2}$ times to make the index $\bi = (n_1+1, i_2)$ of the transformed $y$ match the one of $y'$. Let $T$ be the composition of these operations for $i_2 = 2, \dots, n_2$. It remains to show that $y'_{n_1+1, 1} = T(y)_{n_1+1, 1}$. The degree equality associated with $i_1 = n_1+1$ yields
\[ y'_{n_1+1,1} + \sum_{i_2=1}^{n_2}y'_{n_1+1, i_2} = y_{n_1+1,1} + \sum_{i_2=1}^{n_2}y_{n_1+1, i_2} = T(y)_{n_1+1,1} + \sum_{i_2=1}^{n_2}T(y)_{n_1+1, i_2} = T(y)_{n_1+1,1} + \sum_{i_2=1}^{n_2}y'_{n_1+1, i_2}\]
and so $y'_{n_1+1,1} =  T(y)_{n_1+1,1}$. We have thus proven the induction step for $D=2$.

Finally we move to $D > 2$. The base case with $n_1=\cdots=n_D =2$ already given by the Lemma. By induction assume that it holds for $n_1,\dots,n_D \geq 2$ and for dimension $D-1$ too. We now show it holds for $n_1+1,n_2,\dots,n_D$. Notice that $\{y_\bi : i_1 = n_1+1\}$ and $\{y'_\bi : i_1 = n_1+1\}$ are $n_2 \times \dots \times n_D$ graphs and, since the degrees of $y$ are the same as the ones of $y'$, they have the same degrees. For example, the degree associated with $(i_2, \dots, i_{D-1})$ in the subgraphs is the one previously associated with $(n_1+1, i_2, \dots, i_{D-1})$. Therefore, there are $D-1$-dimensional polyads that can turn one subgraph into the other. To obtain $D$-dimensional polyads from these we can simply concatenate with the entry $i_1 = n_1+1$ and $i_1' = 1$. This finishes the proof.

\subsection{General results on orbits, $m_\xi$ and $M_\xi$ and the proof of Lemma~\ref{lem:derivatives}} 
\label{sec:appendix_a}
We start with two simple observations on the orbits that justifies its name: an orbit is an equivalent class.

\begin{lemma}\label{lem:equiv_polyads}
    Let $\xi$ be a polyad. For all $y,z\in\bN^{\cI}$, the following are equivalent:
    \begin{itemize}
        \item[(a)] $y\in\cO_\xi(z)$
        \item[(b)] $\cO_\xi(y) = \cO_\xi(z)$.  
    \end{itemize}
\end{lemma}

\begin{lemma}\label{lem:m_M_orbits}Let $y,y^\prime\in\bN^{\cI}$ and assume that $y^\prime \in \cO_\xi(y)$. Let $-m_\xi(y)\leq r \leq M_\xi(y)$ then the following are equivalent:
\begin{itemize}
    \item[(i)] $y^\prime=y^r$
    \item[(ii)] $m_\xi(y^\prime) = m_\xi(y) + r$
    \item[(iii)] $M_\xi(y^r) = M_\xi(y)-r$.   
\end{itemize}
\end{lemma}

\begin{proposition}\label{prop:diff_in_diff}(the diff-in-diff property of $s_\xi$)
     Let $g$ be an ordered tuple of numbers in $[D]$ such that $g\neq \emptyset$ and $g\neq (1, 2, \ldots, D)$. We have for all $(i_d)_{d\in g}\in \cI_g$,
     \begin{equation*}
         \sum_{\bi=(i_d)_{d=1}^D:(i_d)_{d\in\bar g}\in\cI_{\bar g}} s_\xi(\bi) = 0
     \end{equation*}where $\cI_{g} = \prod_{d\in g}[n_d]$ and $\cI_{\bar g} = \prod_{d\notin g}[n_d]$\footnote{Here, we identify the tuple $g$ with the set of  elements made of the coordinates in $g$.}.
     In particular, for all $\theta = \left(\theta_{g(\bi)}^g \in \R : \bi \in \sI, g \in \sG\right)$ and $g\in\cG$, we have 
\begin{equation}
    \label{eq:preserve_suf}
   \sum_{\bi} s_\xi(\bi) \theta_{g(\bi)}^g =  \sum_{\bi : s_\xi(\bi) \neq 0} s_\xi(\bi) \theta_{g(\bi)}^g = 0.
\end{equation}
 \end{proposition} 

\begin{proof}
     To verify the first equality, let $\bi'$ be fixed such that $s_\xi(\bi') \neq 0$. There are $2^{D-|g|}$ choices of $\bi$ such that $g(\bi) = g(\bi')$ and $s_\xi(\bi) = 0$, half of these choices have $s_\xi(\bi) = 1$ and the other half $s_\xi(\bi) = -1$ since it suffices to flip one index $d \not \in g$ to obtain a bijection between positive and negative signs. The second inequality is a direct consequence of the first property.
 \end{proof}

\begin{remark}
      The diff-in-diff property also holds if we define polyads with equal nodes for some axis $d\in[D]$, i.e. we may not assume that $j_d\neq j_d^\prime$ for all $d\in[D]$ - the minimal requirement is to have at least $2$ axis with two different nodes --  i.e. a tetrads. This may be useful when the number of polyads is too large regarding computational cost. 
  \end{remark}

Next we prove Lemma~\ref{lem:derivatives}.
\begin{proof}
    Let $\xi$ be a polyad, $\beta\in\bR^p$ and $y=(y_\bi)_{\bi}$. Let $r\in\{-m_\xi(y), \ldots,  M_\xi(y)\}$. It follows from the diff-in-diff property of the sign function from Proposition~\ref{prop:diff_in_diff} that
\begin{align}\label{eq:proba_distrib_m(Y)}
  \nonumber&\bP_{\beta}[Y=y^r|X, Y \in\cO_\xi(y)] = \frac{\bP[Y=y^r|X]}{\sum_{R=-m_\xi(y)}^{M_\xi(y)} \bP[Y=y^R|X]} = \frac{\Pi_{\bi\in\cE(\xi)} \exp(-\lambda_\bi) \lambda_\bi^{y_\bi^r}/y_\bi^r! }{\sum_{R=-m_\xi(y)}^{M_\xi(y)} \Pi_{\bi\in\cE(\xi)} \exp(-\lambda_\bi) \lambda_\bi^{y_\bi^R}/y_\bi^R!}\\
  \nonumber&=\left[\sum_{R=-m_\xi(y)}^{M_\xi(y)} \Pi_{\bi\in\cE(\xi)} \frac{y_\bi^r!}{y_\bi^R!}  \lambda_\bi^{y_\bi^R-y_\bi^r}\right]^{-1} = \left[\sum_{R=-m_\xi(y)}^{M_\xi(y)} \exp(-a_r)\exp(a_R)\exp[(R-r)\inr{\tilde X_\xi, \beta}]\right]^{-1}\\ 
  &= \frac{\exp(a_r + r\inr{\widetilde{X}_\xi, \beta})}{\sum_{R=-m_\xi(y)}^{M_\xi(y)}\exp(a_R+R \inr{\widetilde{X}_\xi, \beta})}   
    \end{align}  where $a_R = \ln(\Pi_{\bi\in\cE(\xi)}y_\bi!/y_\bi^R!)$ for all $R\in\{-m_\xi(y), \ldots,  M_\xi(y)\}$. Next, let $m\in\{0, \ldots, m_\xi(y) +   M_\xi(y)\}$.  We observe that if $Y\in\cO_\xi(y)$ then $m_\xi(Y)=m$ iff $Y=y^{m-m_\xi(y)}$. As a consequence, it follows from  \eqref{eq:proba_distrib_m(Y)} that
    \begin{align*}
        \ln \Pr_\beta( m_\xi(Y) = m |X,  Y \in \sO_\xi(y) ) = \ln \Pr_\beta( Y = y^{m-m_\xi(y)} |X,  Y \in \sO_\xi(y) ) = a_{m-m_\xi(y)} + (m-m_\xi(y))\inr{\widetilde{X}_\xi, \beta} - \ell_\xi(y|X, \beta)
    \end{align*}where we used that 
    \begin{equation}\label{eq:loss_xi_a_R}
        \ell_\xi(y|X, \beta) = -\ln \Pr_\beta(Y = y |X,  Y \in \sO_\xi(y) ) = \ln\left[\sum_{R=-m_\xi(y)}^{M_\xi(y)}\exp(a_R+R \inr{\widetilde{X}_\xi, \beta})\right].
    \end{equation}This proves the last result of Lemma~\ref{lem:derivatives} regarding the probability distribution of $m_\xi(Y)$ given $Y\in\cO_\xi(y)$ and $X$.

 The gradient of $\beta\to \ell_\xi(y|X, \beta)$ can be derived from  \eqref{eq:loss_xi_a_R} and \eqref{eq:proba_distrib_m(Y)}:
\begin{align*}
    &\nabla_\beta \ell_\xi(y|X, \beta) = \frac{\sum_{r=-m_\xi(y)}^{M_\xi(y)} r \exp(a_r+r \inr{\widetilde{X}_\xi, \beta})}{\sum_{R=-m_\xi(y)}^{M_\xi(y)}\exp(a_R+R \inr{\widetilde{X}_\xi, \beta})} \widetilde{X}_\xi = \sum_{r=-m_\xi(y)}^{M_\xi(y)} r \bP_{\beta}[Y=y^r|X, Y \in\cO_\xi(y)] \widetilde{X}_\xi\\
    & = \sum_{m=0}^{m_\xi(y) + M_\xi(y)} (m-m_\xi(y)) \bP_{\beta}[m_\xi(Y)=m|X, Y \in\cO_\xi(y)] \widetilde{X}_\xi =  \left(  \Ex_\beta\left[ 
m_\xi(Y) |X,  Y \in \sO_\xi(y) \right] - m_\xi(y) \right) \widetilde{X}_\xi 
\end{align*}where we used the fact that if $Y\in\cO_\xi(y)$ then $Y=y^r$ iff $m_\xi(Y) = r+m_\xi(y)$. Finally, we derive the formula for the Hessian of  $\beta\to \ell_\xi(y|X, \beta)$:
\begin{align}\label{eq:computation_Hessian}
 &\nabla^2_\beta \ell_\xi(y|X, \beta) = \left[\frac{\sum_{r=-m_\xi(y)}^{M_\xi(y)} r^2 \exp(a_r+r \inr{\widetilde{X}_\xi, \beta})}{\sum_{R=-m_\xi(y)}^{M_\xi(y)}\exp(a_R+R \inr{\widetilde{X}_\xi, \beta})}  - \left(\frac{\sum_{r=-m_\xi(y)}^{M_\xi(y)} r \exp(a_r+r \inr{\widetilde{X}_\xi, \beta})}{\sum_{R=-m_\xi(y)}^{M_\xi(y)}\exp(a_R+R \inr{\widetilde{X}_\xi, \beta})}   \right)^2 \right]  \widetilde{X}_\xi \widetilde{X}_\xi^\top\\ 
 \nonumber & = \mathbb{V}_\beta\left[ 
m_\xi(Y) |X,  Y \in \sO_\xi(y) \right] \widetilde{X}_\xi\widetilde{X}_\xi^T
\end{align}where we used again that if $Y\in\cO_\xi(y)$ then $Y=y^r$ iff $m_\xi(Y) = r+m_\xi(y)$ and that $m_\xi(Y)$ has the same variance as $m_\xi(Y) + m_\xi(y)$.

To prove similar results for the gradient and the Hessian in terms of $M_\xi(Y)$ we also rely on Lemma~\ref{lem:m_M_orbits} and use the same arguments as above. 
\end{proof}

\subsection{Proof of Lemma~\ref{lem:suf_beta}} 
\label{sec:appendix_a_proof_of_lemma_lem:suf_beta}

\begin{proof}[Proof of Lemma~\ref{lem:suf_beta}:] 
     Denote by $B_2 = \{\theta\in\bR^d : \norm{\theta}_2\leq 1\}$ the unit ball with respect to (w.r.t.) the $\ell_2^d$-norm. Let $\eps>0$ be such that $\theta_0+3\eps B_2\subset \Theta$. We want to show that with probability approaching $1$ (w.p.a.1), $\hat \theta_n\in \theta_0+\eps B_2$. Let $\tilde \theta_n\in\argmin_{\theta\in \theta_0 + 2\eps B_2} \hat Q_n(\theta)$. Since $\theta_0 + 2\eps B_2$ is a compact set, it follows from Lemma~\ref{lem:unif_cv} (applied to $\Theta = \theta_0+3\eps \ove{B_2}{\circ}$ and $K = \theta_0+2\eps B_2$) that the following uniform convergence result holds
\begin{equation*}
    \sup_{\theta\in \theta_0 + 2\eps B_2}|\hat Q_n(\theta) - Q_n(\theta)| \overset{p}{\to}  0
\end{equation*} and so, by Lemma~\ref{lem:N_McF_compact} (applied to $\Theta = \theta_0+2\eps B_2$), that $\tilde \theta_n \overset{p}{\to}  \theta_0$. In particular, wpa1, $\tilde \theta_n\in \theta_0 + \eps B_2$.  Let us now place ourselves on the event $\tilde \theta_n\in \theta_0 + \eps B_2$ and let $\theta\notin \theta_0 + 2\eps B_2$. There exists $\theta_1\in \theta_0 + 2\eps S_2$ - where $S_2 = \{\theta\in\bR^d :\norm{\theta}_1= 1\}$ is the unit sphere of the $\ell_2^d$-norm - and $\lambda\geq 1$ such that $\theta = \tilde \theta_n + \lambda(\theta_1-\tilde \theta_n)$. By convexity of $\hat Q_n$, we have
    \begin{equation*}
        \hat Q_n(\theta) - \hat Q_n(\tilde \theta_n) \geq \lambda(\hat Q_n(\theta_1) - \hat Q_n(\tilde \theta_n)) \geq0
    \end{equation*}where the last inequality is due to the fact that $\theta_1\in \theta_0+2\eps S_2\subset \theta_0+2\eps B_2$ and $\tilde \theta_n$ minimizes $\hat Q_n$ over $\theta_0+2\eps B_2$. As a consequence, $\tilde \theta_n$ minimizes $\hat Q_n$ over $\Theta$ and by the uniqueness of $\hat \theta_n$, we have $\hat \theta_n = \tilde \theta_n$. This concludes the proof since $\tilde \theta_n$ is a consistent estimator of $\theta_0$. 
\end{proof}

The proof of Lemma~\ref{lem:suf_beta} provided above  is based on three ingredients: 1) convexity of $\hat Q_n$, 2) the uniform convergence in probability result over the compact set $\theta_0 + 2\eps B_2$ derived from  Lemma~\ref{lem:unif_cv}  and, 3) the consistency result over the compact model $\theta_0 + 2\eps B_2$ that follows from  Lemma~\ref{lem:N_McF_compact}. We start with the proof of the latter result which is an adaptation of Theorem~2.1 from \cite{newey1994large}.

\begin{lemma}\label{lem:N_McF_compact}
    Let $\Theta$ be a non empty and convex set in $\bR^d$. Let $(\hat Q_n)_n$ be a sequence of random functions and $(Q_n)_n$ be a sequence of (deterministic) convex functions all defined on $\Theta$. We assume that:
    \begin{itemize}
     \item[(a)] with probability approaching $1$, there exists $\hat \theta_n\in\Theta$ minimizing $\hat Q_n$ over $\Theta$,
         \item[(b)] there exists $\theta_0\in\Theta$ and $\eps_0>0$ such that $\theta_0+\eps_0 B_2\subset \Theta$ and for all $0<\eps\leq \eps_0$ there exists $\eta>0$ and $n_0$ such that for all $n\geq n_0$ and all $\theta\in\Theta$, if $\norm{\theta - \theta_0}_2 =  \eps$ then $Q_n(\theta) - Q_n(\theta_0)\geq \eta$,
         \item[(c)] $\sup_{\theta\in\Theta}|\hat Q_n(\theta) - Q_n(\theta)| \overset{p}{\to}  0$.
     \end{itemize} Then $\hat \theta_n \overset{p}{\to}  \theta_0$.
\end{lemma}

\begin{proof}[Proof of Lemma~\ref{lem:N_McF_compact}]
   Let $0<\eps\leq \eps_0$ and denote by $\bar{B}=\theta_0+\eps B_2$ the closed $\ell_2^d$-ball centered at $\theta_0$ with radius $\eps$. We want to show that with probability approaching $1$, $\hat \theta_n \in \bar{B}$. To that end, it is enough to show that for $n$ large enough we have $Q_n(\hat \theta_n)< \inf_{\theta\notin \bar B}Q_n(\theta)$. 

We first show that thanks to the convexity of the $Q_n$'s we have  $\inf_{\theta\notin \bar B}Q_n(\theta) = \min_{\theta:\norm{\theta - \theta_0}_2=\eps}Q_n(\theta)$. First, it follows from convexity and \textit{(b)} that $\theta_0$ is the unique global minimizer of $Q_n$ over $\Theta$ for all $n\geq n_0$. Second, if $\theta\notin \bar B$, there exists $\lambda\geq 1$ and $\theta_1\in \theta_0+\eps S_2$ (where $S_2$ is the unit $\ell_2^d$-sphere) such that $\theta = \theta_0 + \lambda(\theta_1-\theta_0)$ and, from the convexity of $Q_n$, we have 
\begin{equation*}
    Q_n(\theta) - Q_n(\theta_0)\geq \lambda (Q_n(\theta_1) - Q_n(\theta_0))\geq Q_n(\theta_1) - Q_n(\theta_0)
\end{equation*}where the last inequality follows because $\theta_0$ is a global minimizer of $Q_n$ and $\lambda\geq1$. As a consequence, $Q_n(\theta)\geq Q_n(\theta_1)$ and so, by continuity of $Q_n$, we obtain that $\inf_{\theta\notin \bar B}Q_n(\theta) = \min_{\theta:\norm{\theta - \theta_0}_2=\eps}Q_n(\theta)$. 

   It follows from the property \textit{(b)} of $(Q_n)_n$ that there exists $n_0$ such that $\eta>0$ where $$\eta := \min_{n\geq n_0}\min_{\theta:\norm{\theta - \theta_0}_2=\eps}Q_n(\theta) - Q_n(\theta_0).$$ As a consequence, with probability approaching $1$, the following holds from uniform convergence:
\begin{equation*}
    Q_n(\hat \theta_n)\leq \hat Q_n(\hat \theta_n) + \frac{\eta}{3} < \hat Q_n(\theta_0) + \frac{2\eta}{3}\leq  Q_n(\theta_0)+\eta \geq \min_{\theta:\norm{\theta - \theta_0}_2=\eps}Q_n(\theta) = \inf_{\theta\notin \bar B}Q_n(\theta).
\end{equation*}  
\end{proof}

The next result shows that uniform convergence in probability over a compact follows from pointwise convergence in probability thanks to the convexity assumption. 

\begin{lemma}\label{lem:unif_cv}
Let $\Theta$ be a non empty, open and convex set in $\bR^d$. Let $(\hat Q_n)_n$ be a sequence of convex random functions and $(Q_n)_n$ be a sequence of (deterministic) convex functions all defined on $\Theta$. We assume that: 
\begin{enumerate}
    \item for every  $\theta\in\Theta$ there exists $L>0$ such that $\sup_{n}Q_n(\theta)\leq L$,
    \item there exists $L_0\in\bR$ and $\theta_1\in\Theta$ such that $\inf_n Q_n(\theta_1)\geq L_0$.
    \item for all \( \theta \in \Theta \), as $n$ tends to infinity, \( \hat Q_n(\theta) - Q_n(\theta) \overset{p}{\to}  0 \).

\end{enumerate}
Then, for any compact set $K$ in $\Theta$,
$$\sup_{\theta\in K}|\hat Q_n(\theta) - Q_n(\theta)| \overset{p}{\to}  0.$$
    
\end{lemma}

\begin{proof}[Proof of Lemma~\ref{lem:unif_cv}] Let $K$ be a compact set in $\Theta$. To show the uniform convergence in probability over the compact set $K$, it is enough to show that for any increasing sequence $(\phi_n)_n$ of integers we can extract a sub-sequence  $(\psi_n)_n\subset (\phi_n)_n$  along which $\sup_{\theta\in K}|\hat Q_{\psi_n}(\theta) - Q_{\psi_n}(\theta)|$ tends to $0$ almost surely. 

Let $\Theta^\prime$ be a countable dense subset of $\Theta$. According to the diagonalization argument from \cite{Andersen1982} (recalled in Lemma~\ref{lem:Andersen82} below), there exists a sub-sequence $(\psi_n)_n$ of $(\phi_n)_n$ such that for all $\theta\in\Theta^\prime, |\hat Q_{\psi_n}(\theta) - Q_{\psi_n}(\theta)| \overset{a.s.}{\to}  0.$  

Next it follows from  Lemma~\ref{lem:extension_10_6} that almost surely for $n$ large enough, $\hat Q_{\psi_n} - Q_{\psi_n}$ is Lipchitz on the compact set $K$.  Then, we conclude with Lemma~\ref{lem:rockafellar_70} that   $\sup_{\theta\in K}|\hat Q_{\psi_n}(\theta) - Q_{\psi_n}(\theta)|$ tends to $0$ almost surely. This is true for any increasing sequence $(\phi_n)_n$ and so this shows the uniform over $\Theta$ convergence in probability. 
\end{proof}

The proof of Lemma~\ref{lem:suf_beta} requires to revisit classical results from convex analysis and asymptotic statistics. We first start with Theorem~10.6 and Theorem~10.8 from \cite{rockafellar-1970a} that need to be adapted to our setup.

\begin{lemma}\label{lem:extension_10_6}(Adapted from Theorem~10.6 in \cite{rockafellar-1970a}) Let $\Theta$ be a non empty, open and convex set in $\bR^d$. Let $(f_n)_n$ and $(g_n)_n$ be two sequences of convex functions defined on $\Theta$.  We assume that there exists $\Theta^\prime\subset \Theta$ such that $\Theta\subset {\rm conv}(\bar \Theta^\prime)$ and
    \begin{itemize}
        \item  for all $\theta\in\Theta^\prime$, $(f_n-g_n)(\theta)\to 0$,
        \item for all $\theta\in\Theta^\prime$, there exists $L>0$ such that $\sup_{n}g_n(\theta)\leq L$,
        \item there exists $L_0\in\bR$ and $\theta_1\in\Theta$ such that $\inf_n g_n(\theta_1)\geq L_0$.
    \end{itemize}
    Then, for any non empty compact set $K$ in $\Theta$, there exists $L>0$ such that  for $n$ large enough, we have for all $x,y\in K$, $$|(f_n-g_n)(x) - (f_n - g_n)(y)|\leq L\norm{x-y}_2.$$ 
\end{lemma}

\begin{proof} Let $K$ be a non-empty compact set in $\Theta$.
Since $(g_n)_n$ is a sequence of convex functions satisfying \textit{(a)} and \textit{(b)} from Theorem~10.6 in \cite{rockafellar-1970a}, it follows from the latter theorem that $(g_n)_n$ is equi-Lipschitzian relative to $K$. Next, since $(f_n-g_n)(\theta)\to 0$ for all $\theta\in\Theta^\prime$ and $(g_n(\theta))_n$ is uniformly (in $n$) bounded from above for all  $\theta\in\Theta^\prime$ and there is some $\theta_1$ for which $(g_n(\theta_1))$ is bounded from below,  then for all $\theta\in\Theta^\prime$, $(f_n(\theta))_n$ is also bounded from above.  Hence, it follows from Theorem~10.6 in \cite{rockafellar-1970a} that $(f_n)_n$ is equi-Lipschitzian relative to $K$. We conclude that the sequence of differences $(f_n-g_n)_n$ is also equi-Lipschitzian relative to $K$. 
\end{proof}

\begin{lemma}\label{lem:rockafellar_70}(adapted from Theorem~10.8 from \cite{rockafellar-1970a}) Let $\Theta$ be a non empty set in $\bR^d$.  Let $(f_n)_n$ and $(g_n)_n$ be two sequences of functions defined on $\Theta$. We assume that for any compact set $K$ in $\Theta$, there exists $L>0$ such that for all $n$ large enough and all $x,y\in K$, $$|(f_n-g_n)(x) - (f_n - g_n)(y)|\leq L\norm{x-y}_2.$$ 
     If $(f_n-g_n)_n$ tends to $0$ pointwise over a dense subset of $\Theta$ then, for any compact set $K$ in $\Theta$, $(f_n-g_n)_n$ converges uniformly over $K$ to $0$.
\end{lemma}

\begin{proof}[Proof of Lemma~\ref{lem:rockafellar_70}]
  For all integer $n$ we let $h_n=f_n-g_n$ and we denote by $\Theta_0$ a dense subset in $\Theta$ onto which $(h_n)_n$ converges pointwise to $0$. Let $K$ be a compact subset in $\Theta$. Let $L>0$ and $n_1$ be such that for all $n\geq n_1$ and all $x,y\in K$, 
\begin{equation}\label{eq:h_n_Lipschitz}
    |h_n(x) - h_n(y)|\leq L \norm{x-y}_2.
\end{equation}Let $\eps>0$. Denote by $\Theta_1$ an $(\eps/(2L))$-net for $K$ in $K\cap \Theta_0$ with respect to the $\ell_2^d$-norm. Since $K$ is compact we can choose $\Theta_1$ to be finite. Since $\Theta_1$ is finite and $(h_n(\theta))_n$ tends to $0$ for all $\theta\in\Theta_1$, there exists $n_2\geq n_1$ such that for all $\theta\in\Theta_1$, $|h_n(\theta)|\leq \eps/2$. Let $\theta\in K$. Let $\theta_1\in\Theta_1$ be such that $\norm{\theta - \theta_1}_2\leq \eps/2$. Then, for all $n\geq n_2$, we have 
\begin{align*}
     |h_n(\theta)| \leq |h_n(\theta) - h_n(\theta_1)| + |h_{n}(\theta_1)| \leq L \norm{\theta-\theta_1}_2 + \frac{\eps}{2}\leq\eps.
 \end{align*} This concludes the proof since the latter holds for all $\theta\in K$.   
\end{proof}

The next result is the diagonalization method used to prove Theorem~II.1 in \cite{Andersen1982} that we reproduce here in our setup for the sake of completeness.

\begin{lemma}\label{lem:Andersen82}(Theorem~II.1 from \cite{Andersen1982})Let $\Theta$ be a non empty set.
    Let $(\hat Q_n)_n$ be a sequence of random functions and $(Q_n)_n$ be a sequence of (deterministic) functions, all defined on $\Theta$. We assume that for all $\theta\in\Theta$, $\hat Q_n(\theta) - Q_n(\theta) \overset{p}{\to}  0$ as $n\to+\infty$. Then, for any countable subset $\Theta^\prime$ in $\Theta$ there exists an increasing sequence $(\psi_n)_n$ of integers such that for all $\theta\in\Theta^\prime$,  
    \begin{equation*}
        \hat Q_{\psi_n}(\theta) - Q_{\psi_n}(\theta) \overset{a.s.}{\to}  0.
    \end{equation*}
\end{lemma}

{\it{Proof of Lemma~\ref{lem:Andersen82}}:} Denote for all $n$, $\hat H_n = \hat Q_n - Q_n $. Denote by $(x_j)_{j\in\bN}$ the sequence of all elements in $\Theta^\prime$.   Since $(\hat H_{n}(x_1))_n$ tends in probability to $0$, we can extract a sub-sequence $(\phi_{n,1})_n$ along which the converge is almost sure. Next, since  $(\hat H_{\phi_{n,1}}(x_2))_n$ tends in probability to $0$, it is possible to further construct a sub-sequence $(\phi_{n,2})_n$ along which the converge is almost sure. We repeat the argument and get that for all integers $k$, it is possible to construct a sequence $(\phi_{n,k})_n$ such that for all $i\in[k]$, $(\hat H_{\phi_{n,k}}(x_i))_n$ tends almost surely to $0$. Once, this construction is done, we move to the diagonalization argument used in the proof of Theorem~II.1 from \cite{Andersen1982}: we construct a new sequence of integers, denoted by $(\psi_n)_n$, by setting $\psi_1$ to be the first element in $(\phi_{n,1})_n$, then $\psi_2$ to be the second one in $(\phi_{n,2})_n$, etc.. Along this new sequence $(\psi_n)_n$ we have for all $i\in\bN$, $(\hat H_{\psi_n}(x_i))_n$ converges almost surely to $0$. $\square$

\subsection{Some properties of the loss functions, its differential and Hessian}

We first start with two results used to prove the consistency result from Theorem~\ref{theo:consistency_tetrads}. 

\begin{lemma}\label{lem:sampling_device_for_consistency}Let $\beta\in\bR^p$. There exists an absolute constant $c_0$ such that for all $\xi,\xi^\prime \in\Xi$, we have
$$
\bE_{\beta_\star}\left[(\ell_\xi(Y | X, \beta)-\bE_{\beta_\star}\ell_\xi(Y | X, \beta))(\ell_{\xi^\prime}(Y | X, \beta)-\bE_{\beta_\star}\ell_{\xi^\prime}(Y | X, \beta)) |X\right] \leq C_0 \bP_{\beta_\star}\left[\xi \mbox{ and  } \xi^\prime  \mbox{ are both active}|X\right]
$$  where
\begin{equation}\label{eq:def_C_1}
    C_0 = c_0 \left(1+ L_N^2 \norm{\beta}^2_2 + \ln^2(\bar{\lambda} + 1)\right) (\bar{\lambda} + 1)^2,
\end{equation} $L_N = \max_\xi \norm{\widetilde X_\xi}_2$ and  $\bar \lambda = \max\left(\lambda_\bi(\beta_\star,\theta^\sG):\bi\in\cI\right)$.
\end{lemma}

\begin{proof}
Since $\ell_\xi(Y|X,\beta)\geq0$ - as the negative log of a probability - the correlation between $\ell_\xi$ and $\ell_{\xi^\prime}$ is less than the expectation of their product:
\begin{align}\label{eq:start_consistency}
\bE_{\beta_\star}\left[(\ell_\xi(Y | X, \beta)-\bE_{\beta_\star}\ell_\xi(Y | X, \beta))(\ell_{\xi^\prime}(Y | X, \beta)-\bE_{\beta_\star}\ell_{\xi^\prime}(Y | X, \beta)) |X\right] \leq \bE_{\beta_\star}\left[\ell_\xi(Y | X, \beta)\ell_{\xi^\prime}(Y | X, \beta)|X\right].
\end{align}

Then, we observe that for all $y\in\bN^\cI$, when $\xi$ is not active with respect to some $y$ then $\cO_\xi(y)$ contains only one element which is $y$ and so $\bP_{\beta_\star}[Y = y|X, Y \in\cO_\xi(y)]=1$, hence $\ell_\xi(y| X, \beta)=0$. As a consequence, for all $\xi$ and $\xi^\prime$, we have
\begin{equation*}
    \bE_{\beta_\star}\left[\ell_\xi(Y | X, \beta)\ell_{\xi^\prime}(Y | X, \beta)|X\right] = \bE_{\beta_\star}\left[\ell_\xi(Y | X, \beta)\ell_{\xi^\prime}(Y | X, \beta) I\left(\xi \mbox{ and  } \xi^\prime \mbox{ are both active}\right)|X\right].
\end{equation*}However, the loss functions $\ell_\xi$ are not almost surely bounded and so we cannot simply use the last inequality to conlcude; below we deal carefuly with this issue. 

We recall that for every polyads $ \xi = \begin{pmatrix}
            j_1 & j_2 & \cdots & j_D\\ 
            j_1^\prime & j_2^\prime & \cdots & j_D^\prime\\ 
        \end{pmatrix}$ we denote $\cI_\xi = \{j_1, j_1^\prime\}\times \cdots \{j_D, j_D^\prime\}$. We have
\begin{align*}
  &\bE_{\beta_\star}\left[\ell_\xi(Y | X, \beta)   \ell_{\xi^\prime}(Y | X, \beta)|X\right]  =  \sum_{y=(y_\bi)_\bi}\bP_{\beta_\star}[Y = y|X] \left(\ln \bP_{\beta}\left( Y = y |X, Y \in \sO_\xi(y) \right) \right) \left(\ln \bP_{\beta}\left( Y = y |X, Y \in \sO_{\xi^\prime}(y) \right) \right) \\
  &=  \sum_{y=(y_\bi)_\bi}\bP_{\beta_\star}[Y = y|X] \left(\ln \frac{\bP_\beta[Y=y|X]}{\sum_{y^\prime\in \sO_\xi(y)} \bP_\beta[Y=y^\prime|X]}\right) \left(\ln \frac{\bP_\beta[Y=y|X]}{\sum_{y^\prime\in \sO_{\xi^\prime}(y)} \bP_\beta[Y=y^\prime|X]}\right)\\ 
  &= \sum_{y=(y_\bi)_\bi}\bP_{\beta_\star}[Y = y|X] \left(\ln  \sum_{y^\prime\in \sO_\xi(y)} \frac{\bP_\beta[Y=y^\prime|X]}{\bP_\beta[Y=y|X]}\right)\left(\ln  \sum_{y^\prime\in \sO_{\xi^\prime}(y)} \frac{\bP_\beta[Y=y^\prime|X]}{\bP_\beta[Y=y|X]}\right)\\
  & = \sum_{(y_\bi)_{\bi\in\cI_\xi\cup \cI_{\xi^\prime}}} \bP_{\beta_\star}[Y_{\bi}=y_\bi, \forall \bi\in\cI_\xi\cup \cI_{\xi^\prime}|X]\left(\ln\left( \sum_{r=-m_\xi(y)}^{M_\xi(y)} \frac{\bP_\beta[Y_\bi=y_\bi^r, \forall \bi\in\cI_\xi|X]}{\bP_\beta[Y_\bi=y_\bi, \forall \bi\in\cI_\xi|X]}\right)\right) \\
  &\hspace{4cm}\times\left(\ln\left( \sum_{r=-m_{\xi^\prime}(y)}^{M_{\xi^\prime}(y)} \frac{\bP_\beta[Y_\bi=y_\bi^r, \forall \bi\in\cI_{\xi^\prime}|X]}{\bP_\beta[Y_\bi=y_\bi, \forall \bi\in\cI_{\xi^\prime}|X]}\right)\right)
\end{align*}where we recall that for all polyads $\xi$, all $\bi\in\cI$ and all $r\in\{-m_\xi(y),\ldots, M_\xi(y)\}, y_\bi^r = y_\bi+r s_\xi(\bi)$. Next, we use the independence of the $Y_\bi$'s conditionally on $X$ to get that for all $r\in\{-m_\xi(y),\ldots, M_\xi(y)\}$, 
\begin{equation*}
  \frac{\bP_\beta[Y_\bi=y_\bi^r, \forall \bi\in\cI_\xi|X]}{\bP_\beta[Y_\bi=y_\bi, \forall \bi\in\cI_\xi|X]} = \prod_{\bi\in\cI_\xi} \frac{\lambda_\bi^{y_\bi^r} y_\bi !}{\lambda_\bi^{y_\bi} y_\bi^r!} = \prod_{\bi\in\cI_\xi} \lambda_\bi^{rs_\xi(\bi)} \frac{y_\bi !}{y_\bi^r!}
\end{equation*}where $\lambda_\bi = \exp\left(\inr{X_\bi, \beta}+ \sum_{g\in\cG}\theta_{g(\bi)}^g)\right)$.
By the diff-in-diff property of the sign function $s_\xi$ from Proposition~\ref{prop:diff_in_diff}, we have
\begin{equation*}
   \prod_{\bi\in\cI_\xi} \lambda_\bi^{rs_\xi(\bi)}  = \exp(r\inr{\tilde X_\xi, \beta}):= w_{\xi,r}
\end{equation*}where $\tilde X_\xi$ is the polyads feature defined in \eqref{eq:def:diff:and:diff}. Furthermore, we define for all $r\in\{-m_\xi(y),\ldots, M_\xi(y)\}$,
\begin{equation*}
    v_{\xi,r}(y) : = \prod_{\bi\in\cI_\xi} \frac{y_\bi !}{y_\bi^r!} = 
    \left\{
\begin{array}{cc}
   \frac{\prod_{\bi: s_\xi(\bi)=-1}y_\bi(y_\bi-1)\cdots (y_\bi-r+1)}{\prod_{\bi: s_\xi(\bi)=1}(y_\bi+r)\cdots(y_\bi+1)} & \mbox{ if } r>0\\
   \frac{\prod_{\bi: s_\xi(\bi)=1}y_\bi(y_\bi-1)\cdots (y_\bi+r+1)}{\prod_{\bi: s_\xi(\bi)=-1}(y_\bi-r)\cdots(y_\bi+1)} & \mbox{ if } r<0\\
   1 & \mbox{ if } r=0.  
\end{array}
    \right.
\end{equation*}We have, for $M_{\xi} = M_\xi(Y)$, $m_{\xi} = m_\xi(Y)$ and $V_{\xi, r}:= v_{\xi, r}(Y)$ for all $r\in\{-m_\xi, \cdots, M_\xi\}$,
$$
\bE_{\beta_\star}\left[\ell_\xi(Y | X, \beta) \ell_\xi(Y | X, \beta) |X\right]  = \bE_{\beta_\star}\left[ \ln\left(\sum_{r=-m_\xi}^{M_\xi} w_{\xi,r} V_{\xi,r}\right) \ln\left(\sum_{r=-m_{\xi^\prime}}^{M_{\xi^\prime}} w_{\xi^\prime,r} V_{\xi^\prime,r}\right) |X\right].
$$Next, we use that $|\ln(t)|\leq \ln(a)$ for all $a^{-1}\leq t\leq a$ and that for polyads $\zeta$ and all $r\in\{-m_\zeta,\cdots, M_\zeta\}\backslash\{0\}$,
\begin{align*}
    w_{\zeta,r}, w_{\zeta,r}^{-1}\leq \exp\left[(M_\zeta+m_\zeta) |\inr{\tilde X_\zeta, \beta}| \right] \mbox{ and } V_{\zeta,r}, V_{\zeta,r}^{-1}\leq \prod_{\bi\in\cI_\zeta} (Y_\bi+1)^{|r|} \leq \prod_{\bi\in\cI_\zeta} (Y_\bi+1)^{M_\zeta+m_\zeta}
\end{align*} to get the following bound:
\begin{align}\label{eq:main_ccl_L2_loss}
 \nonumber  &\bE_{\beta_\star}\left[\ell_\xi(Y | X, \beta)\ell_{\xi^\prime}(Y | X, \beta)|X\right] \\ 
 \nonumber&\leq \bE_{\beta_\star}\left[\prod_{\zeta\in\{\xi,\xi^\prime\}}\left((M_\zeta+m_\zeta)|\inr{\tilde X_\zeta, \beta}| +\ln\left((M_\zeta+m_\zeta)\prod_{\bi\in\cI_\xi} (Y_\bi+1)^{M_\zeta+m_\zeta}\right)I(M_\zeta+m_\zeta\geq1)\right)|X\right] \\
 \nonumber&\leq \bE_{\beta_\star}\left[\prod_{\zeta\in\{\xi,\xi^\prime\}}\left((M_\zeta+m_\zeta)\left(|\inr{\tilde X_\zeta, \beta}|+1\right) + (M_\zeta+m_\zeta)\ln\left(\prod_{\bi\in\cI_\xi} (Y_\bi+1)\right)\right)|X\right] \\
 &\leq \bE_{\beta_\star}\left[(M_\xi+m_\xi)(M_{\xi^\prime}+m_{\xi^\prime})\prod_{\zeta\in\{\xi,\xi^\prime\}}\left(L_N\norm{\beta}_2+1 + \sum_{\bi\in \cI_\zeta} \ln\left(Y_\bi+1\right) \right)|X\right] 
\end{align}where $L_N = \max_\xi \norm{\tilde X_\xi}_2$.  


Next, we remark that all of the four terms in the product $(M_\xi+m_\xi)(M_{\xi^\prime}+m_{\xi^\prime})$ from \eqref{eq:main_ccl_L2_loss} are product of two minimum of Poisson variables. To handle \eqref{eq:main_ccl_L2_loss} we rely on the next lemma.

\begin{lemma}\label{lem:moment_min_Poisson} There exists an absolute constant $c_0>0$ such that the following holds.
Let $(U_i)_{i\in\bN}$ be a sequence of independent Poisson variables with intensities $(\lambda_i)_{i\in\bN}$. Let $I_1,I_2,I_3,I_4\subset \bN$ be such that $I_1\cap (I_2\cup I_4)=\emptyset$ and $I_3\cap (I_2\cup I_4)=\emptyset$ (but possibly $I_1\cap I_3\neq\emptyset$ and $I_2\cap I_4\neq\emptyset$). Let $\bar{\lambda}$ be such that  $\lambda_i\leq \bar{\lambda}$ for all $i\in \cup_{k=1,2,3,4}I_k$. We have for all $i,j\in \cup_{k=1,2,3,4}I_k$,
\begin{equation*}
  \bE\left[ (m_1+m_2)(m_3 +m_4 ) \ln(U_i+1)\ln(U_j+1)\right]\leq C_0^2 \bP\left[m_1 + m_2 \geq 1 \mbox{ and  } m_3+m_4 \geq 1\right]
\end{equation*}and
\begin{equation*}
 \bE\left[ \ln(m_1+m_2)\right] \leq \bE\left[m_1+m_2\right] \leq C_0 \bP\left[m_1 + m_2 \geq 1\right]   
\end{equation*} where $m_k = \wedge_{i\in I_k}U_i$ and $C_0 = c_0 (\bar{\lambda}+1) \ln(\bar\lambda+1)$.

For all integer $p\geq1$, we also have $ \bE[m_1^p]\leq C_{1} \bP[m_1\geq1] \mbox{ where } C_{1} =c_0 (p^2 + \bar{\lambda}^p)\bP[m_1\geq1].$
\end{lemma}

\begin{proof}  In this proof, we will repeatedly use the following deviation bound for Poisson variable that follows from a Cram{\'e}r-Chernoff method (see for instance Section~2.2, p.23 in \cite{MR3185193}): let $U$ be a Poisson variable with parameter $\lambda\geq0$ then for all $k\geq (e-1)\lambda$,
\begin{equation}\label{eq:cramer_chernoff_poisson}
    \bP[U\geq k]\leq \exp(-\lambda)\left(\frac{e \lambda}{k}\right)^k.
\end{equation}

Let $I\subset \cup_{k=1,2,3,4}I_k$ and $M = \wedge_{i\in I}U_i$. We first establish some preliminary results on $M$. We denote $\underline{\lambda} = \min(\lambda_i: i\in I)$, $\underline{i}\in I$ such that $\lambda_{\ui} = \ul$ and $\bar{I} = I\backslash\{\ui\}$. We first show that there exists an absolute constant $c_0>0$ such that
\begin{equation}\label{eq:exp_poisson_less_proba}
    \bE[M]\leq c_0 (\underline{\lambda}+1) \bP[M\geq1].
\end{equation}

We first assume that $\underline{\lambda}\leq 1$. On one side, we have
\begin{equation*}
    \bE[M] = \sum_{k\geq 1}\bP[M\geq k] \leq 3\bP[M\geq 1] +  \sum_{k\geq 4}\bP[M\geq k].
\end{equation*} Next, it follows from \eqref{eq:cramer_chernoff_poisson} that
\begin{align*}
 &\sum_{k\geq 4} \bP[M\geq k]  = \sum_{k\geq4} \prod_{i\in \bar{I} } \bP[U_i\geq k] \bP[U_{\ui} \geq k]  \leq \left(\prod_{i\in \bar{I}}\bP[U_i\geq 1] \right) \left(\sum_{k\geq4} e^{-\ul}\left(\frac{e \ul}{k}\right)^k\right)\\
 &\leq \left(\prod_{i\in \bar{I}}\bP[U_i\geq 1] \right) e^{-\ul} \left(\sum_{k\geq4}  \left(\frac{e\ul}{k}\right)^k\right) \leq c_0 \left(\prod_{i\in \bar{I}}\bP[U_i\geq 1] \right) e^{-\ul} \ul.
\end{align*}
On the other side, we use that for all $t\in\bR, 1-e^{-t}\geq t e^{-t}$ to get
\begin{equation}\label{eq:other_side_proba}
    \bP[M\geq1] = \left(\prod_{i\in \bar{I}}\bP[U_i\geq 1] \right) \left(1-e^{-\ul}\right)\geq \left(\prod_{i\in \bar{I}}\bP[U_i\geq 1] \right) \left(\ul e^{-\ul} \right).
\end{equation}As a consequence, we obtain $ \bE[M] \leq (c_0+3) \bP[M\geq1]$ and so \eqref{eq:exp_poisson_less_proba} holds in the case $\ul\leq1$. Let us now assume that $\ul>1$. We have 
\begin{equation*}
   \bE[M] = \sum_{k\geq1} \bP[M\geq k] = \sum_{k=1}^{2e \lceil \ul \rceil} \bP[M\geq k] + \sum_{k\geq 2e \lceil \ul \rceil + 1} \bP[M\geq k] \leq 2e \lceil \ul \rceil \bP[M\geq 1] + \sum_{k\geq 2e \lceil \ul \rceil + 1} \bP[M\geq k].
\end{equation*}Let $k\geq 2e \lceil \ul \rceil + 1$. It follows from \eqref{eq:cramer_chernoff_poisson} that
\begin{align*}
 \bP[M\geq k] = \left(\prod_{i\in\bar{I}} \bP[U_i\geq k]\right) \bP[U_{\ui}\geq k] \leq   \left(\prod_{i\in\bar{I}} \bP[U_i\geq 1]\right) e^{-\ul}\left(\frac{e \ul}{k}\right)^k 
\end{align*}and so 
\begin{align*}
  \sum_{k\geq 2e \lceil \ul \rceil + 1} \bP[M\geq k] \leq \left( \prod_{i\in\bar{I}} \bP[U_i\geq 1]\right) e^{-\ul} \left( \sum_{k\geq 2e \lceil \ul \rceil + 1} \left(\frac{e \ul}{k}\right)^k \right) \leq \left( \prod_{i\in\bar{I}} \bP[U_i\geq 1]\right) e^{-\ul}\ul.
\end{align*}Then, we conclude with \eqref{eq:other_side_proba} that \eqref{eq:exp_poisson_less_proba} also holds in the case $\ul>1$.

Now, we state our second preliminary result. There exists an absolute constant $c_0>0$ such that the following holds: let $U$ be a Poisson distribution with parameter $\lambda>0$ then
\begin{equation}\label{eq:log_poisson}
    \bE\left[\log(U+1)\right]\leq c_0\log(\lambda + e) \bP[U\geq1].
\end{equation}Let us first prove this result when $\lambda\leq1$. We have
\begin{align*}
    &\bE\left[\log(U+1)\right]  = \sum_{k\geq 1} \log(k+1) \bP[U=k]= \log(2) \bP[U=1] + \log(3) \bP[U=2] + \sum_{k\geq3} a_k
\end{align*}where $a_k = \log(k+1) \bP[U=k]$. Since $a_{k+1}/a_k\leq 1/2$ when $k\geq 2$ and $k+1\geq 4\lambda$, we obtain
\begin{equation*}
    \bE\left[\log(U+1)\right] \leq  \log(2) \bP[U=1] + 2 a_3 = \log(2) \bP[U=1] + \log(3) \bP[U=2] + 2\log(4) \bP[U=3] \leq c_0 \bP[U\geq1]
\end{equation*}for some absolute constant $c_0$ and so \eqref{eq:log_poisson} holds when $\lambda\leq1$. We follow a similar argument when $\lambda>1$; we have
\begin{align*}
    &\bE\left[\log(U+1)\right]  = \sum_{k\geq 1} \log(k+1) \bP[U=k]\leq \sum_{k=1}^{4\lceil \lambda \rceil} \log(k+1) \bP[U=k] + \sum_{k\geq 4\lceil \lambda \rceil + 1} \log(k+1) \bP[U=k]\\
    &\leq \log(4\lceil \lambda \rceil +1 ) \bP[U\geq1] + 2 a_{4\lceil \lambda \rceil +1}\leq 3\log(4\lceil \lambda \rceil +2 ) \bP[U\geq1]
\end{align*}and so \eqref{eq:log_poisson} always holds.

Now that we have all the necessary tools, we go back to our initial problem. We first observe that
  \begin{equation*}
        \bE\left[ (m_1+m_2)(m_3 +m_4 ) \ln(U_i+e)\ln(U_j+e)\right] = \sum_{r=1,2;s=3,4}\bE\left[m_r m_s \ln(U_i+1)\ln(U_j+1)\right]
    \end{equation*}The four terms in the right-hand side of the equality above are of the same 'type', the only difference we need to take care of is the relation between $I_s$, $I_r$, $i$ and $j$. Let us consider the hardest situation when $I_s$ and $I_r$ intersect and $i$ and $j$ are different and both in this intersection. The other situations are easier using independence and can be handled using similar technics as below. We are therefore considering now the situation where $i\neq j\in I:=I_r\cap I_s$.  We use the following decomposition and upper bound:
    \begin{align*}
      m_r m_s \ln(U_i+1)\ln(U_j+1)  &= \wedge_{p\in I_r}U_p \wedge_{q\in I_s}U_q \ln(U_i+1)\ln(U_j+1) \\
      &  \leq  \left(\wedge_{p\in I_r\backslash I_s\cup\{i,j\}}U_p\right) \left( \wedge_{q\in I_s\backslash \{i,j\}}U_q\right) \ln(U_i+1)\ln(U_j+1).
    \end{align*}It follows from independence, \eqref{eq:exp_poisson_less_proba} and \eqref{eq:log_poisson} that there exists an absolute constant $c_0>0$ such that
    \begin{align*}
        &\bE\left[m_r m_s \ln(U_i+1)\ln(U_j+1) \right] \leq \bE \left[\wedge_{p\in I_r\backslash I_s\cup\{i,j\}}U_p\right] \bE\left[\wedge_{q\in I_s\backslash \{i,j\}}U_q\right] \bE\left[\ln(U_i+1)\right] \bE\left[\ln(U_j+1)\right]\\
        & \leq c_0 (\bar{\lambda}+1) \bP[U_p\geq1, \forall p\in I_r\backslash I_s\cup\{i,j\} ](\bar{\lambda}+1) \bP[U_q\geq1, \forall q\in I_s\backslash \{i,j\}] \ln(\bar \lambda+1)\bP[U_i\geq1] \ln(\bar\lambda+1)\bP[U_j\geq1]\\
        &\leq c_0 (\bar{\lambda}+1)^2 \ln^2(\bar\lambda+1) \bP[U_p\geq1, \forall p\in I_r\backslash I_s\cup\{i,j\} ] \bP[U_q\geq1, \forall q\in I_s\backslash \{i,j\}] \bP[U_i\geq1]\bP[U_j\geq1]
    \end{align*}and on the other side, using independence, we have
    \begin{equation*}
        \bP[m_r\geq1 \mbox{ and  } m_s\geq1] =  \bP[U_p\geq1, \forall p\in I_r\backslash I_s\cup\{i,j\} ] \bP[U_q\geq1, \forall q\in I_s\backslash \{i,j\}] \bP[U_i\geq1]\bP[U_j\geq1].
    \end{equation*}This shows that 
    \begin{equation*}
      \bE\left[m_r m_s \ln(U_i+1)\ln(U_j+1) \right] \leq   c_0 (\bar{\lambda}+1)^2 \ln^2(\bar\lambda+1) \bP[m_r\geq1 \mbox{ and  } m_s\geq1].
    \end{equation*}

   Finally, the result follows since we have for all  $r\in\{1,2\}$ and $s\in\{3,4\}$ that 
   \begin{equation*}
       \bP[m_r\geq1 \mbox{ and  } m_s\geq1] \leq \bP\left[m_1 + m_2 \geq 1 \mbox{ and  } m_3+m_4 \geq 1\right].
   \end{equation*}The upper bound for $\bE\left[ (m_1+m_2)\ln(U_i+1)\right]$ follows the same strategy as above.

Next, we move to the result on the $p$-th moment of the minimum of independent Poisson variables $M = \wedge_{i\in I}U_i$. Let $p\geq2$ be an integer. We want to show that there exists an absolute constant $c_0>0$ such that $\bE[M^p]\leq c_0 (\underline{\lambda}+p^2) \bP[M\geq1]$. We use the same strategy employed above to handle the first moment. We denote $\underline{\lambda} = \min(\lambda_i: i\in I)$, $\underline{i}\in I$ such that $\lambda_{\ui} = \ul$ and $\bar{I} = I\backslash\{\ui\}$.  We first consider the case $\ul\leq1$. It follows from independence, that
\begin{align*}
  \bE [M^p] &= \sum_{k\geq1} \bP[M^p\geq k] = \sum_{k=1}^{\lceil e^p\rceil}  \bP[M\geq k^{1/p}] +  \sum_{k\geq \lceil e^p\rceil+1}  \bP[M\geq k^{1/p}] \\
  &\leq \lceil e^p\rceil \bP[M\geq 1] + \sum_{k\geq \lceil e^p\rceil+1}  \left(\prod_{i\in \bar{I}} \bP[U_i\geq k^{1/p}]\right)\bP[U_{\ui}\geq k^{1/p}]
\end{align*}It follows from \eqref{eq:cramer_chernoff_poisson} that for all $k\geq \lceil e^p\rceil+1$, 
\begin{align*}
  \bP[U_{\ui}\geq k^{1/p}]\leq e^{-\ul}\left(\frac{e \ul}{k^{1/p}}\right)^{k^{1/p}}.  
\end{align*}Next, using the integral method, one can show that for $k_0 = \lceil e^p\rceil+1$
\begin{equation*}
    \sum_{k\geq k_0}\left(\frac{e \ul}{k^{1/p}}\right)^{k^{1/p}}\leq \left(\frac{e \ul}{k_0^{1/p}}\right)^{k_0^{1/p}}\left(1 + 2p^2 \frac{k_0^{1-1/p}}{\ln(k_0)}\right)\leq c_0 p^2 \ul^e\leq c_0 p^2 \ul
\end{equation*}for some absolute constant $c_0$. As a consequence, we get
\begin{equation*}
   \bE [M^p]\leq  \lceil e^p\rceil \bP[M\geq 1] + c_0 p^2 e^{-\ul} \ul \left(\prod_{i\in \bar{I}} \bP[U_i\geq 1]\right). 
\end{equation*}On the other side, using $1-e^{-1}\geq t e^{-t}, \forall t\in\bR$, we have
\begin{equation*}
    \bP[M\geq1] = \left(\prod_{i\in \bar{I}} \bP[U_i\geq 1]\right) \bP[U_{\ui}\geq1] \geq \left(\prod_{i\in \bar{I}} \bP[U_i\geq 1]\right)  \ul e^{-\ul}
\end{equation*}so that we conclude $\bE [M^p]\leq c_0p^2\bP[M\geq1]$ for some absolute constant $c_0>0$ when $\ul\leq1$. Let us now consider the case $\ul>1$. Using similar arguments as above we obtain for $k_0 = \lceil e\ul\rceil$
\begin{align*}
    \bE[M^p]\leq k_0^p \bP[M\geq1] + \left(\prod_{i\neq \ui} \bP[U_i\geq1] 
\right)\sum_{k\geq k_0} \bP[U_{\ui}\geq k^{1/p}]\leq k_0^p \bP[M\geq1] + c_0 e^{-\ul} \ul \left(\prod_{i\neq \ui} \bP[U_i\geq1] 
\right)
\end{align*}and so we conclude that for some absolute constant $c_0>0$, we have $ \bE[M^p]\leq c_0p^2 \ul^p\bP[M\geq1]$.
\end{proof}

Applying Lemma~\ref{lem:moment_min_Poisson} for $\lambda_\bi = \lambda_\bi(\beta_\star, \theta^\sG)$, $U_\bi = Y_\bi$, $I_1 = \{\bi:s_\xi(\bi)=1\}$, $I_2 = \{\bi:s_\xi(\bi)=-1\}$, $I_3 = \{\bi:s_{\xi^\prime}(\bi)=1\}$ and $I_4 = \{\bi:s_{\xi^\prime}(\bi)=1\}$  in \eqref{eq:main_ccl_L2_loss}, we get
\begin{align*}
    \bE_{\beta_\star}\left[\ell_\xi(Y | X, \beta)\ell_{\xi^\prime}(Y | X, \beta)|X\right] \leq  C_0 \bP[M_\xi+m_\xi\geq 1 \mbox{ and } M_{\xi^\prime}+m_{\xi^\prime}\geq1|X]
\end{align*} where  $C_0$ is defined in \eqref{eq:def_C_1}.
\end{proof}

Next we move to results that will be useful for both consistency and a.n. regarding quantities involving random polyads and/or random edges. We recall that $n=|\cI|$ is the total number of edges and $N = |\Xi|$ is the total number of polyads. We also recall that two polyads can either share no edge in common and have a total number of edges in common in $\{2^r:r\in\{0,1,\cdots,D\}\}$. That is why in the lemmas below we only consider the cases where $q\in\{2^r:r\in\{0,1,\cdots,D\}\}$.

\begin{lemma}\label{lem:Delta}
Let $\xi$ and $\xi^\prime$ be two polyads. Let $c\in\bR^d$. We have
\begin{equation*}
   \left| \bE_{\beta_\star}\left[\inr{\nabla \ell_{\xi}(\beta_\star),c}\inr{\nabla \ell_{\xi^\prime}(\beta_\star),c}|X\right]\right| \leq C_0 \bP_{\beta_\star}[\xi \mbox{ and  } \xi^\prime \mbox{ are active}|X]|\inr{\widetilde{X}_\xi,c}\inr{\widetilde{X}_{\xi^\prime},c}|
\end{equation*}and 
\begin{equation*}
    \left| \bE_{\beta_\star}\left[\bE_{\beta_\star}[\inr{\nabla \ell_\xi(\beta_\star),c}|X, Y_\bi]  \bE_{\beta_\star}[\inr{\nabla \ell_{\xi^\prime}(\beta_\star),c}|X, Y_\bi]|X\right] \right| \leq C_0 \bP_{\beta_\star}[\xi \mbox{ and  } \xi^\prime \mbox{ are active}|X]|\inr{\widetilde{X}_\xi,c}\inr{\widetilde{X}_{\xi^\prime},c}|
\end{equation*}where $C_0$ is defined in \eqref{eq:def_C_1}. 
\end{lemma}

\begin{proof}
It follows from Lemma~\ref{lem:derivatives} that 
\begin{equation*}
    \inr{\nabla \ell_{\xi}(\beta_\star),c} = \left(  \bE_{\beta_\star}\left[ 
m_\xi(Y^\prime) |X,  Y^\prime \in \sO_\xi(Y) \right] - m_\xi(Y) \right) \inr{\widetilde{X}_\xi,c}
\end{equation*}where $Y,Y^\prime$ are iid distributed according to $\bP_{\beta_\star, \theta}$. It follows from Lemma~\ref{lem:m_M_orbits} that, given $Y^\prime \in \sO_\xi(Y) $, we have $m_\xi(Y^\prime) \leq m_\xi(Y) + M_\xi(Y)$. As a consequence, 
\begin{equation}\label{eq:M_xi_bound}
    |\inr{\nabla \ell_{\xi}(\beta_\star),c}| \leq  M_\xi(Y) |\inr{\widetilde{X}_\xi,c}|
\end{equation}and so
\begin{align*}
 \bE_{\beta_\star}\left[|\inr{\nabla \ell_{\xi}(\beta_\star),c}\inr{\nabla \ell_{\xi^\prime}(\beta_\star),c}||X\right] \leq \bE_{\beta_\star}\left[ M_\xi(Y) M_{\xi^\prime}(Y)|X \right]  |\inr{\widetilde{X}_\xi,c}| |\inr{\widetilde{X}_{\xi^\prime},c}
\end{align*}and the result follows from Lemma~\ref{lem:moment_min_Poisson} since 
\begin{align*}
    \bE_{\beta_\star}\left[ M_\xi(Y) M_{\xi^\prime}(Y)|X \right] \leq C_0 \bP_{\beta_\star}[ M_\xi(Y)\geq 1 \mbox{ and  } M_{\xi^\prime}(Y)\geq 1 |X ] \leq C_0 \bP_{\beta_\star}[\xi \mbox{ and  } \xi^\prime \mbox{ are active}|X].
\end{align*}

It follows from Lemma~\ref{lem:derivatives} (note that $m_\xi(Y)$ and $M_\xi(Y)$ play symetric roles and we can express the gradient and the Hessian of $\ell_\xi$ in Lemma~\ref{lem:derivatives} using $M_\xi(Y)$ in place of $m_\xi(Y)$) that we also have 
\begin{equation*}
    \inr{\nabla \ell_{\xi}(\beta_\star),c} = \left( M_\xi(Y)  -  \bE_{\beta_\star}\left[ 
M_\xi(Y^\prime) |X,  Y^\prime \in \sO_\xi(Y) \right] \right) \inr{\widetilde{X}_\xi,c}
\end{equation*}and from Lemma~\ref{lem:m_M_orbits} that, given $Y^\prime \in \sO_\xi(Y) $, we have $M_\xi(Y^\prime) \leq m_\xi(Y) + M_\xi(Y)$. As a consequence, 
\begin{equation}\label{eq:m_xi_bound}
    |\inr{\nabla \ell_{\xi}(\beta_\star),c}| \leq  m_\xi(Y) |\inr{\widetilde{X}_\xi,c}|.
\end{equation}Let $\bi\in \cE(\xi)\cap \cE(\xi^\prime)$. We have $s_\xi(\bi), s_{\xi^\prime}(\bi)\in\{-1, 1\}$ and so $Y_\bi$ is independent of $m_\xi(Y)$ or $M_\xi(Y)$ and of $m_\xi(Y^\prime)$ or $M_\xi(Y^\prime)$. For all cases, one can choose one of the upper bounds in \eqref{eq:M_xi_bound} or \eqref{eq:m_xi_bound} to get an upper bound independent of $Y_\bi$ for both $\xi$ and $\xi^\prime$. For instance, if $Y_\bi$ is independent of $m_\xi(Y)$ and of $M_{\xi^\prime}(Y)$ (the three other cases can be handled using similar arguments) then we get
\begin{align*}
     &\bE_{\beta_\star}\left[\bE_{\beta_\star}[\inr{\nabla \ell_\xi(\beta_\star),c}|X, Y_\bi]  \bE_{\beta_\star}[\inr{\nabla \ell_{\xi^\prime}(\beta_\star),c}|X, Y_\bi]|X\right] \leq \bE_{\beta_\star}\left[ m_\xi(Y) M_{\xi^\prime}(Y)|X \right]  |\inr{\widetilde{X}_\xi,c} \inr{\widetilde{X}_{\xi^\prime},c}|\\
     &\leq C_0 \bP_{\beta_\star}[ m_\xi(Y)\geq 1 \mbox{ and  } M_{\xi^\prime}(Y)\geq 1 |X ] |\inr{\widetilde{X}_\xi,c}\inr{\widetilde{X}_{\xi^\prime},c}| \leq C_0 \bP_{\beta_\star}[\xi \mbox{ and  } \xi^\prime \mbox{ are active}|X]|\inr{\widetilde{X}_\xi,c}\inr{\widetilde{X}_{\xi^\prime},c}|    
 \end{align*} where the last but one inequality follows from Lemma~\ref{lem:moment_min_Poisson}.
\end{proof}



\subsection{A result on the expectation of the empirical $L_3/L_2$ ratio and other probability results}\label{sec:result_on_ratio_L3_L2}
In this section, we provide the tools we used in \eqref{eq:UB_R_3_Y_3} to show that as $n\to \infty$
\begin{equation}\label{eq:aim_end_CLT}
    \frac{1}{n}\sum_{\bi} \bE\left[|R_\bi|^3 |X\right] = o(\sqrt{n}).
\end{equation}Since we can write the mean above as the expectation of the ratio to the power three of the third empirical moment over the second one of independent variables, we provide upper bound for such a quantity in this section.

\begin{proposition}\label{prop:ratio_L3_L2_V0}
Let $((X_{i,n})_{i\in[n]})_n$ be a triangular array of real-valued random variables. We assume that for all $n\in\bN^*$, $(X_{i,n})_{i\in[n]}$ are independent. For all $n\in\bN^*$, we consider for $k=2,3$,
\begin{equation*}
    \hat S_{k,n} = \frac{1}{n}\sum_{i=1}^n |X_{i,n}|^k,   S_{k,n} = \bE \hat S_{k,n} \mbox{ and } V_{k,n} := \frac{1}{n}\sum_{i=1}^n \bV\left(|X_{i,n}|^k\right).
\end{equation*}We assume that 
\begin{equation}\label{eq:weak_assum_for_L3_L2}
    \frac{S_{3,n}+ u_{n}\sqrt{V_{3,n}}}{\left(S_{2,n} - u_n \sqrt{V_{2,n}}\right)^{3/2}} = o\left(\sqrt{n}\right)
\end{equation}for some $u_n\sim a_n/\sqrt{n}$ and $a_n\to +\infty$ when $n\to\infty$. Under this assumption, as $n$ tends to $\infty$, we have
\begin{equation*}
    \bE\left[\frac{\hat S_{3,n}}{\left(\hat  S_{2, n}\right)^{3/2}}\right] = o\left(\sqrt{n}\right).
\end{equation*}   
\end{proposition}
\begin{proof}Let $n\in\bN^*$ and denot $\hat S_{k} = \hat S_{k,n}$, $S_{k} = S_{k,n}$ and $V_k=V_{k,n}$.
We consider the following decomposition:
    \begin{align*}
 \bE\left[\frac{\hat S_{3}}{\left(\hat  S_{2}\right)^{3/2}}\right]  = \bE\left[\frac{\hat S_{3}}{\left(\hat  S_{2}\right)^{3/2}}\left(I_{\Omega_0} + I_{\Omega_0^c}\right)\right]        
    \end{align*}according to the event
    \begin{equation*}
        \Omega_0:=\left\{|\hat S_2 - S_2|\leq u_n \sqrt{V_{2}} \mbox{ and  } |\hat S_3 - S_3|\leq u_n \sqrt{V_{3}}\right\}.
    \end{equation*}It follows from Markov's inequality that $\bP[\Omega_0^c]= o(1)$ and since we have $\hat S_3\leq \sqrt{n} \left(\hat S_{2}\right)^{3/2}$ a.s., we get 
    \begin{equation*}
       \bE\left[\frac{\hat S_{3}}{\left(\hat  S_{2}\right)^{3/2}}I_{\Omega_0^c}\right]=o(\sqrt{n}).
    \end{equation*}On the complementary event, we use Assumption~\ref{eq:weak_assum_for_L3_L2} to get 
        \begin{equation*}
       \bE\left[\frac{\hat S_{3}}{\left(\hat  S_{2}\right)^{3/2}}I_{\Omega_0}\right] \leq \bE\left[\frac{S_3 + |S_3-\hat S_{3}| }{\left(S_2 - |\hat  S_{2}-S_2|\right)^{3/2}}I_{\Omega_0}\right]\leq \frac{S_{3}+ u_{n}\sqrt{V_{3}}}{\left(S_{2} - u_n \sqrt{V_{2}}\right)^{3/2}} = o\left(\sqrt{n}\right).
    \end{equation*}
\end{proof}

Assumption~\ref{eq:weak_assum_for_L3_L2} is a pretty weak assumption. For instance, when the $X_{i,n}$'s have all the same order $3$ and $2$ moments $\mu_3$ and $\mu_2$ as well as the same variance terms $v_3$ and $v_2$ then Assumption~\ref{eq:weak_assum_for_L3_L2} is equivalent to 
\begin{equation}\label{eq:weak_assum_for_L3_L2_2}
    \frac{\mu_3+ a_{n}\sqrt{v_3/n}}{\left(\mu_2 - a_n \sqrt{v_2/n}\right)^{3/2}} = o\left(\sqrt{n}\right)
\end{equation}where $(a_n)_n$ is any sequence such that $a_n\to\infty$, for instance $a_n = \log n$. In that case, \eqref{eq:weak_assum_for_L3_L2_2} is trivially satisfied.

 \begin{lemma}[Example~11.9 from \cite{MR1652247}]
 \label{lem:simple_vaart}
     Let $X, Y, Z$ be random variables. We have $\bE[\bE[X|Y,Z]|Y] = \bE[X|Y]$. 
 \end{lemma}

\begin{lemma}
    \label{lem:simple_cond}
    Let $X, Y, Z$ be random variables such that $Y$ and $Z$ are independent conditionally on $X$. Then for all measurable functions $g$ and $f$, we have 
\begin{itemize}
    \item  $\bE[f(X, Y) g(X, Z)|X] = \bE[f(X, Y)|X] \bE[g(X, Z)|X]$, in other words, conditionally on $X$, $(X, Y)$ and $(X, Z)$ are independent,
    \item $\bE[f(X, Y)|X, Z] = \bE[f(X, Y)|X]$.
\end{itemize}
\end{lemma}
\begin{proof}For the first item, we check the result for all tensor functions $g=g_1\otimes g_2$ and $f = f_1\otimes f_2$ for which the result is trivial. For the second item, it is enough to show that $\bP^{(X, Y)|(X, Z)} = \bP^{(X,Y)|X}$. we have $\bP^{(X, Y)|(X, Z)} = \bP^{X|(X, Z)}\otimes \bP^{Y|(X, Z)}$, $\bP^{X|(X, Z)} = \delta_X$  and for all measurable sets $A, B, C$, by conditional independence, 
\begin{equation*}
    \bP[Y\in A | X\in B, Z\in C] = \frac{\bP[Y\in A, Z\in C| X \in B]}{\bP[X\in B, Z\in C]} = \bP[Y\in A|X\in B]
\end{equation*}so that $\bP^{Y|(X, Z)} = \bP^{Y|X}$. Hence, $\bP^{(X, Y)|(X, Z)} = \delta_{X}\otimes \bP^{Y|X} = \bP^{(X, Y)|X}$.
\end{proof}



\end{document}